\documentclass[lettersize,journal]{IEEEtran}
\usepackage{aliascnt}
\usepackage{cite}
\usepackage{setspace}
\usepackage{wrapfig}
\usepackage{gensymb}

\usepackage{diagbox}

\IEEEoverridecommandlockouts
\newenvironment{talign}
 {\align}
 {\endalign}
\usepackage{multirow}
\usepackage{enumitem}
\usepackage{pifont}

\usepackage{amsmath,amssymb,amsfonts,amsthm}
\usepackage{algorithmic}
\usepackage[ruled,linesnumbered]{algorithm2e}
\usepackage{graphicx}
\usepackage{textcomp}
\usepackage{caption}
\usepackage {epstopdf}
\usepackage[hidelinks]{hyperref}
\usepackage{float}
\usepackage[utf8]{inputenc}
\usepackage{tabularx}
\usepackage{booktabs}
\usepackage{cases}
\usepackage{xcolor}
\usepackage{authblk}
\usepackage{cleveref}
\usepackage{subfigure}
\usepackage{subcaption}
\usepackage{bm}
\def\BibTeX{{\rm B\kern-.05em{\sc i\kern-.025em b}\kern-.08em
    T\kern-.1667em\lower.7ex\hbox{E}\kern-.125emX}}

\usepackage{graphicx}
\usepackage{mathrsfs}

\usepackage{bm}
\usepackage{hyperref,bookmark}
\PassOptionsToPackage{bookmarks={false}}{hyperref}

\DeclareRobustCommand{\vect}[1]{\bm{#1}}
\newtheoremstyle{slanted}
{0em plus 0em minus 0em}
  {0em plus 0em minus 0em}
  {\em}
  {}
  {\bfseries}
  {.}
  { }
  {}

\theoremstyle{slanted}

\theoremstyle{slanted}
\newtheorem{definition}{Definition}

\theoremstyle{slanted}
\newtheorem{theorem}{Theorem}

\theoremstyle{slanted}

\theoremstyle{slanted}

\theoremstyle{slanted}

\theoremstyle{slanted}
\newtheorem{lemma}{Lemma}

\usepackage{enumitem}

\title{Parameter Training Efficiency Aware Resource Allocation for AIGC in Space-Air-Ground Integrated Networks}
\author{\IEEEauthorblockN{Liangxin Qian, \emph{Graduate Student Member, IEEE,} Peiyuan Si, \emph{Graduate Student Member, IEEE,} Jun Zhao, \emph{Member, IEEE}, and 
Kwok-Yan Lam, \emph{Senior Member, IEEE} }\vspace{-10pt}\thanks{Liangxin Qian, Peiyuan Si, Jun Zhao, and Kwok-Yan Lam are with the College of Computing and Data Science (CCDS) at Nanyang Technological University, Singapore. Email: qian0080@e.ntu.edu.sg, peiyuan001@e.ntu.edu.sg, junzhao@ntu.edu.sg, kwokyan.lam@ntu,edu,sg.
}
}
\begin{document}
\maketitle
\begin{abstract}
With the evolution of artificial intelligence-generated content (AIGC) techniques and the development of space-air-ground integrated networks (SAGIN), there will be a growing opportunity to enhance mobile user experiences with customized AIGC applications. This is enabled by combining parameter-efficient fine-tuning (PEFT) with mobile edge computing. In this paper, we formulate the optimization problem of maximizing the parameter training efficiency of the SAGIN system over wireless networks under limited resource constraints. We propose the \underline{P}arameter training efficiency \underline{A}ware \underline{R}esource \underline{A}llocation (PARA) technique to jointly optimize user association, data offloading, and communication and computational resource allocation. Detailed derivations are presented to solve this difficult sum of ratios problem based on quadratically constrained quadratic programming (QCQP), semidefinite programming (SDP), graph theory, and fractional programming (FP) techniques. Our proposed PARA technique is effective in finding a stationary point of this non-convex problem. The simulation results demonstrate that the proposed PARA method outperforms other baselines.
\end{abstract}
\section{Introduction}
\subsection{Background}

Parameter-efficient fine-tuning (PEFT) techniques, e.g., low-rank adaptation (LoRA), model pruning, and knowledge distillation, have emerged as essential tools for adapting large artificial intelligence (AI) models to specific downstream tasks with significantly reduced computational overhead \cite{hu2021lora,cao2023comprehensive,gou2021knowledge,ding2023parameter,wu2023ai}. These methods enable faster and more efficient training by fine-tuning only a small portion of the model's parameters, making them particularly suitable for edge scenarios where resources are limited \cite{lu2023llama}. In parallel, artificial intelligence-generated content (AIGC) systems have seen rapid adoption across domains, e.g., personalized assistants. To support these applications, frequent and adaptive fine-tuning on user data is needed, pushing the demand for scalable and distributed model update mechanisms.


However, traditional infrastructure faces serious limitations in enabling such distributed PEFT workloads. Terrestrial edge servers face coverage and capacity limitations, whereas centralized cloud servers suffer from high latency and energy overhead. These limitations render them impractical for latency-sensitive and resource-constrained scenarios.

To address these limitations, space-air-ground integrated networks (SAGIN) have gained attention as a promising hierarchical architecture for global communication and computation~\cite{cui2022space,ray2022review,bai2020relay}. By incorporating satellites, aerial platforms, and terrestrial base stations, SAGIN offers an infrastructure capable of wide-area coverage and flexible resource coordination. This makes SAGIN an attractive candidate for supporting on-demand PEFT tasks across diverse user locations \cite{du2023age}.

\subsection{Motivation and challenges}
Given the fact that much of the research has focused on resource allocation for terrestrial networks, there is a need to explore the potential performance improvements of high-altitude and satellite platforms for communications and computing missions. 
Terrestrial edge servers can provide a fundamental infrastructure, but they may not always be sufficient to handle all computing tasks efficiently, especially for AI model training. While offloading residual training tasks to remote cloud servers is an option, relying solely on cloud computing introduces significant bandwidth consumption and transmission delays. Instead, a hierarchical computing framework that integrates aerial and satellite platforms enables a more balanced and scalable approach to resource management. Even though AI training is not a real-time application, minimizing energy consumption and optimizing computing resources across multiple layers is crucial for efficiency. Aerial and satellite computing layers can serve as intermediate nodes, reducing cloud dependency and distributing the workload dynamically based on available resources.

The main difficulty in rolling out PEFT services across SAGIN is dealing with the limited resources these networks have \cite{liu2018space}. Resources like the amount of data the network can handle, the computing power of aerial platforms, and the energy available for sending data are all limited and can change based on actual requirements \cite{zhou2020deep}. SAGIN is made up of different levels, from mobile users to ground, air, and satellite servers, each with its own set of rules for how things work, making the task of managing resources even more complex \cite{zhang2023ai}. It is also necessary to find a suitable balance between system delay and energy consumption and make sure AI content creation tools are trained properly.

To tackle these issues, our study proposes a novel method to manage resources that are specially made to improve how efficiently parameters are trained in SAGIN. This approach focuses on user association, partial offloading, transmit power, bandwidth, and computation resource optimization. Our goal is to make all levels of SAGIN work better together, enhancing support for services that fine-tune models with minimal resources.

\subsection{Studied problem}
Our research focuses on enhancing parameter training efficiency (PTE), i.e., $\frac{\text{training parameter sizes}}{\text{delay}+\text{energy} }$ to be introduced in Section \ref{sec.optimization_problem}, in SAGIN through a mobile edge computing mechanism. This approach involves a sequential distribution of training tasks, starting from users and moving through terrestrial, aerial, and finally satellite servers. Each server in this hierarchy is responsible for processing a specific portion of the user's workload, with the task progressively offloaded from one level to the next. Initially, one user's task is sent to a terrestrial server, which undertakes a part of the training parameters, leaving the remainder for subsequent levels. The task is then further divided, with subsequent portions handled by aerial and satellite servers, ensuring the entire workload is distributed across the network's levels. This approach is similar to how heat spreads out in thermodynamics, using the closeness of each layer in the network to reduce how far data needs to travel and make better use of resources. The problem we are tackling is how to improve this detailed task offloading and resource allocation strategy. This includes figuring out how to best offload work and manage resources among users on the ground, servers in the air, and satellites in space, to make the data processing more efficient throughout the SAGIN system.
\subsection{Main contributions}
Our main contributions are as follows:
\begin{itemize}
    \item[$\bullet$]We propose a novel metric to quantify parameter training efficiency across SAGIN. This metric provides a foundational basis for evaluating and optimizing the training process, setting a new standard for assessing performance in complex network environments. To the best of our knowledge, there is no research on the proposed metric.
    \item[$\bullet$]To address the challenging non-convex sum of ratios optimization problem in Section \ref{sec.optimization_problem}, we propose the \underline{P}arameter training efficiency-\underline{A}ware \underline{R}esource \underline{A}llocation (PARA) technique. This method is distinct from the approaches discussed in Section \ref{sec.related_work} (refer to papers \cite{jong2012efficient,shen2018fractional,10368052,10.1145/3565287.3610271}) as it enables the joint optimization of user association, offloading ratio, and communication and computation resource allocations. Note that no approximation method is used but a novel fractional programming (FP) technique to conduct the joint optimization of bandwidth, transmit power, and computation resources across all four levels of the SAGIN architecture, including users, terrestrial servers, aerial servers, and satellite servers.
    \item[$\bullet$]We have given rigorous proofs of the proposed PARA technique by utilizing quadratically constrained quadratic programming (QCQP), semidefinite programming (SDP), graph theory, and FP techniques. This rigorous theoretical framework ensures the reliability and effectiveness of the PARA technique.
    \item[$\bullet$]Through comprehensive simulation results, we demonstrate the PARA technique's capability to reliably find a stationary point for the proposed optimization problem in Section \ref{sec.optimization_problem}. These results showcase its superiority in enhancing PTE within the SAGIN framework.
\end{itemize}
This paper is organized as follows: Section \ref{sec.related_work} reviews the related work. The system model is detailed in Section \ref{sec.system_model}. The formulation of the optimization problem is presented in Section \ref{sec.optimization_problem}. Our proposed solution, the PARA algorithm, is introduced in Section \ref{sec.proposed_para_algorithm}, followed by an analysis of its complexity in Section \ref{sec.complexity_analysis}. Simulation results demonstrating the effectiveness of our approach are discussed in Section \ref{sec.numerical_results}. Finally, the paper concludes with Section \ref{sec.conclusion}.

\begin{table*}[t]
\centering
\caption{Comparison of our paper with related papers.}
\tiny 
\resizebox{\textwidth}{!}{%
\begin{tabular}{lcccccccccc}
\toprule
\textbf{Paper} & \textbf{Optimization Objective} & \textbf{SAGIN} & \textbf{AIGC} & \textbf{Delay} & \textbf{Energy} & \textbf{Utility} & \textbf{Transmit Power} & \textbf{Bandwidth} & \textbf{Computation Resource} \\
\midrule
Nguyen \textit{et al.} \cite{nguyen2023integrated} & Energy & \checkmark & $\times$ & \checkmark & \checkmark & $\times$ & $\times$ & \checkmark & \checkmark \\
Wang \textit{et al.} \cite{wang2021incorporating} & Storage & \checkmark & $\times$ & $\times$ & $\times$ & $\times$ & $\times$ & $\times$ & \checkmark \\
Zhang \textit{et al.} \cite{zhang2022distributed} & Caching & \checkmark & $\times$ & \checkmark & $\times$ & \checkmark & $\times$ & $\times$ & $\times$ \\
Qin \textit{et al.} \cite{qin2023energy} & Energy efficiency & \checkmark & $\times$ & $\times$ & \checkmark & $\times$ & \checkmark & \checkmark & $\times$ \\
Cao \textit{et al.} \cite{cao2021edge} & Resource scheduling & \checkmark & $\times$ & \checkmark & $\times$ & $\times$ & $\times$ & $\times$ & \checkmark \\
Shen \textit{et al.} \cite{shen2018fractional} & Power/beamforming & $\times$ & $\times$ & $\times$ & \checkmark & $\times$ & \checkmark & $\times$ & $\times$ \\
Zhao \textit{et al.} \cite{10368052} & Utility-cost ratio & $\times$ & $\times$ & \checkmark & \checkmark & \checkmark & \checkmark & \checkmark & \checkmark \\
Zhao \textit{et al.} \cite{10.1145/3565287.3610271} & Utility-energy efficiency & $\times$ & $\times$ & $\times$ & \checkmark & \checkmark & \checkmark & \checkmark & $\times$ \\
Huang \textit{et al.} \cite{huang2019reconfigurable} & Energy efficiency & $\times$ & $\times$ & $\times$ & \checkmark & $\times$ & \checkmark & $\times$ & $\times$ \\
Ju \textit{et al.} \cite{ju2013throughput} & Sum throughput & $\times$ & $\times$ & $\times$ & $\times$ & $\times$ & $\times$ & $\times$ & $\times$ \\
Liu \textit{et al.} \cite{liu2025optimizing} & Quality of experience & $\times$ & \checkmark & \checkmark & $\times$ & \checkmark & $\times$ & \checkmark & $\times$ \\
Zhou \textit{et al.} \cite{zhou2020computation, sun2019joint} & Computation efficiency & $\times$ & $\times$ & $\times$ & \checkmark & $\times$ & \checkmark & $\times$ & \checkmark \\
This paper & Parameter training efficiency & \checkmark & \checkmark & \checkmark & \checkmark & \checkmark & \checkmark & \checkmark & \checkmark \\
\bottomrule
\end{tabular}%
}
\label{tab:related work comparison}
\end{table*}

\section{Related Work}\label{sec.related_work}
In this section, we discuss the related work on the research of efficiency metrics, resource allocation in SAGIN, and novel fractional programming techniques.
\subsection{Efficiency metric research}
In wireless communication systems, performance has traditionally been measured using well-established metrics, e.g., spectral efficiency, energy efficiency, cost efficiency, and throughput efficiency. Each of these focuses on a specific aspect of transmission performance. Spectral efficiency quantifies how effectively bandwidth is used, typically measured in bits per second per Hz \cite{foschini1999simplified}. Energy efficiency focuses on the amount of data transmitted per unit of energy consumed, an important factor in battery-constrained devices \cite{huang2019reconfigurable}. Cost efficiency evaluates the financial overhead of transmitting data \cite{tombaz2011energy}, aiming to balance performance with affordability. Throughput efficiency reflects the system’s ability to handle traffic density in space and time \cite{ju2013throughput}. Another metric closely related to our work is computation efficiency, which is generally defined as the ratio of total computed bits to energy consumption \cite{zhou2020computation, sun2019joint}.

While these metrics are effective for transmission-focused applications, they fall short in addressing computation-intensive tasks, e.g., AIGC model training, over edge networks. In such scenarios, the system must not only transmit data but also allocate computation resources efficiently across multiple layers of a hierarchical infrastructure (e.g., terrestrial, aerial, and satellite servers).

\subsubsection{Differences between parameter training efficiency and other efficiency metrics}
In this research, we introduce PTE as a novel and application-driven metric specifically tailored for AI model training in hierarchical edge computing environments such as SAGIN. PTE is defined as the ratio of the total number of trainable parameters (i.e., the actual computation workload) to the aggregate system cost, which comprises both end-to-end communication delay and total energy consumption during the training process. This definition directly integrates the core aspects of modern edge intelligence: computational intensity, network responsiveness, and power efficiency.

Traditional metrics, e.g., spectral efficiency (bits per second per Hz), energy efficiency (bits per Joule), computation efficiency (bits per Joule), or throughput efficiency (bits per second) are well-suited for evaluating data transmission in conventional wireless networks. However, they are insufficient for quantifying the effectiveness of distributed training tasks, especially those involving partial model offloading, multi-hop computation, and layer-wise coordination across devices with heterogeneous resources. In contrast, PTE is explicitly designed to capture the computational utility achieved per unit of incurred system cost, making it more relevant for evaluating how efficiently AI workloads are processed and learned across the network.

Moreover, PTE reflects the critical trade-offs between training scope and system constraints. Increasing the number of trainable parameters (e.g., fine-tuning more adapter modules) may yield better model performance, but it also escalates communication overhead and power usage. In resource-constrained environments, where UAVs, satellites, and edge devices have limited battery and bandwidth, PTE serves as a meaningful optimization target: it quantifies how much training benefit can be obtained relative to the delay and energy budget.

From a system design perspective, PTE offers a holistic decision-making tool for resource allocation and task scheduling. For example, it can guide whether to keep training tasks locally at the user, offload them to terrestrial or aerial servers, or further push them to satellites, depending on real-time delay and energy trade-offs. This makes PTE not only a performance metric but also an optimization objective that bridges communication and computation in AI-driven wireless systems.


\subsection{Resource allocation research in SAGIN}
In addressing the difficulty of resource allocation within SAGIN, recent studies have introduced a spectrum of innovative solutions tailored to enhance network performance across varying dimensions. The authors in \cite{nguyen2023integrated} delve into the computation offloading challenges in hybrid edge-cloud-based SAGIN, focusing on an integrated approach to optimize computation offloading, UAV trajectory, user scheduling, and radio resource allocation, aiming to minimize energy consumption while adhering to delay constraints. This approach leverages alternating optimization and the successive convex approximation method to address the non-convex optimization problem, demonstrating significant efficiency gains over conventional methods. Another research in \cite{wang2021incorporating} introduces a distributed deep reinforcement learning algorithm for managing SAGIN's limited storage resources, showcasing notable improvements in resource allocation revenue and user request acceptance rate. Furthermore, to cater to the IoE scenario, another study in \cite{zhang2022distributed} advocates for wireless edge caching within SAGIN, optimized through distributed DRL to minimize transmission delays and alleviate task offloading pressures. In the context of the industrial power IoT, a NOMA-enabled SAGIN-IPIoT model is proposed in \cite{qin2023energy} to enhance system throughput and energy efficiency by optimizing subchannel and terminal power through a mixed-integer nonlinear programming approach. To bridge the communication gap within the Internet of Vehicles, a novel SAGIN-IoV edge-cloud architecture is proposed in \cite{cao2021edge}, leveraging software-defined network (SDN) and network function virtualization (NFV) to optimize resource scheduling, highlighting the pivotal role of advanced computational models in refining resource allocation and ensuring seamless connectivity within SAGIN environments. 

In addition to these foundational studies, several recent works have addressed more specific challenges in SAGIN resource management. Wei \textit{et al}. \cite{wei2024energy} focused on energy-efficient caching and user selection in emergency scenarios, using UAV-assisted caching to extend coverage while preserving satellite energy. Jia \textit{et al}. \cite{jia2025nfv} investigated service recovery under resource failures by designing an NFV-enabled SFC recovery model, allowing tasks to dynamically reallocate resources when failures occur. He \textit{et al}. \cite{he2024online} developed a two-timescale graph model and online algorithm to optimize joint data offloading and power control under the dynamic topology and energy constraints of SAGIN, achieving near-optimal results with low computation cost.

\subsubsection{Detailed comparison to our proposed PARA technique}
Our proposed PARA technique, aimed at optimizing parameter training efficiency within SAGIN, introduces a comprehensive optimization framework that is distinct from existing research in both its objectives and methodologies. The authors in \cite{nguyen2023integrated} aim to minimize the weighted energy consumption of systems and use three alternative optimization steps to optimize user scheduling, partial offloading control, computation resource, and bandwidth allocation. Note that in terms of resource allocation, only computation resource and bandwidth are considered in \cite{nguyen2023integrated}, without the consideration of transmit power. For the optimization of bandwidth allocation, the successive convex approximation (SCA) method is used to find an upper bound of the Shannon formula. Compared to this method, we not only consider the joint optimization of transmit power and bandwidth allocation simultaneously but also use a novel fractional programming technique without any approximation. Besides, three levels’ resources (i.e., terrestrial, aerial, satellite servers) are also not considered in \cite{nguyen2023integrated}. What’s more, the authors ignore the joint optimization of user scheduling and partial offloading control simultaneously in our paper. 

Unlike previous works \cite{nguyen2023integrated,wang2021incorporating,zhang2022distributed,qin2023energy,cao2021edge} that focus on specific aspects such as computation offloading, storage management, or throughput enhancement, we ambitiously target a holistic improvement by jointly optimizing user association, partial offloading, transmit power, bandwidth, and computation resources across the SAGIN's user, terrestrial, aerial, and satellite layers. Utilizing optimization methods such as QCQP, SDP, graph theory, and FP techniques, without resorting to approximations, the proposed PARA algorithm uniquely addresses challenges in the simultaneous optimization of bandwidth, transmit power, and computation resource allocation.
\subsection{Novel fractional programming technique research}
For the sum of ratio optimization problem $\sum_{i=1}^{N}\frac{A_n(\bm{x})}{B_n(\bm{x})}$, the authors in \cite{jong2012efficient} proposed to transform it into parametric convex optimization problem to obtain a global optimum for maximization or minimization problem. However, the technique proposed in \cite{jong2012efficient} can't be applied for the optimization problem $C(\bm{x})+\sum_{i=1}^{N}\frac{A_n(\bm{x})}{B_n(\bm{x})}$. To address this issue, \cite{shen2018fractional} replaced $\frac{A_n(\bm{x})}{B_n(\bm{x})}$ as $2y_n\sqrt{A_n(\bm{x})}-y_n^2B_n(\bm{x})$. Based on the proof of \cite{shen2018fractional}, the maximization of $C(\bm{x})+\sum_{i=1}^{N}\frac{A_n(\bm{x})}{B_n(\bm{x})}$ is same as that of $C(\bm{x})+\sum_{i=1}^{N}2y_n\sqrt{A_n(\bm{x})}-y_n^2B_n(\bm{x})$, where $y_n$ is iteratively updated to $\frac{A_n(\bm{x})}{B_n(\bm{x})}$. In an alternative manner of optimizing $\bm{y}$ and $\bm{x}$, a stationary point can be obtained. Note that the minimization problem of $C(\bm{x})+\sum_{i=1}^{N}\frac{A_n(\bm{x})}{B_n(\bm{x})}$ can't be solved by this technique. Zhao~\emph{et~al.}~\cite{10368052} proposed to replace $\frac{A_n(\bm{x})}{B_n(\bm{x})}$ as $A_n^2(\bm{x})y_n+\frac{1}{4B_n^2(\bm{x})y_n}$, where $y_n = \frac{1}{2A_n(\bm{x})B_n(\bm{x})}$, and they successfully solve the minimization problem by proofs. In \cite{10368052}, optimizing just one ratio $\frac{\text{utility}}{\text{cost}}$ is also considered. The optimization of $C(\bm{x})+\sum_{i=1}^{N}\frac{A_n(\bm{x})}{B_n(\bm{x})}$ is used in the ``cost'' term. To tackle the sum of ratio optimization problem $\sum\frac{\text{utility}}{\text{cost}}$, the authors in \cite{10.1145/3565287.3610271} propose a parametric optimization technique to obtain the global optimum. However, only energy efficiency is considered in \cite{10.1145/3565287.3610271}. Therefore, we consider extending these technologies to solve more sum of ratio optimization problems like the optimization problem we propose in Section \ref{sec.optimization_problem} and applying those techniques proposed in \cite{10368052} and \cite{10.1145/3565287.3610271} to solve more problems like this class.

\section{System Model}\label{sec.system_model}
In this section, we present the SAGIN system, edge training mechanism, and analysis of system costs.

\begin{table}[t]
\caption{Important Notations.}\label{table.notation}
\centering
\setlength{\tabcolsep}{0pt}
\small
\renewcommand{\arraystretch}{1.1}
\begin{tabular}{ m{7em}  m{6.5cm}}
    \toprule                                   
    \textbf{Notation}  & \textbf{Description} \\
    \midrule
    $\mathcal{N}$ & Set of all users ($n \in\{1,...,N\}$)\\ 
    $\mathcal{M}^{(i)}$ & Set of servers ($m^{(i)} \in\{1,...,M^{(i)}\}$, $i \in \{t, a, s\}$)\\
    $\varphi_n^{(i)}$ & Offloading ratio of user $n$ ($i \in \{u, t, a, s\}$)\\
    $d_n$ & Training parameter size of user $n$\\
    $d_n^{(t)}$ & Number of input tokens for user $n$\\
    $d_n^{(l)}$ & Intermediate results and labeling data size of the local dataset\\
    $\gamma_n^{(u)}$ & Computing speed partition ratio of user $n$\\
    $f_n$ & Maximum computing speed of user $n$\\
    $\gamma_{n,m}^{(i)}$ & Computing speed partition ratio of server $m^{(i)}$ for user $n$'s task ($i \in \{t, a, s\}$)\\
    $f_{m_{i}}$ & Maximum computing speed of server $m^{(i)}$ ($i \in \{t, a, s\}$)\\
    $\phi_{n,m}^{(i)}$ & Communication bandwidth partition ratios of server $m^{(i)}$ for user $n$'s task ($i \in \{t, a, s\}$)\\
    $b_{m_{i}}$ & Maximum allocated bandwidth of server $m^{(i)}$ ($i \in \{t, a, s\}$)\\
    $\rho_n^{(u)}$ & Transmission power partition ratio of user $n$\\
    $p_n$ & Maximum transmission power of user $n$\\
    $\rho_{n,m_i}$ & Transmission power partition ratios of server $m^{(i)}$ for user $n$'s task ($i \in \{t, a\}$)\\
    $p_{n,m}^{(i)}$ & Maximum transmission power of server $m^{(i)}$ ($i \in \{t, a\}$)\\
    $x_{n,m}^{(i)}$ & Association of user $n$ and selected server $m^{(i)}$ ($i \in \{t, a, s\}$)\\
    $\psi_n^{(u)}$ & Auxiliary variable: $\psi_n^{(u)} = \frac{c_n^{(u)}\varphi_n^{(u)}d_n}{cost_n^{(u)}}$\\
    $\psi_{n,m}^{(i)}$ & Auxiliary variable: $\psi_{n,m}^{(i)} = \frac{c_{n,m}^{(i)}\varphi_{n,m}^{(i)}d_n}{cost_{n,m}^{(i)}}$ ($i \in \{t, a, s\}$)\\
    $c_n^{(u)}$ & PTE preference of user $n$\\
    $c_{n,m}^{(i)}$ & PTE preference of server $m^{(i)}$ for user $n$'s task\\
    $\alpha_n^{(u)}$ & Auxiliary variable: $\alpha_n^{(u)} = \frac{1}{cost_n^{(u)}}$\\
    $\alpha_{n,m}^{(i)}$ & Auxiliary variable: $\alpha_{n,m}^{(i)} = \frac{1}{cost_{n,m}^{(i)}}$ ($i \in \{t, a, s\}$)\\
    $\varrho_{n,m}^{(i)}$ & Auxiliary variable ($i \in \{t, a, s\}$, see Lemma 3 for details)\\
    \bottomrule
\end{tabular}
\vspace{-10pt}
\end{table}
\subsection{SAGIN networks}
We consider a SAGIN network consisting of $N$ mobile users and $M$ servers (including $M^{(t)}$ \underline{t}errestrial servers, $M^{(a)}$ \underline{a}erial servers and $M^{(s)}$ \underline{s}atellite servers, i.e., $M = M^{(t)} + M^{(a)} + M^{(s)}$) in Fig. \ref{fig.system_model}. We use $m$ to indicate the $m$-th server, where $m \in \mathcal{M} := \{1,2,\cdots,M^{(t)} + M^{(a)} + M^{(s)}\}$, and $n$ represents the $n$-th mobile user, where $n\in \mathcal{N}:= \{1,2,\cdots,N\}$.
\begin{figure}[htbp]
\vspace{-0.1cm}
\centering
\includegraphics[width=0.49\textwidth]{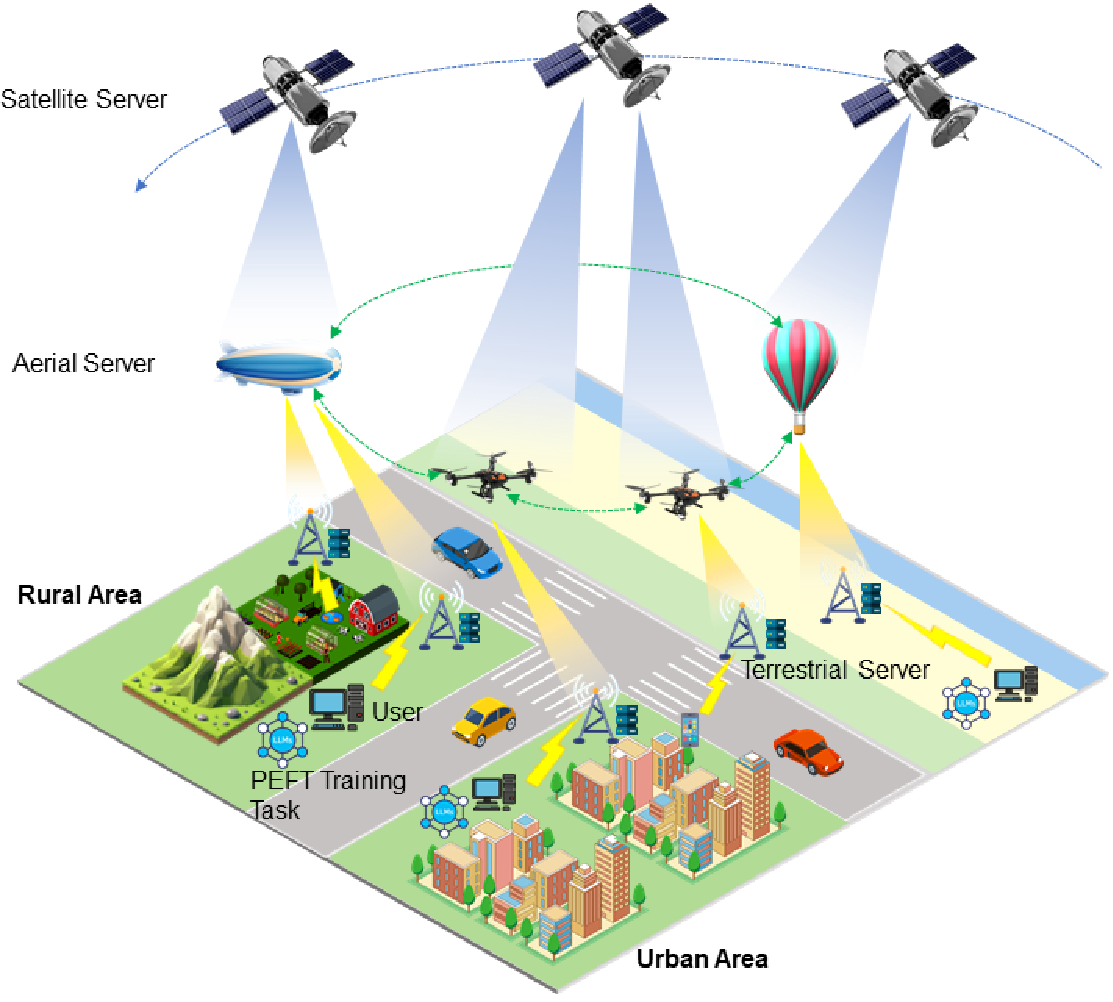}
\caption{Schematic diagram of the SAGIN system for PEFT training tasks.}\vspace{-5pt}
\label{fig.system_model}
\end{figure}
\subsubsection{Terrestrial networks}
In terrestrial networks, there are $M^{(t)}$ terrestrial edge servers. $m^{(t)}$ is used to represent the $m$-th terrestrial edge server, where $m^{(t)}\in \mathcal{M}^{(t)}:= \{1,2,\cdots,M^{(t)}\}$. Each terrestrial base station has a specific communication coverage area. 
For each terrestrial edge server, GPU resources are provided for computing mobile users' services. 

\subsubsection{Aerial networks}
In the aerial networks, there are $M^{(a)}$ aerial edge servers, which are made up of several high-altitude platforms (HAPs), i.e., drones, hot balloons, and airships, and $m^{(a)}$ is the index of the \mbox{$m^{(a)}$-th} aerial edge servers, where $m^{(a)}\in \mathcal{M}^{(a)}:= \{1,2,\cdots,M^{(a)}\}$. Those aerial vehicles are typically at altitudes of around 17 to 22 kilometers, e.g., Project Loon from Google. Due to their high altitudes, HAPs can cover a much larger area compared to terrestrial base stations, making them ideal for providing connectivity in remote or rural areas. Each aerial edge server is equipped with enough computing resources. They can provide computing services for users within their coverage area.

\subsubsection{Satellite networks}
In the satellite networks, there are $M^{(s)}$ low earth orbit (LEO) satellites, and $m^{(s)}$ is used to denote the $m^{(s)}$-th LEO satellite, where $m^{(s)}\in \mathcal{M}^{(s)}:= \{1,2,\cdots,M^{(s)}\}$. 
In the assignment of mobile devices, directly connecting mobile users to the LEO satellite would introduce many unstable factors, including high user mobility, limited and intermittent satellite visibility, and the long propagation distance through the atmosphere. In our SAGIN architecture, mobile users instead access nearby terrestrial edge servers via short-range wireless links, and the terrestrial and aerial servers then communicate with the LEO satellite through more stable links with less interference \cite{wang2024high}. Therefore, we leverage the LEO satellite to assist the aerial edge server in executing PEFT training tasks for users.

\subsection{PEFT edge training model}
In this section, we discuss the PEFT edge training scheme. In the proposed SAGIN system, PEFT follows a hierarchical offloading framework where training tasks are distributed across multiple layers of edge servers. The base model, such as large language model meta AI (LLaMA), and the upcoming PEFT technique are known to both the user and edge servers in the pre-communication. As shown in Fig. \ref{fig.peft}, take the serial adapter \cite{houlsby2019parameter}, one additive PEFT method, as an example. Additive PEFT enhances the model architecture by integrating new trainable modules or parameters \cite{han2024parameter}. In the serial adapter method, each Transformer block is augmented with two adapter modules—one after the self-attention layer and another after the feedforward layer.
\begin{figure}[t]
\vspace{-0.1cm}
\centering
\includegraphics[width=0.4\textwidth]{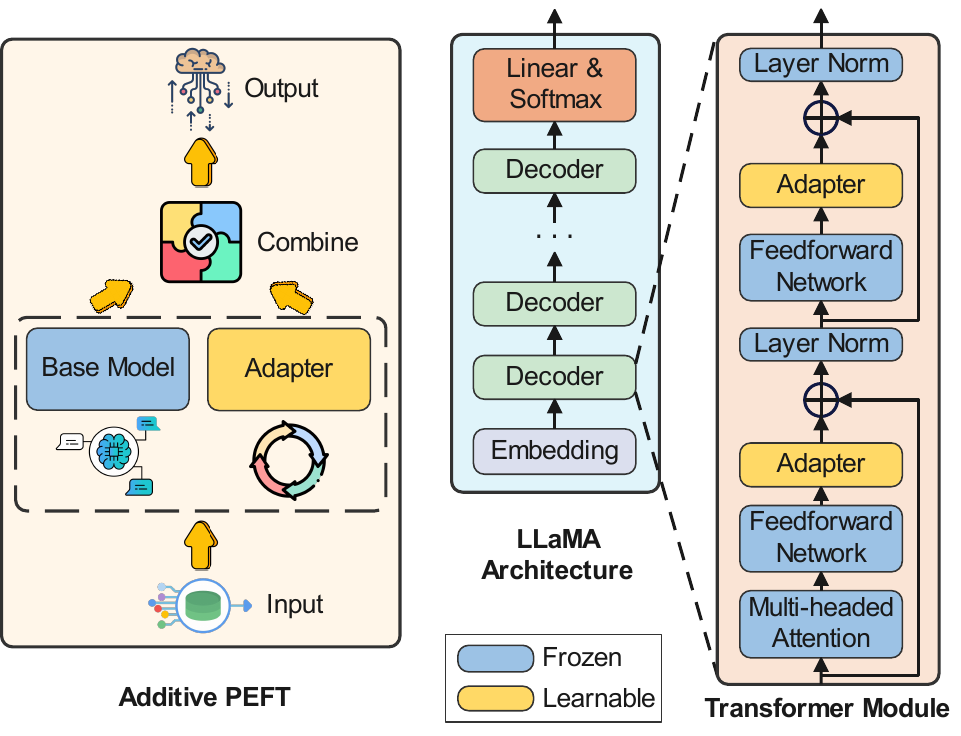}
\caption{Illustration of additive PEFT techniques in the LLaMA architecture with trainable adapters. Blue/yellow colors indicate frozen/trainable parameters respectively.}\vspace{-5pt}
\label{fig.peft}
\end{figure}
\begin{figure}[t]
\vspace{-0.1cm}
\centering
\includegraphics[width=0.48\textwidth]{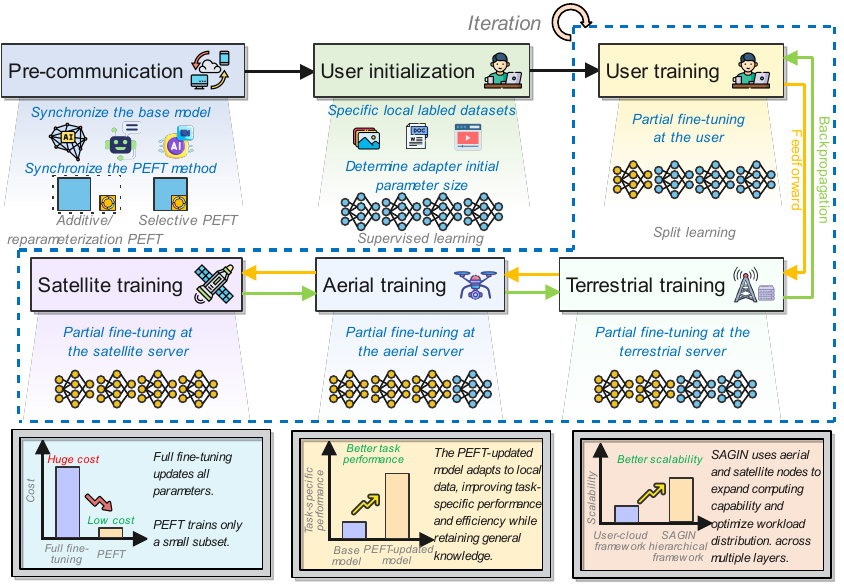}
\caption{Illustration of the PEFT edge training procedure in the SAGIN system. Blue represents frozen, and yellow means trainable.}\vspace{-5pt}
\label{fig.procedure}
\end{figure}

We give the illustration of the PEFT edge training procedure in the SAGIN system in Fig. \ref{fig.procedure}. The user determines the number of Transformer blocks to fine-tune and the size of the adapter parameters, which define the total trainable parameters. The fine-tuning process starts with the user training a fraction of the adapter parameters on its local input data while keeping the base model frozen. The remaining adapter parameters, along with intermediate results and labeled output data, are transmitted to the terrestrial server for further processing. The terrestrial server trains a portion of the remaining parameters and offloads the rest to the aerial server, which continues the process before finally passing the workload to the satellite server for completion. The user or each server processes adapter modules sequentially in integer multiples of Transformer blocks before transferring the remaining workload to the next level. This hierarchical training strategy reduces communication overhead, optimizes energy consumption, and effectively distributes computational resources across the SAGIN architecture, enabling efficient fine-tuning of large AI models over wireless networks.

\subsubsection{Work offloading ratio decisions}
In the PEFT training task, the number of user and server training transformer modules is an integer. But for simplicity, let us first consider the case where they are continuous numbers. For the case where the offloading ratio is a discrete value, the solution of the continuous value can be obtained first and then approximated to the discrete value. We consider using continuous variables $\varphi_n^{(u)}$, $\varphi_n^{(t)}$, $\varphi_n^{(a)}$, and \mbox{$\varphi_n^{(s)} \in [0,1]$} to indicate the work offloading ratios of the local \underline{u}ser, \underline{t}errestrial server, \underline{a}erial server, and \underline{s}atellite server, respectively. The sum of $\varphi_n^{(u)}$, $\varphi_n^{(t)}$, $\varphi_n^{(a)}$, and $\varphi_n^{(s)}$ is one. We define $\bm{\varphi^{(u)}}:=[\varphi^{(u)}_n]|_{n \in \mathcal{N}}$, $\bm{\varphi^{(i)}}:=[\varphi^{(i)}_{n,m}]|_{n \in \mathcal{N}, m \in \mathcal{M}^{(i)}}$, for $i \in \{t,a,s\}$, and $\bm{\varphi}:=\{\bm{\varphi^{(u)}},\bm{\varphi^{(t)}},\bm{\varphi^{(a)}},\bm{\varphi^{(s)}}\}$.

\subsubsection{PEFT training offloading data}
We assume the input tokens' number of user $n$ is $d_n^{(t)}$, the training parameter size of user $n$ is $d_n$, the whole data size of PEFT training of the user $n$ is $\omega_b d_n$, and the intermediate results and labeling data size of the local dataset is $d_n^{(l)}$, where $\omega_b$ is bits used to represent each parameter. As we discussed in the previous section, the user and servers share the knowledge of the base model and upcoming PEFT method. However, the whole parameter size of the adapter, which can be determined by the user, is not shared with the servers due to privacy requirements. For simplicity, let's assume that the user and the servers are trained on the same foundation model structure (e.g., LLaMA), use the serial adapter as the PEFT technique, and each of them processes adapter modules sequentially in integer multiples of Transformer blocks. Therefore, the intermediate results are of the same size.

Based on these assumptions, the size of data communicated (if there is) between the user $n$ and the connected terrestrial server $m^{(t)}$ is $(1-\varphi_n^{(u)})\omega_b d_n + d_n^{(l)}$, where $(1-\varphi_n^{(u)})d_n$ is the remaining neural networks modules excluding those that user $n$ having trained locally. Similarly, the sizes of data communicated between the terrestrial server $m^{(t)}$ and the aerial server $m^{(a)}$, and between the aerial server $m^{(a)}$ and the LEO server $m^{(s)}$ are $(1-\varphi_n^{(u)}-\varphi_n^{(t)})\omega_b d_n + d_n^{(l)}$ and $(1-\varphi_n^{(u)}-\varphi_n^{(t)}-\varphi_n^{(a)})\omega_b d_n + d_n^{(l)}$, respectively. We assume that users' tasks offloaded to the selected servers are processed instantly. Therefore, queue backlogs or queueing delays are ignored \cite{gao2024cost}.

\subsubsection{Edge training mechanism}
The uplink work offloading is studied in this SAGIN network. In the context of cooperation layer training in the SAGIN networks, training task communication at all levels between users, ground servers, aerial servers, and satellite servers includes the remaining layer parameters, intermediate results of the previous level, and labeling data. Initially, the user offloads work to a terrestrial server. The remaining training work is transmitted from the terrestrial server to an aerial server and finally from the aerial server to a satellite server. This layered offloading strategy is adopted primarily due to the reduced communication distance between successive network layers, specifically between the aerial and satellite servers, compared to the longer distance between the user and the satellite server directly.

Let's consider an edge training scheme where the work $d_n$ of user $n$ is first pushed to the terrestrial server, which gets $(1-\varphi_n^{(u)}) d_n$ training parameters but only completes $\varphi_n^{(t)} d_n$ of that while freezing the remaining modules' parameters (i.e., $(1-\varphi_n^{(u)}-\varphi_n^{(t)})d_n$). The terrestrial server then pushes some of its work to the aerial server, which finishes $\varphi_n^{(a)}d_n$ parameters. Finally, the aerial server pushes some of its tasks to the satellite server, which trains $\varphi_n^{(s)}d_n$ parameters. Note that $\varphi_n^{(u)} + \varphi_n^{(t)} + \varphi_n^{(a)} + \varphi_n^{(s)} = 1$. Users' tasks are gradually distributed to servers at all levels, a process similar to diffusion in thermodynamics.


\subsubsection{Computing speed partition ratio decisions}
We consider using continuous variables $\gamma_{n}^{(u)}$, $\gamma_{n,m}^{(t)}$, $\gamma_{n,m}^{(a)}$, and \mbox{$\gamma_{n,m}^{(s)} \in [0,1]$} to indicate the computing speed partition ratios of the user, terrestrial server, aerial server, and satellite server, respectively. Thus, the actually used computing speeds of user $n$, terrestrial server $m^{(t)}$, aerial server $m^{(a)}$, satellite server $m^{(s)}$ are $\gamma_{n}^{(u)}f_n$, $\gamma_{n,m}^{(t)}f_{m_t}$, $\gamma_{n,m}^{(a)}f_{m_a}$, $\gamma_{n,m}^{(s)}f_{m_s}$, respectively. $f_n$ (unit: FLOPs) is the maximum computing speed of user $n$. $f_{m_t}$, $f_{m_a}$, and $f_{m_s}$ are the maximum computing speeds of terrestrial server $m^{(t)}$, aerial server $m^{(a)}$, and satellite server $m^{(s)}$, respectively. 
We define $\bm{\gamma^{(u)}}:=[\gamma^{(u)}_n]|_{n \in \mathcal{N}}$, $\bm{\gamma^{(i)}}:=[\gamma^{(i)}_{n,m}]|_{n \in \mathcal{N}, m \in \mathcal{M}^{(i)}}$, for $i \in \{t,a,s\}$, and $\bm{\gamma}:=\{\bm{\gamma^{(u)}},\bm{\gamma^{(t)}},\bm{\gamma^{(a)}},\bm{\gamma^{(s)}}\}$.

\subsubsection{Communication bandwidth partition ratio decisions}
We consider using continuous variables $\phi_{n,m}^{(t)}$, $\phi_{n,m}^{(a)}$, and \mbox{$\phi_{n,m}^{(s)} \in [0,1]$} to indicate the communication bandwidth partition ratios of the terrestrial server $m^{(t)}$, aerial server $m^{(a)}$, and satellite server $m^{(s)}$, respectively. The actual allocated bandwidth from terrestrial server $m^{(t)}$, aerial server $m^{(a)}$, and satellite server $m^{(s)}$ to the user or server of the previous level are $\phi_{n,m}^{(t)}b_{m_t}$, $\phi_{n,m}^{(a)}b_{m_a}$, and $\phi_{n,m}^{(s)}b_{m_s}$, respectively. $b_{m_t}$, $b_{m_a}$, and $b_{m_s}$ are the maximum bandwidth that terrestrial server $m^{(t)}$, aerial server $m^{(a)}$, and satellite server $m^{(s)}$ can allocate to the user or server of the previous level. We define $\bm{\phi^{(i)}}:=[\phi^{(i)}_{n,m}]|_{n \in \mathcal{N}, m \in \mathcal{M}^{(i)}}$, for $i \in\{t,a,s\}$, and $\bm{\phi}:=\{\bm{\phi^{(t)}},\bm{\phi^{(a)}},\bm{\phi^{(s)}}\}$.

\subsubsection{Transmission power partition ratio decisions}
We consider using continuous variables $\rho_{n}^{(u)}$, $\rho_{n,m}^{(t)}$, and \mbox{$\rho_{n,m}^{(a)} \in [0,1]$} to indicate the transmission power partition ratios of the user $n$, terrestrial server $m^{(t)}$, and aerial server $m^{(a)}$, respectively. Therefore, the actual used transmission power of user $n$, terrestrial server $m^{(t)}$, and aerial server $m^{(a)}$ are $\rho_{n}^{(u)}p_n$, $\rho_{n,m}^{(t)}p_{m_t}$, and $\rho_{n,m}^{(a)}p_{m_a}$, respectively. $p_n$, $p_{m_t}$, and $p_{m_a}$ is the maximum transmission power of user $n$, terrestrial server $m^{(t)}$, and aerial server $m^{(a)}$, respectively. We define $\bm{\rho^{(u)}}:=[\rho^{(u)}_n]|_{n \in \mathcal{N}}$, $\bm{\rho^{(i)}}:=[\rho^{(i)}_{n,m}]|_{n \in \mathcal{N}, m \in \mathcal{M}^{(i)}}$, for $i \in \{t,a\}$, and $\bm{\rho}:=\{\bm{\rho^{(u)}},\bm{\rho^{(t)}},\bm{\rho^{(a)}}\}$.
\subsubsection{User association decisions}
We use binary variables $x_{n,m}^{(t)}$, $x_{n,m}^{(a)}$, and $x_{n,m}^{(s)} \in \{0,1\}$ to indicate the connection of $\{$user $n$, terrestrial server $m^{(t)}\}$, $\{$terrestrial server $m^{(t)}$, aerial server $m^{(a)}\}$, and $\{$aerial server $m^{(a)}$, satellite server $m^{(s)}\}$ for processing the offloading training work from user $n$, respectively. These connection decisions are made based on some metrics that we want to optimize. We define $\bm{x^{(i)}}:=[x^{(i)}_{n,m}]|_{n \in \mathcal{N}, m \in \mathcal{M}^{(i)}}$, for $i \in \{t,a,s\}$, and $\bm{x}:=\{\bm{x^{(t)}},\bm{x^{(a)}},\bm{x^{(s)}}\}$.

\subsubsection{Wireless communication model}
The up-link channel is considered in the wireless communication between users and one terrestrial base station or aerial/satellite server. We employ frequency division multiple access (FDMA) to ensure non-interfering communication between users and servers.
For the mobile user $n$ and the terrestrial server $m^{(t)}$ and the transmission rate is
\begin{talign}
    r_{n,m_t} = \phi_{n,m}^{(t)}b_{m_t} \log_2 (1 + \frac{\rho_n^{(u)} p_n g_{n,m_t}}{\sigma^2 \phi_{n,m}^{(t)} b_{m_t}}),
\end{talign}
where $\phi_{n,m}^{(t)}b_{m_t}$ is the allocated bandwidth between user $n$ and the server $m^{(t)}$, $\rho_n^{(u)}p_n$ is the transmit power of user $n$, $g_{n,m_t}$ is the channel gain between the user $n$ and the server $m^{(t)}$, and $\sigma^2$ is the noise power spectral density. Similarly, we can define the transmission rate between terrestrial server $m^{(t)}$ and aerial server $m^{(a)}$, and that between aerial server $m^{(a)}$ and satellite server $m^{(s)}$ as $r_{m_t,m_a}$ and $r_{m_a,m_s}$, respectively.

\subsection{System cost}
\subsubsection{Time consumption}
In this section, we discuss time consumption in the PEFT edge training system. For user $n$, he needs to train $\varphi_n^{(u)}d_n$ parameters. Based on the training time estimation given in \cite{narayanan2021efficient}, the training time is 
\begin{talign}
T^{(up)}_n = \frac{e_n t_n \varphi_n^{(u)}d_n}{\gamma_n^{(u)}f_n},
\end{talign}
where $e_n$ is the training epochs of user $n$, $t_n = \omega_f d_n^{(t)}$, $\omega_f$ is the ratio that transforms each training parameter into FLOPs, and $\omega_f$ is eight (FLOPs/(parameters$\cdot$tokens)) in \cite{narayanan2021efficient}.
Then, user $n$ transmits the remaining parameters, intermediate results, and labeling data $(1-\varphi_n^{(u)})d_n + d_n^{(l)}$ to the connected terrestrial server $m^{(t)}$. Data transmission time in this phase is 
\begin{talign}
T^{(ut)}_{n,m} = \frac{x_{n,m}^{(t)}[\omega_b(1-\varphi_n^{(u)})d_n+ d_n^{(l)}]}{r_{n,m_t}}.
\end{talign}
For terrestrial server $m^{(t)}$, after receiving the remaining training parameters, the intermediate results, and labeling data from user $n$, it would allocate some computing resources for processing the partial $\varphi_n^{(t)}$ of training task for user $n$. The training time of terrestrial server $m^{(t)}$ is 
\begin{talign}
T^{(tp)}_{n,m} = \frac{e_{m_t}x_{n,m}^{(t)}t_n\varphi_n^{(t)}d_n }{\gamma_{n,m}^{(t)}f_{m_t}},
\end{talign}
where $e_{m_t}$ is the training epochs of terrestrial server $m^{(t)}$.
Then, terrestrial server $m^{(t)}$ transmits $\omega_b(1-\varphi_n^{(u)}-\varphi_n^{(t)})d_n + d_n^{(l)}$ data to the connected aerial edge server $m^{(a)}$. Data transmission time within this period is
\begin{talign}
T^{(tt)}_{n,m} = \frac{x_{n,m}^{(a)}[\omega_b(1-\varphi_n^{(u)}-\varphi_n^{(t)})d_n + d_n^{(l)}]}{r_{m_t,m_a}}.
\end{talign}
Aerial server $m^{(a)}$ processes $\varphi_n^{(a)}$ part of those parameters once received and the training time consumed is given as
\begin{talign}
T^{(ap)}_{n,m} = \frac{e_{m_a}x_{n,m}^{(a)}t_n\varphi_n^{(a)}d_n}{\gamma_{n,m}^{(a)}f_{m_a}},
\end{talign}
where $e_{m_a}$ is the training epochs of aerial server $m^{(a)}$.
After finishing partial training tasks, aerial server $m^{(a)}$ would send the remaining data to connected (if any) satellite server $m^{(s)}$ and related data transmission time is 
\begin{talign}
T^{(at)}_{n,m} = \frac{x_{n,m}^{(s)}[\omega_b(1-\varphi_n^{(u)}-\varphi_n^{(t)}-\varphi_n^{(a)})d_n + d_n^{(l)}]}{r_{m_a,m_s}}.
\end{talign}
For the satellite server $m^{(s)}$, it completes the remaining training tasks, and related training time can be given as
\begin{talign}
T^{(sp)}_{n,m} = \frac{e_{m_s}x_{n,m}^{(s)}t_n\varphi_n^{(s)}d_n}{\gamma_{n,m}^{(s)}f_{m_s}},
\end{talign}
where $e_{m_s}$ is the training epochs of satellite server $m^{(s)}$.
\subsubsection{Energy consumption}
Next, we analyze the energy consumption in the PEFT edge training system. For the user $n$, it finishes local training work and the energy consumption is 
\begin{talign}
E^{(up)}_n = e_n\kappa_n t_n\varphi_n^{(u)}d_n(\gamma_n^{(u)}f_n)^2,
\end{talign}
where $\kappa_n$ is the GPU computational efficiency of user $n$, indicating how power consumption increases with faster computing speeds. The wireless transmission energy of user $n$ is
\begin{talign}
E^{(ut)}_{n,m} = \rho_n^{(u)}p_n \frac{x_{n,m}^{(t)}[\omega_b(1-\varphi_n^{(u)})d_n+ d_n^{(l)}]}{r_{n,m_t}}.
\end{talign}
For the terrestrial server $m^{(t)}$, it trains $\varphi_n^{(t)}d_n$ parameters and energy consumption of this training phase is 
\begin{talign}
E^{(tp)}_{n,m} = x_{n,m}^{(t)}e_{m_t}\kappa_{m_t}t_n\varphi_n^{(t)}d_n(\gamma_{n,m}^{(t)}f_{m_t})^2,
\end{talign}
where $\kappa_{m_t}$ is the GPU computational efficiency of terrestrial server $m_t$. 
The transmission energy consumption at the terrestrial server level is 
\begin{talign}
E^{(tt)}_{n,m} = \rho_{n,m}^{(t)}p_{m_t}\frac{x_{n,m}^{(a)}[\omega_b(1-\varphi_n^{(u)}-\varphi_n^{(t)})d_n + d_n^{(l)}]}{r_{m_t,m_a}}.
\end{talign}
For the aerial server $m^{(a)}$, its training energy consumption is
\begin{talign}
E^{(ap)}_{n,m} = x_{n,m}^{(a)}e_{m_a}\kappa_{m_a}t_n\varphi_n^{(a)}d_n(\gamma_{n,m}^{(a)}f_{m_a})^2,
\end{talign}
where $\kappa_{m_a}$ is the GPU computational efficiency of aerial server $m_a$. The transmission energy consumption of aerial server $m^{(a)}$ is 
\begin{talign}
E^{(at)}_{n,m} = \rho_{n,m}^{(a)}p_{m_a}\frac{x_{n,m}^{(s)}[\omega_b(1-\varphi_n^{(u)}-\varphi_n^{(t)}-\varphi_n^{(a)})d_n+ d_n^{(l)}]}{r_{m_a,m_s}}.
\end{talign}
The training energy consumption of satellite server $m^{(s)}$ is 
\begin{talign}
E^{(sp)}_{n,m} = x_{n,m}^{(s)}e_{m_s}\kappa_{m_s}t_n\varphi_n^{(s)}d_n(\gamma_{n,m}^{(s)}f_{m_s})^2.
\end{talign}

\section{Studied Optimization Problem} \label{sec.optimization_problem}
In this section, we present the studied optimization problem and we first define parameter training efficiency as follows:
\begin{definition}[Parameter Training Efficiency]
    Parameter training efficiency (PTE) $:= \frac{\text{training parameter size}}{\text{delay + energy}}$. The parameter training consumption of each level includes the parameter training consumption of the level and the wireless data consumption of the upper level. For example, we assume $\varphi_n^{(s)}d_n$ parameters are trained in the satellite server $m^{(s)}$. The cost of $\varphi_n^{(s)}d_n$ parameters includes the delay and energy consumption of training them and the data transmission delay and energy consumption from the aerial server $m^{(a)}$.
\end{definition}
Based on the definition of PTE, we give the PTEs of user $n$, terrestrial server $m^{(t)}$, aerial server $m^{(a)}$, and satellite server $m^{(s)}$ as follows:
\begin{talign}
&\frac{\varphi_n^{(u)}d_n}{cost_n^{(u)}} = \frac{\varphi_n^{(u)}d_n}{\omega_t T_n^{(up)} + \omega_e E_n^{(up)}},\\
&\frac{\varphi_n^{(t)}d_n}{cost_{n,m}^{(t)}} =\frac{\varphi_n^{(t)}d_n}{\omega_t (T_{n,m}^{(ut)} + T_{n,m}^{(tp)}) + \omega_e (E_{n,m}^{(ut)} + E_{n,m}^{(tp)})},\\
&\frac{\varphi_n^{(a)}d_n}{cost_{n,m}^{(a)}} =\frac{\varphi_n^{(a)}d_n}{\omega_t (T_{n,m}^{(tt)} + T_{n,m}^{(ap)}) + \omega_e (E_{n,m}^{(tt)} + E_{n,m}^{(ap)})},\\
&\frac{\varphi_n^{(s)}d_n}{cost_{n,m}^{(s)}} =\frac{\varphi_n^{(s)}d_n}{\omega_t (T_{n,m}^{(at)} + T_{n,m}^{(sp)}) + \omega_e (E_{n,m}^{(at)} + E_{n,m}^{(sp)})},
\end{talign}
where $\omega_t$ and $\omega_e$ are weight parameters of delay and energy terms, respectively.
Our studied optimization problem is to maximize the sum of PTE at all levels in SAGIN and it is given as follows: 
\begin{subequations}\label{prob1}
\begin{talign}
\mathbb{P}_{1}:&\max\limits_{\bm{x},\bm{\varphi},\bm{\gamma},\bm{\phi},\bm{\rho}}  \sum\limits_{n \in \mathcal{N}}\sum\limits_{m \in \mathcal{M}}\Big(\frac{c_{n,m}^{(t)}\varphi_n^{(t)}d_n}{cost_{n,m}^{(t)}} + \frac{c_{n,m}^{(a)}\varphi_n^{(a)}d_n}{cost_{n,m}^{(a)}} \nonumber \\
&\hspace{70pt}+ \frac{c_{n,m}^{(s)}\varphi_n^{(s)}d_n}{cost_{n,m}^{(s)}}\Big)+\sum\limits_{n \in \mathcal{N}}\frac{c_n^{(u)}\varphi_n^{(u)}d_n}{cost_n^{(u)}}
\tag{\ref{prob1}}\\
\text{s.t.} \quad 
& x_{n,m}^{(i)} \in \{0,1\}, \forall n \in \mathcal{N}, m \in \mathcal{M}^{(i)}, i \in \{t,a,s\},\label{x_range_constr} \\
& \sum_{m \in \mathcal{M}^{(i)}} x_{n,m}^{(i)} = 1, \forall n \in \mathcal{N}, i \in \{t,a,s\},\label{x_sum_constr} \\
& \varphi_n^{(i)} \in [0,1], \forall n \in \mathcal{N}, i \in \{u,t,a,s\},\label{varphi_range_constr}\\
&\varphi_n^{(u)} + \varphi_n^{(t)} + \varphi_n^{(a)} + \varphi_n^{(s)} = 1, \forall n \in \mathcal{N},\label{varphi_sum_constr2}\\
&\phi_{n,m}^{(i)} \in [0,1], \forall n \in \mathcal{N}, m \in \mathcal{M}^{(i)}, i \in \{t,a,s\},\label{phi_range_constr}\\
&\sum_{n\in \mathcal{N}} x_{n,m}^{(i)} \phi_{n,m}^{(i)} \leq 1, \forall m \in \mathcal{M}^{(i)}, i \in \{t,a,s\},\label{phi_sum_constr}\\
&\gamma_n^{(u)},\gamma_{n,m}^{(i)} \in [0,1], \forall n \in \mathcal{N}, m \in \mathcal{M}^{(i)}, i \in \{t,a,s\},\label{gamma_range_constr}\\
&\sum_{n\in \mathcal{N}} x_{n,m}^{(i)} \gamma_{n,m}^{(i)} \leq 1, \forall m \in \mathcal{M}^{(i)},i \in \{t,a,s\},\label{gamma_sum_constr}\\
&\rho_n^{(u)},\rho_{n,m}^{(i)} \in [0,1], \forall n \in \mathcal{N}, m \in \mathcal{M}^{(i)}, i \in \{t,a\},\label{rho_range_constr}\\
&\sum_{n\in \mathcal{N}} x_{n,m}^{(i)} \rho_{n,m}^{(i)} \leq 1, \forall m \in \mathcal{M}^{(i)}, i \in \{t,a\}\label{rho_sum_constr},
\end{talign}
\end{subequations}
where $c_n^{(u)}$ is the PTE preference of user $n$, $c_{n,m}^{(t)}$, $c_{n,m}^{(a)}$, and $c_{n,m}^{(s)}$ are the PTE preferences of terrestrial server $m^{(t)}$, aerial server $m^{(a)}$, and satellite server $m^{(s)}$ for user $n$'s training tasks. Constraint (\ref{x_range_constr}) means server $m$ is chosen for user $n$'s tasks or not. Constraint (\ref{x_sum_constr}) indicates that there is one and only one terrestrial/aerial/satellite server chosen for user $n$'s tasks. 
Constraint (\ref{varphi_range_constr}) represents the offloading ratio of training tasks from user $n$ to the selected server at each hierarchical layer. Constraint (\ref{varphi_sum_constr2}) ensures that the total training task of the user $n$ is partitioned among all four layers. Constraint (\ref{phi_range_constr}) represents the bandwidth allocation ratio for the user $n$ from the server $m$ at the layer $i$.
Constraint (\ref{phi_sum_constr}) represents the allocated bandwidth limit of each server. 
Constraint (\ref{gamma_range_constr}) represents the computing resource allocation ratio for user $n$'s task, either locally or from the server.
Constraint (\ref{gamma_sum_constr}) is the allocated computing resource limit of each server. 
Constraint (\ref{rho_range_constr}) is the transmission power allocation ratio for data sent from the user $n$ or from the server $m$ at the layer $i$.
Constraint (\ref{rho_sum_constr}) is the allocated transmission power limit of each server. For the sake of simplicity, we first ignore the mobility of HAPs and satellites. The discussion of the mobility-aware PEFT edge training under SAGIN networks will be presented in Section \ref{sec.mobility}.

\section{Proposed PARA algorithm for SAGIN}\label{sec.proposed_para_algorithm}
In this section, we present our proposed PARA algorithm to solve the very difficult sum of ratios Problem $\mathbb{P}_1$. This problem is known to be non-convex and NP-hard, posing significant difficulties in obtaining a tractable solution. To address this, we leverage an alternating optimization (AO) framework, which iteratively refines different subsets of variables. The theoretical foundation of our approach is summarized in the following theorem:
\begin{theorem}\label{theorem_solvep1}
    Problem $\mathbb{P}_{1}$ can be transformed into a solvable problem if we alternatively optimize $[\bm{x},\bm{\varphi}]$ and $[\bm{\phi},\bm{\rho},\bm{\gamma}]$.
\end{theorem}
\begin{proof}
    \textbf{Theorem \ref{theorem_solvep1}} is proven by the following \textbf{Lemma \ref{lemma_p1top2}}, \textbf{Lemma \ref{lemma_p2top3}}, \textbf{Theorem \ref{theorem_subproblem1}} in Section \ref{sec_subproblem1}, and \textbf{Theorem \ref{theorem_subproblem2}} in Section \ref{sec_subproblem2}.
\end{proof}
\subsection{Pre-transformations for Problem \texorpdfstring{$\mathbb{P}_1$}{}}
Problem $\mathbb{P}_1$ is a sum of multiple ratios problem, where each ratio is a complex non-convex or concave expression. Direct analysis is very difficult. Therefore, we consider adding the following auxiliary variables to simplify the Problem $\mathbb{P}_1$. 
\begin{lemma}\label{lemma_p1top2}
Define new auxiliary variables $\psi_n^{(u)}$, $\psi_{n,m}^{(t)}$, $\psi_{n,m}^{(a)}$, $\psi_{n,m}^{(s)}$, $T_n^{(u)}$, $T_{n,m}^{(t)}$, $T_{n,m}^{(s)}$, and $T_{n,m}^{(s)}$. Let $\bm{\psi^{(u)}}:=[\psi_n^{(u)}]|_{n\in\mathcal{N}}$, $\bm{\psi^{(i)}}:=[\psi_{n,m}^{(i)}]|_{n\in\mathcal{N},m\in\mathcal{M}}$, $i \in \{t,a,s\}$, $\bm{T^{(u)}}:=[T_n^{(u)}]|_{n\in\mathcal{N}}$, $\bm{T^{(i)}}:=[T_{n,m}^{(i)}]|_{n\in\mathcal{N},m\in\mathcal{M}}$, $i \in \{t,a,s\}$, $\bm{T}:=\{\bm{T^{(u)}},\bm{T^{(t)}},\bm{T^{(a)}},\bm{T^{(s)}}\}$, and $\bm{\psi}:=\{\bm{\psi^{(u)}},\bm{\psi^{(t)}},\bm{\psi^{(a)}},\bm{\psi^{(s)}}\}$. Besides, we define functions $\varpi^{(u)}_n$, $\varpi^{(t)}_{n,m}$, $\varpi^{(a)}_{n,m}$, and $\varpi^{(s)}_{n,m}$ as follows:
\begin{talign}
&\varpi^{(u)}_n(\varphi_n^{(u)},\gamma_n^{(u)},\psi_n^{(u)},T_n^{(u)})\nonumber \\
&:=\omega_t T^{(u)} + \omega_e e_n\kappa_n t_n\varphi_n^{(u)}d_n(\gamma_n^{(u)}f_n)^2 - \frac{c_n^{(u)}\varphi_n^{(u)}d_n}{\psi_n^{(u)}}, \label{defvarpiu}\\
&\varpi^{(t)}_{n,m}(x_{n,m}^{(t)},\varphi_n^{(u)},\varphi_n^{(t)},\phi_{n,m}^{(t)},\rho_n^{(u)},\gamma_{n,m}^{(t)},\psi_{n,m}^{(t)},T_{n,m}^{(t)})\nonumber \\
&:=\omega_t T^{(t)} + \omega_e (\rho_n^{(u)}p_n \frac{x_{n,m}^{(t)}[\omega_b(1-\varphi_n^{(u)})d_n+d_n^{(l)}]}{r_{n,m_t}} \nonumber \\
&+ x_{n,m}^{(t)}\kappa_{m_t} e_{m_t}t_n \varphi_n^{(t)}d_n(\gamma_{n,m}^{(t)}f_{m_t})^2) - \frac{c_{n,m}^{(t)}\varphi_n^{(t)}d_n}{\psi_{n,m}^{(t)}}, \label{defvarpit}\\
&\varpi^{(a)}_{n,m}(x_{n,m}^{(a)},\varphi_n^{(u)},\varphi_n^{(t)},\varphi_n^{(a)},\phi_{n,m}^{(a)},\rho_{n,m}^{(t)},\gamma_{n,m}^{(a)},\psi_{n,m}^{(a)},T_{n,m}^{(a)})\nonumber \\
&:=\omega_t T^{(a)} + \omega_e (\rho_{n,m}^{(t)}p_{m_t}\frac{x_{n,m}^{(a)}[\omega_b(1-\varphi_n^{(u)}-\varphi_n^{(t)})d_n+d_n^{(l)}]}{r_{m_t,m_a}} \nonumber \\
&+x_{n,m}^{(a)} e_{m_a}\kappa_{m_a}t_n \varphi_n^{(a)}d_n(\gamma_{n,m}^{(a)}f_{m_a})^2) - \frac{c_{n,m}^{(a)}\varphi_n^{(a)}d_n}{\psi_{n,m}^{(a)}}, \label{defvarpia}\\
&\varpi^{(s)}_{n,m}(x_{n,m}^{(s)},\!\varphi_n^{(u)}\!,\!\varphi_n^{(t)}\!,\!\varphi_n^{(a)}\!\!,\!\varphi_n^{(s)}\!,\!\phi_{n,m}^{(s)},\rho_{n,m}^{(a)},\gamma_{n,m}^{(s)},\psi_{n,m}^{(s)},T_{n,m}^{(s)})\nonumber \\
&:=\omega_t T^{(s)} + \omega_e (\rho_{n,m}^{(a)}p_{m_a}\frac{x_{n,m}^{(s)}[\omega_b(1-\varphi_n^{(u)}-\varphi_n^{(t)}-\varphi_n^{(a)})d_n+d_n^{(l)}]}{r_{m_a,m_s}} \nonumber \\
&+x_{n,m}^{(s)}e_{m_s}\kappa_{m_s}t_n\varphi_n^{(s)}d_n(\gamma_{n,m}^{(s)}f_{m_t})^2) - \frac{c_{n,m}^{(s)}\varphi_n^{(s)}d_n}{\psi_{n,m}^{(s)}}. \label{defvarpis}
\end{talign}
Then the sum of ratios Problem $\mathbb{P}_{1}$ can be transformed into a summation Problem $\mathbb{P}_{2}$:
\begin{subequations}\label{prob2}
\begin{talign}
\mathbb{P}_{2}:&\max\limits_{\bm{x},\bm{\varphi},\bm{\gamma},\bm{\phi},\bm{\rho},\bm{\psi},\bm{T}}  \sum\limits_{n \in \mathcal{N}}\!\sum\limits_{m \in \mathcal{M}} \!(\psi_{n,m}^{(t)} \!+\! \psi_{n,m}^{(a)} \!+\! \psi_{n,m}^{(s)}) \!+\! \sum\limits_{n \in \mathcal{N}}\psi_n^{(u)}
\tag{\ref{prob2}}\\
\text{s.t.} \quad 
& (\text{\ref{x_range_constr}})\text{-}(\text{\ref{rho_sum_constr}}) \nonumber\\
& \varpi^{(u)}_n \leq 0, \forall n \in \mathcal{N},\label{psi_u_constr}\\
& \varpi^{(i)}_{n,m} \leq 0, \forall n \in \mathcal{N}, \forall m \in \mathcal{M}, i \in \{t,a,s\}\label{psi_tas_constr}\\
&T_n^{(up)} \leq T_n^{(u)}, \forall n \in \mathcal{N}, \forall m \in \mathcal{M},\label{Tu_constr}\\
&T_{n,m}^{(ut)} + T_{n,m}^{(tp)} \leq T_{n,m}^{(t)}, \forall n \in \mathcal{N}, \forall m \in \mathcal{M},\label{Tt_constr} \\
&T_{n,m}^{(tt)} + T_{n,m}^{(ap)} \leq T_{n,m}^{(a)}, \forall n \in \mathcal{N}, \forall m \in \mathcal{M},\label{Ta_constr}\\
&T_{n,m}^{(at)} + T_{n,m}^{(sp)} \leq T_{n,m}^{(s)}, \forall n \in \mathcal{N}, \forall m \in \mathcal{M}. \label{Ts_constr}
\end{talign}
\end{subequations}
\end{lemma}
\begin{proof}
    Refer to Appendix \ref{append_lemma_p1top2}.
\end{proof}
According to \textbf{Lemma \ref{lemma_p1top2}}, we can transform the sum of ratios Problem $\mathbb{P}_{1}$ to a summation Problem $\mathbb{P}_{2}$ by adding the extra auxiliary variables $\psi_n^{(u)}$, $\psi_{n,m}^{(t)}$, $\psi_{n,m}^{(a)}$, $\psi_{n,m}^{(s)}$, $T_n^{(u)}$, $T_{n,m}^{(t)}$, $T_{n,m}^{(a)}$, $T_{n,m}^{(s)}$, and new funtions $\varpi^{(u)}_n$, $\varpi^{(t)}_{n,m}$, $\varpi^{(a)}_{n,m}$, and $\varpi^{(s)}_{n,m}$. Thanks to $\psi_n^{(u)}$, $\psi_{n,m}^{(t)}$, $\psi_{n,m}^{(a)}$, $\psi_{n,m}^{(s)}$, we convert the sum of ratios of the objective function in Problem $\mathbb{P}_{1}$ to the sum of three variables and the sum of one variable. Besides, we can transfer the troublesome terms about the delay of the objective function in Problem $\mathbb{P}_{1}$ into the constraints (\ref{Tu_constr}), (\ref{Tt_constr}), (\ref{Ta_constr}), and (\ref{Ts_constr}) by introducing the variables $T_n^{(u)}$, $T_{n,m}^{(t)}$, $T_{n,m}^{(a)}$, $T_{n,m}^{(s)}$. However, the constraints (\ref{psi_u_constr}) and (\ref{psi_tas_constr}) are not convex and \mbox{Problem $\mathbb{P}_{2}$} is still hard to solve, and then we introduce the \textbf{Lemma \ref{lemma_p2top3}}.

\begin{lemma}\label{lemma_p2top3}
Define non-negative multipliers $\alpha^{(u)}_n$ and $\alpha^{(i)}_{n,m}$, $i\in\{t,a,s\}$. Let $\bm{\alpha^{(u)}}:=[\alpha^{(u)}_n]|_{n\in\mathcal{N}}$, $\bm{\alpha^{(i)}}:=[\alpha^{(i)}_{n,m}]|_{n\in\mathcal{N},m\in\mathcal{M}^{(i)}}$, and $\bm{\alpha}:=\{\bm{\alpha^{(u)}}, \bm{\alpha^{(i)}}\}$, $i\in\{t,a,s\}$. The Problem $\mathbb{P}_{2}$ can be transformed into $\mathbb{P}_{3}$:
\begin{subequations}\label{prob3}
\begin{talign}
\mathbb{P}_{3}:&\max\limits_{\bm{x},\bm{\varphi},\bm{\gamma},\bm{\phi},\bm{\rho},\bm{\psi},\bm{\alpha},\bm{T}}  \sum\limits_{n \in \mathcal{N}}\sum\limits_{m \in \mathcal{M}} \Big[\alpha_{n,m}^{(t)}(c_{n,m}^{(t)}\varphi_n^{(t)}d_n \nonumber \\
&-\psi_{n,m}^{(t)}  cost_{n,m}^{(t)}) +  \alpha_{n,m}^{(a)}(c_{n,m}^{(a)}\varphi_n^{(a)}d_n - \psi_{n,m}^{(a)}cost_{n,m}^{(a)}) \nonumber \\
&+  \alpha_{n,m}^{(s)} (c_{n,m}^{(s)}\varphi_n^{(s)}d_n - \psi_{n,m}^{(s)}cost_{n,m}^{(s)})\Big] \nonumber\\
&+ \sum\limits_{n \in \mathcal{N}} \alpha_n^{(u)}(c_n^{(u)}\varphi_n^{(u)}d_n - \psi_n^{(u)}cost_n^{(u)})
\tag{\ref{prob3}}\\
\text{s.t.} \quad 
& (\text{\ref{x_range_constr}})\text{-}(\text{\ref{rho_sum_constr}}), (\text{\ref{Tu_constr}})\text{-}(\text{\ref{Ts_constr}})\nonumber.
\end{talign}
\end{subequations}
At Karush–Kuhn–Tucker (KKT) points of Problem $\mathbb{P}_{3}$, we can obtain that
\begin{talign}
&\psi_n^{(u)} = \frac{c_n^{(u)}\varphi_n^{(u)}d_n}{cost_n^{(u)}},\label{eq_psi_u}\\ 
&\psi_{n,m}^{(i)} = \frac{c_{n,m}^{(i)}\varphi_n^{(i)}d_n}{cost_{n,m}^{(i)}}, i \in \{t,a,s\}\label{eq_psi_tas},\\
&\alpha_n^{(u)} = \frac{1}{cost_n^{(u)}},\label{eq_alpha_u}\\
&\alpha_{n,m}^{(i)} = \frac{1}{cost_{n,m}^{(i)}}, i \in \{t,a,s\}.\label{eq_alpha_tas}
\end{talign}
\end{lemma}
\begin{proof}
    Refer to Appendix \ref{append_lemma_p2top3}.
\end{proof}
Based on \textbf{Lemma \ref{lemma_p2top3}}, we can split the ratio form of the objective function in Problem $\mathbb{P}_{1}$ and transform the non-convex constraints (\ref{psi_u_constr})-(\ref{psi_tas_constr}) in Problem $\mathbb{P}_{2}$ into the objective function in Problem $\mathbb{P}_{3}$ by introducing new auxiliary variables $\alpha^{(u)}_n$, $\alpha^{(t)}_{n,m}$, $\alpha^{(a)}_{n,m}$, $\alpha^{(s)}_{n,m}$. Besides, based on the analysis of the KKT conditions of Problem $\mathbb{P}_{3}$, we can obtain the relationships between auxiliary variables $[\alpha^{(u)}_n$, $\alpha^{(t)}_{n,m}$, $\alpha^{(a)}_{n,m}$, $\alpha^{(s)}_{n,m}$, $\psi_n^{(u)}$, $\psi^{(t)}_{n,m}$, $\psi^{(a)}_{n,m}$, $\psi_{n,m}^{(s)}]$ and original variables $[x_{n,m}^{(t)}$, $x_{n,m}^{(a)}$, $x_{n,m}^{(s)}$, $\varphi_n^{(u)}$, $\varphi_n^{(t)}$, $\varphi_n^{(a)}$, $\varphi_n^{(s)}$, $\gamma_n^{(u)}$, $\gamma_{n,m}^{(t)}$, $\gamma_{n,m}^{(a)}$, $\gamma_{n,m}^{(s)}$, $\phi_{n,m}^{(t)}$, $\phi_{n,m}^{(a)}$, $\phi_{n,m}^{(s)}$, $\rho_n^{(u)}$, $\rho_{n,m}^{(t)}$, $\rho_{n,m}^{(a)}$, $\rho_{n,m}^{(s)}$, $T^{(u)}_n$, $T^{(t)}_{n,m}$, $T^{(a)}_{n,m}$, $T^{(s)}_{n,m}]$ as Equations (\ref{eq_psi_u}), (\ref{eq_psi_tas}), (\ref{eq_alpha_u}), (\ref{eq_alpha_tas}). At the \mbox{$i$-th} iteration, we first fix $\bm{\alpha}^{(i-1)}$ and $\bm{\psi}^{(i-1)}$, and then optimize $\bm{x}^{(i)},\bm{\varphi}^{(i)},\bm{\phi}^{(i)},\bm{\gamma}^{(i)},\bm{\rho}^{(i)}, \bm{T}^{(i)}$. We then update $\bm{\alpha}^{(i)}$ and $\bm{\psi}^{(i)}$ according to their results. Repeat the above optimization steps until the objective function value of Problem $\mathbb{P}_{3}$ in the $i$-th and \mbox{$(i-1)$-th} iterations is less than an acceptable threshold, and we get a stationary point for Problem $\mathbb{P}_{3}$. Next, we analyze how to optimize $\bm{x},\bm{\varphi},\bm{\phi},\bm{\gamma},\bm{\rho},\bm{T}$ with the given $\bm{\psi}, \bm{\alpha}$. We consider decomposing Problem $\mathbb{P}_{3}$ into two sub-problems based on AO. They are \mbox{Sub-problem 1}: solve $\bm{\gamma}$, $\bm{\phi}$, $\bm{\rho}$, and $\bm{T}$ with fixed $\bm{x}$ and $\bm{\varphi}$; Sub-problem 2: solve $\bm{x}$, $\bm{\varphi}$, and $\bm{T}$ with fixed $\bm{\gamma}$, $\bm{\phi}$, and $\bm{\rho}$.

\subsection{Sub-problem 1: Solve \texorpdfstring{$\bm{\gamma},\bm{\phi},\bm{\rho}$}{} and \texorpdfstring{$\bm{T}$}{} with fixed \texorpdfstring{$\bm{x}$}{} and  \texorpdfstring{$\bm{\varphi}$}{}}\label{sec_subproblem1}
In this section, we analyze how to optimize $\bm{\gamma}$, $\bm{\phi}$, $\bm{\rho}$, and $\bm{T}$ with fixed $\bm{x}$ and $\bm{\varphi}$. If $\bm{x}$ and $\bm{\varphi}$ are given, Problem $\mathbb{P}_{3}$ will be a new Problem $\mathbb{P}_{4}$:
\begin{subequations}\label{prob4}
\begin{talign}
\mathbb{P}_{4}:&\max\limits_{\bm{\gamma},\bm{\phi},\bm{\rho},\bm{T}}  \sum\limits_{n \in \mathcal{N}}\sum\limits_{m \in \mathcal{M}} \alpha_{n,m}^{(t)}(c_{n,m}^{(t)}\varphi_n^{(t)}d_n - \psi_{n,m}^{(t)}cost_{n,m}^{(t)})+ \nonumber \\
&  \alpha_{n,m}^{(a)}(c_{n,m}^{(a)}\varphi_n^{(a)}d_n - \psi_{n,m}^{(a)}cost_{n,m}^{(a)}) +  \alpha_{n,m}^{(s)}(c_{n,m}^{(s)}\varphi_n^{(s)}d_n \nonumber \\
&- \psi_{n,m}^{(s)}cost_{n,m}^{(s)}) + \sum\limits_{n \in \mathcal{N}} \alpha_n^{(u)}(c_n^{(u)}\varphi_n^{(u)}d_n - \psi_n^{(u)}cost_n^{(u)})
\tag{\ref{prob4}}\\
\text{s.t.} \quad 
& (\text{\ref{phi_range_constr}})\text{-}(\text{\ref{rho_sum_constr}}), (\text{\ref{Tu_constr}})\text{-}(\text{\ref{Ts_constr}}) \nonumber.
\end{talign}
\end{subequations}
In the objective function of Problem $\mathbb{P}_{4}$, the terms $cost_{n,m}^{(t)}$, $cost_{n,m}^{(a)}$, and $cost_{n,m}^{(s)}$ are not convex due to the existence of $\frac{\text{power}}{\text{transmission data rate}}$. We will introduce the following theorem to show how to transform such non-convex terms into convex ones.
\begin{theorem}\label{theorem_subproblem1}
    Problem $\mathbb{P}_{4}$ can be transformed into a solvable concave optimization problem by a fractional programming (FP) technique.
\end{theorem}
\begin{proof}
   \textbf{Theorem \ref{theorem_subproblem1}} is proven by the following \mbox{\textbf{Lemma \ref{lemma_p4top5}}}.
\end{proof}
\begin{lemma}\label{lemma_p4top5}
Define new auxiliary variables $\varrho_{n,m}^{(t)}$, $\varrho_{n,m}^{(a)}$, $\varrho_{n,m}^{(s)}$, where 
\begin{talign}
&\varrho_{n,m}^{(t)} = \frac{1}{2\rho_n^{(u)}p_nx_{n,m}^{(t)}[\omega_b(1-\varphi_n^{(u)})d_n+d_n^{(l)}] r_{n,m_t}},\\
&\varrho_{n,m}^{(a)} = \frac{1}{2\rho_{n,m}^{(t)}p_{m_t}x_{n,m}^{(a)}[\omega_b(1-\varphi_n^{(u)}-\varphi_n^{(t)})d_n+d_n^{(l)}]r_{m_t,m_a}},\\
&\varrho_{n,m}^{(s)} \!=\!\frac{1}{2\rho_{n,m}^{(a)}p_{m_a}x_{n,m}^{(s)}[\omega_b(1-\varphi_n^{(u)}-\varphi_n^{(t)}-\varphi_n^{(a)})d_n+d_n^{(l)}]r_{m_a,m_s}}.
\end{talign}
Rewrite $cost_{n,m}^{(t)}$, $cost_{n,m}^{(a)}$, and $cost_{n,m}^{(s)}$ as new terms $\widetilde{cost}^{(t)}_{n,m}$, $\widetilde{cost}^{(a)}_{n,m}$, $\widetilde{cost}^{(s)}_{n,m}$ with $\varrho_{n,m}^{(t)}$, $\varrho_{n,m}^{(a)}$, $\varrho_{n,m}^{(s)}$, respectively. We define that
\begin{talign}
& \widetilde{cost}_{n,m}^{(t)}  \nonumber \\
&=\omega_t T_{n,m}^{(t)}+\! \omega_e \{(\rho_n^{(u)}p_n x_{n,m}^{(t)}[\omega_b(1\!-\!\varphi_n^{(u)})d_n\!+\!d_n^{(l)}])^2 \varrho_{n,m}^{(t)} \!\nonumber \\
&+\! \frac{1}{4r_{n,m_t}^2\varrho_{n,m}^{(t)}}\}+ \omega_e x_{n,m}^{(t)}e_{m_t}\kappa_{m_t}t_n\varphi_n^{(t)}d_n(\gamma_{n,m}^{(t)}f_{m_t})^2 ,\\
& \widetilde{cost}_{n,m}^{(a)}\!\! = \!\omega_t T_{n,m}^{(a)}\! +\! \omega_e \{(\rho_{n,m}^{(t)}p_{m_t}x_{n,m}^{(a)}[\omega_b(1-\varphi_n^{(u)}-\varphi_n^{(t)})\nonumber \\
&\cdot d_n+d_n^{(l)}])^2 \varrho_{n,m}^{(a)}\!+\! \frac{1}{4r_{m_t,m_a}^2\varrho_{n,m}^{(a)}}\} \!\nonumber \\
&+ \!\omega_e x_{n,m}^{(a)}e_{m_a}\kappa_{m_a}t_n\varphi_n^{(a)}d_n  f_{m_a}^2 (\gamma_{n,m}^{(a)})^2\!,\\
& \widetilde{cost}_{n,m}^{(s)} \!\!=\! \omega_t T_{n,m}^{(s)} \!\!+\!\! \omega_e \{(\rho_{n,m}^{(a)}p_{m_a}x_{n,m}^{(s)}[\omega_b(1\!\!-\!\!\varphi_n^{(u)}\!\!-\!\!\varphi_n^{(t)}\!\!-\!\!\varphi_n^{(a)})\nonumber \\
&\cdot d_n+d_n^{(l)}])^2 \varrho_{n,m}^{(s)}+ \frac{1}{4r_{m_a,m_s}^2\varrho_{n,m}^{(s)}}\}\nonumber \\
&+ \omega_e x_{n,m}^{(s)}e_{m_s}\kappa_{m_s}t_n\varphi_n^{(s)} d_n (\gamma_{n,m}^{(s)}f_{m_t})^2.
\end{talign}
Let $\bm{\varrho^{(i)}} := [\varrho^{(i)}_{n,m}|_{\forall n \in \mathcal{N},\forall m \in \mathcal{M^{(i)}}}]$, $i \in \{t,a,s\}$ and $\bm{\varrho} := \{\bm{\varrho^{(t)}},\bm{\varrho^{(a)}},\bm{\varrho^{(s)}}\}$. The Problem $\mathbb{P}_{4}$ can be transformed into the following Problem $\mathbb{P}_{5}$:
\begin{subequations}\label{prob5}
\begin{talign}
\mathbb{P}_{5}:&\max\limits_{\bm{\gamma},\bm{\phi},\bm{\rho},\bm{\varrho},\bm{T}}  \sum\limits_{n \in \mathcal{N}}\sum\limits_{m \in \mathcal{M}} \alpha_{n,m}^{(t)}(c_{n,m}^{(t)}\varphi_n^{(t)}d_n - \psi_{n,m}^{(t)}\widetilde{cost}_{n,m}^{(t)})+ \nonumber \\
&\alpha_{n,m}^{(a)}(c_{n,m}^{(a)}\varphi_n^{(a)}d_n - \psi_{n,m}^{(a)}\widetilde{cost}_{n,m}^{(a)}) +  \alpha_{n,m}^{(s)}(c_{n,m}^{(s)}\varphi_n^{(s)}d_n \nonumber \\
&- \psi_{n,m}^{(s)}\widetilde{cost}_{n,m}^{(s)}) + \sum\limits_{n \in \mathcal{N}} \alpha_n^{(u)}(c_n^{(u)}\varphi_n^{(u)}d_n - \psi_n^{(u)}cost_n^{(u)})
\tag{\ref{prob5}}\\
\text{s.t.} \quad 
& \text{(\ref{phi_range_constr})-(\ref{rho_sum_constr})}, \text{(\ref{Tu_constr})-(\ref{Ts_constr})}. \nonumber
\end{talign}
\end{subequations}
If we alternatively optimize $\bm{\varrho^{(t)}}$, $\bm{\varrho^{(a)}}$, $\bm{\varrho^{(s)}}$ and $\bm{\phi},\bm{\rho},\bm{\gamma},\bm{T}$, Problem $\mathbb{P}_{5}$ would be a concave problem. Besides, with the local optimum $\bm{\varrho}^{\bm{(t)}(\star)}$, $\bm{\varrho}^{\bm{(a)}(\star)}$, $\bm{\varrho}^{\bm{(s)}(\star)}$, we can find $\bm{\phi}^{(\star)},\bm{\rho}^{(\star)},\bm{\gamma}^{(\star)}, \bm{T}^{(\star)}$, which is a stationary point of \mbox{Problem $\mathbb{P}_{5}$}.
\end{lemma}
\begin{proof}
    Refer to Appendix \ref{append_lemma_p4top5}.
\end{proof}

From \textbf{Lemma \ref{lemma_p4top5}}, it is obvious that the function $\mathcal{F}(\phi_{n,m}^{(t)},\rho_n^{(u)},\gamma_{n,m}^{(t)},\psi_n^{(t)},T_{n,m}^{(t)})$ is an rigorous and tight upper bound of the function $\mathcal{G}(\phi_{n,m}^{(t)},\rho_n^{(u)},\gamma_{n,m}^{(t)},\psi_n^{(t)},T_{n,m}^{(t)})$. These two functions are tangent to one point, and this tangent point depends on $\varrho_{n,m}$. 
Take $\varrho_{n,m}^{(t)}$ as an example. If we choose one feasible point $(\phi_{n,m}^{(t,0)}, \rho_n^{(u,0)}, \gamma_{n,m}^{(t,0)}, \psi_n^{(t,0)}, T_{n,m}^{(t,0)})$, set $\varrho_{n,m}^{(t,0)} = \frac{1}{2\chi(\rho_n^{(u,0)}) \varsigma (\phi_{n,m}^{(t,0)}, \rho_n^{(u,0)})}$, and then we would find that the functions $\mathcal{F}(\phi_{n,m}^{(t)},\rho_n^{(u)},\gamma_{n,m}^{(t)},\psi_n^{(t)},T_{n,m}^{(t)})$ and $\mathcal{G}(\phi_{n,m}^{(t)},\rho_n^{(u)},\gamma_{n,m}^{(t)},\psi_n^{(t)},T_{n,m}^{(t)})$ are tangent to the point $(\phi_{n,m}^{(t,0)}, \rho_n^{(u,0)}, \gamma_{n,m}^{(t,0)}, \psi_n^{(t,0)}, T_{n,m}^{(t,0)})$. With the progress of optimization, this feasible point would gradually approach a local minimum.

Now, if given $\bm{\varrho}$, Problem $\mathbb{P}_{5}$ is a convex optimization problem. Fix $\bm{\varrho}$, and then optimize other variables; fix other variables, and then optimize $\bm{\varrho}$. We can transform Problem $\mathbb{P}_{4}$ into a solvable concave problem Problem $\mathbb{P}_{5}$ with the help of $\bm{\varrho}$. During the $i$-th iteration, we initially hold $\bm{\varrho}^{\bm{(t)}(i-1)}$, $\bm{\varrho}^{\bm{(a)}(i-1)}$, $\bm{\varrho}^{\bm{(s)}(i-1)}$ constant and focus on optimizing $\bm{\phi}^{(i)}, \bm{\gamma}^{(i)}, \bm{\rho}^{(i)}, \bm{T}^{(i)}$. Once these values are determined, we update $\bm{\varrho}^{\bm{(t)}(i)}$, $\bm{\varrho}^{\bm{(a)}(i)}$, $\bm{\varrho}^{\bm{(s)}(i)}$ based on the obtained results. This optimization cycle is repeated until the difference in the objective function value of Problem $\mathbb{P}_{5}$ between the $i$-th and $(i-1)$-th iterations falls in a predefined threshold. Reaching this point signifies a solution for Problem $\mathbb{P}_{5}$, and consequently, for Problem $\mathbb{P}_{4}$. The whole procedure of solving sub-problem 1 is presented in Algorithm 1.
Next, we analyze how to optimize $\bm{x}$, $\bm{\varphi}$, and $\bm{T}$ with fixed $\bm{\phi}$, $\bm{\gamma}$, and $\bm{\rho}$.

\begin{algorithm}[t]
\caption{FP technique to solve Sub-problem 1.}
\label{algorithm_subproblem1}
For all $n \in \mathcal{N}, m \in \mathcal{M}$: Randomly set a one-hot vector for each user $n$ in $\bm{x^{(k)(0)}}$, $k \in \{t,a,s\}$,
$\varphi_n^{(k)(0)}=0.25$, $k \in \{u,t,a,s\}$, $\phi_{n,m}^{(k)(0)}=\frac{1}{N}$, $k \in \{t,a,s\}$, $\rho_n^{(u)(0)}=1$, $\rho_{n,m}^{(t)(0)}=\rho_{n,m}^{(a)(0)}=\frac{1}{N}$, $\gamma_n^{(u)(0)}=1$, $\gamma_{n,m}^{(k)(0)}=\frac{1}{N}$, $k \in \{t,a,s\}$;

Calculate $\bm{\alpha}^{(0)}, \bm{\psi}^{(0)}$ with initial settings;


Initialize $j = -1$; 

Calculate $\bm{\varrho}^{(i,0)}$ with $\bm{x}^{(i)}, \bm{\varphi}^{(i)}, \bm{\phi}^{(i)}, \bm{\rho}^{(i)}, \bm{\gamma}^{(i)}$;

Set $[\bm{\phi}^{(i,0)}, \bm{\rho}^{(i,0)}, \bm{\gamma}^{(i,0)}]\leftarrow$ $[\bm{\phi}^{(i)}, \bm{\rho}^{(i)}, \bm{\gamma}^{(i)}]$;

\Repeat{$\frac{V_{\mathbb{P}_{5}}(\bm{\phi}^{(i,j+1)}, \bm{\rho}^{(i,j+1)}, \bm{\gamma}^{(i,j+1)})}{V_{\mathbb{P}_{5}}(\bm{\phi}^{(i,j)}, \bm{\rho}^{(i,j)}, \bm{\gamma}^{(i,j)})}- 1 \leq \epsilon_1$, where $\epsilon_1$ is a small positive number}{
Let $j\leftarrow j+1$;

Obtain $[\bm{\phi}^{(i,j+1)}, \bm{\rho}^{(i,j+1)}, \bm{\gamma}^{(i,j+1)}, \bm{T}^{(i,j+1)}]$ by solving Problem $\mathbb{P}_{5}$ with $\bm{\varrho}^{(i,j)}$;

Update $\bm{\varrho}^{(i,j+1)}$ with $[\bm{\phi}^{(i,j+1)}, \bm{\rho}^{(i,j+1)}, \bm{\gamma}^{(i,j+1)}, \bm{T}^{(i,j+1)}]$;
}
Return $[\bm{\phi}^{(i,j+1)}, \bm{\rho}^{(i,j+1)}, \bm{\gamma}^{(i,j+1)}]$ as a solution to Problem $\mathbb{P}_{5}$;

Set \!$[\bm{\phi}^{(i+1)},\!\bm{\rho}^{(i+1)},\bm{\gamma}^{(i+1)}]$\!$\leftarrow$\!$[\bm{\phi}^{(i,j+1)},\bm{\rho}^{(i,j+1)},\bm{\gamma}^{(i,j+1)}]$;

Return $[\bm{\phi}^{(i+1)},\!\bm{\rho}^{(i+1)},\!\bm{\gamma}^{(i+1)}]$  at the $(i+1)$-th iteration as a solution $[\bm{\phi}^\star,\bm{\rho}^\star,\bm{\gamma}^\star]$ to Problem $\mathbb{P}_{3}$.
\end{algorithm}
\subsection{Sub-problem 2: Solve \texorpdfstring{$\bm{\varphi},\bm{x}$}{}, and \texorpdfstring{$T$}{} with fixed \texorpdfstring{$\bm{\gamma},\bm{\phi},\bm{\rho}$}{}}\label{sec_subproblem2}
Once $\bm{\gamma},\bm{\phi},\bm{\rho}$ are given, Problem $\mathbb{P}_{3}$ would be Problem $\mathbb{P}_{6}$:
\begin{subequations}\label{prob6}
\begin{talign}
\mathbb{P}_{6}:&\max\limits_{\bm{x},\bm{\varphi},\bm{T}}  \sum\limits_{n \in \mathcal{N}}\sum\limits_{m \in \mathcal{M}} \Big(\alpha_{n,m}^{(t)}(c_{n,m}^{(t)}\varphi_n^{(t)}d_n - \psi_{n,m}^{(t)}cost_{n,m}^{(t)}) \nonumber \\
&+\alpha_{n,m}^{(a)}(c_{n,m}^{(a)}\varphi_n^{(a)}d_n - \psi_{n,m}^{(a)}cost_{n,m}^{(a)})  \nonumber \\
&+  \alpha_{n,m}^{(s)}(c_{n,m}^{(s)}\varphi_n^{(s)}d_n- \psi_{n,m}^{(s)}cost_{n,m}^{(s)})\Big) \nonumber \\
&+ \sum\limits_{n \in \mathcal{N}} \alpha_n^{(u)}(c_n^{(u)}\varphi_n^{(u)}d_n - \psi_n^{(u)}cost_n^{(u)})
\tag{\ref{prob6}}\\
\text{s.t.} \quad 
& (\text{\ref{x_range_constr}})\text{-}(\text{\ref{varphi_sum_constr2}}), (\text{\ref{phi_sum_constr}}), (\text{\ref{gamma_sum_constr}}), (\text{\ref{rho_sum_constr}}), (\text{\ref{Tu_constr}})\text{-}(\text{\ref{Ts_constr}}) \nonumber.
\end{talign}
\end{subequations}
Problem $\mathbb{P}_{6}$ is still an extremely complex optimization where constraints have a lot of non-convex and non-concave variable expressions with some discrete variables and continuous variables coupled together. We will then divide the complex optimization into a solvable convex optimization step by step. Let's first consider the discrete variables $\bm{x}$ in the constraint (\ref{x_range_constr}). Because of the presence of the discrete variables $\bm{x}$, Problem $\mathbb{P}_{6}$ is a mixed integer nonlinear programming. To remove the complexity caused by this discrete variable and facilitate subsequent analysis, we convert the constraint (\ref{x_range_constr}) into several new constraints:
\begin{talign}
x_{n,m}^{(i)}(x_{n,m}^{(i)} - 1) = 0, i \in \{t,a,s\}.
\end{talign}
Above new constraints can also make $x_{n,m}^{(t)}$ (or $x_{n,m}^{(a)}$ or $x_{n,m}^{(s)}$) equal $0$ or $1$. Thus, Problem $\mathbb{P}_{3}$ can be transformed into the following Problem $\mathbb{P}_{7}$:
\begin{subequations}\label{prob7}
\begin{talign}
\mathbb{P}_{7}:&\max\limits_{\bm{x},\bm{\varphi},\bm{T}}  \sum\limits_{n \in \mathcal{N}}\sum\limits_{m \in \mathcal{M}} \Big(\alpha_{n,m}^{(t)}(c_{n,m}^{(t)}\varphi_n^{(t)}d_n - \psi_{n,m}^{(t)}cost_{n,m}^{(t)}) \nonumber \\
&+\alpha_{n,m}^{(a)}(c_{n,m}^{(a)}\varphi_n^{(a)}d_n - \psi_{n,m}^{(a)}cost_{n,m}^{(a)})  \nonumber \\
&+  \alpha_{n,m}^{(s)}(c_{n,m}^{(s)}\varphi_n^{(s)}d_n- \psi_{n,m}^{(s)}cost_{n,m}^{(s)})\Big) \nonumber \\
&+ \sum\limits_{n \in \mathcal{N}} \alpha_n^{(u)}(c_n^{(u)}\varphi_n^{(u)}d_n - \psi_n^{(u)}cost_n^{(u)})
\tag{\ref{prob7}}\\
\text{s.t.} \quad 
& x_{n,m}^{(i)}(x_{n,m}^{(i)} \!-\! 1) \!=\! 0, \forall n \in \mathcal{N},\forall m \in \mathcal{M}^{(t)},\forall i \in \{t,a,s\},\label{x_range_constr_new}\\
& (\text{\ref{x_sum_constr}})\text{-}(\text{\ref{varphi_sum_constr2}}), (\text{\ref{phi_sum_constr}}), (\text{\ref{gamma_sum_constr}}), (\text{\ref{rho_sum_constr}}), (\text{\ref{Tu_constr}})\text{-}(\text{\ref{Ts_constr}}) \nonumber.
\end{talign}
\end{subequations}
The reformulated Problem $\mathbb{P}_{7}$ contains at most two coupled continuous variables in each term, i.e., $\bm{x}$ and $\bm{\varphi}$. These variables appear as bilinear products, e.g., $x_{n,m}^{(i)} \varphi_n^{(i)}$, in the objective function. Although bilinear terms are generally non-convex, the fact that both variables are bounded within $[0,1]$ allows us to apply convex relaxation techniques, e.g., QCQP and semidefinite relaxation (SDR), to transform the \mbox{Problem $\mathbb{P}_{7}$} into a solvable convex optimization, i.e., \mbox{\textbf{Theorem \ref{theorem_subproblem2}}}.

\begin{theorem}\label{theorem_subproblem2}
    Problem $\mathbb{P}_{7}$ can be transformed into a solvable convex optimization problem by using QCQP and SDR techniques.
\end{theorem}
\begin{proof}
    \textbf{Theorem \ref{theorem_subproblem2}} is proven by the following \mbox{\textbf{Lemma \ref{lemma_p7top8}}} and \mbox{\textbf{Lemma \ref{lemma_p8top9}}}.
\end{proof}
\begin{lemma}\label{lemma_p7top8}
Problem $\mathbb{P}_{7}$ can be transformed into a standard QCQP Problem $\mathbb{P}_{8}$:
\begin{subequations}\label{prob8}
\begin{talign}
\mathbb{P}_{8}:&\min\limits_{\bm{x},\bm{\varphi},\bm{T}} -\bm{Q}^\intercal \bm{P}_0 \bm{Q} - \bm{W}_0^\intercal \bm{Q} - T^{(u)} - T^{(t)} - T^{(a)} - T^{(s)}
\tag{\ref{prob8}}\\
\text{s.t.} \quad
&\text{diag}(\bm{e}_{M^{(i)}}^\intercal\bm{Q})(\text{diag}(\bm{e}_{M^{(i)}}^\intercal\bm{Q})- \bm{I})= 0, \forall i \in \{t,a,s\},\label{x_range_constr_qcqp}\\
&\text{diag}(\bm{e}_{\overline{1_i},\overline{M^{(i)}_i}}^\intercal \bm{e}_{M^{(i)}}^\intercal\bm{Q}) = \bm{I},\forall i \in \{t,a,s\},\label{x_sum_constr_qcqp}\\
&\text{diag}(\bm{e}_{\varphi_i}^\intercal\bm{Q}) \leq \bm{I},\forall i \in \{u,t,a,s\},\label{varphi_range_constr_qcqp}\\
&\text{diag}\big((\bm{e}_{\varphi_u}^\intercal+\bm{e}_{\varphi_t}^\intercal+\bm{e}_{\varphi_a}^\intercal+\bm{e}_{\varphi_s}^\intercal)\bm{Q}\big) = \bm{I},\label{varphi_sum_constr_qcqp}\\
&\bm{\phi^{(i)}}\bm{e}_{M^{(i)}}^\intercal\bm{Q}-1 \leq 0,\forall i \in \{t,a,s\},\label{phi_sum_constr_qcqp}\\
&\bm{\gamma^{(i)}}\bm{e}_{M^{(i)}}^\intercal\bm{Q}-1 \leq 0,\forall i \in \{t,a,s\},\label{gamma_sum_constr_qcqp}\\
&\bm{\rho^{(i)}}\bm{e}_{M^{(i)}}^\intercal\bm{Q}-1 \leq 0,\forall i \in \{t,a\},\label{rho_sum_constr_qcqp}\\
&{\bm{P^{(T_u)}}}^\intercal \bm{Q} \leq T^{(u)},\label{Tu_constr_qcqp}\\
&\bm{Q}^\intercal \bm{P_1^{(T_i)}} \bm{Q} + {\bm{P_2^{(T_i)}}}^\intercal \bm{Q} \leq T^{(i)},\forall i \in \{t,a,s\}.\label{Ttas_constr_qcqp}
\end{talign}
\end{subequations}
\end{lemma}
\begin{proof}
    Refer to Appendix \ref{append_lemma_p7top8}.
\end{proof}
In Problem $\mathbb{P}_8$, the newly introduced matrix $\bm{Q}$ is constructed based on the variables $\bm{x}$ and $\bm{\varphi}$ so that all quadratic terms in $\mathbb{P}_7$ can be compactly represented in a unified matrix form. The other new terms in the objective function and constraints refer to auxiliary or parameter-dependent matrices introduced to facilitate the transformation, while preserving the equivalence with the original formulation. These matrices mainly serve to collect constant coefficients and coupling terms so that the resulting expressions become standard quadratic forms with respect to $\bm{Q}$. For detailed construction and definitions, please refer to Appendix~\ref{append_lemma_p7top8}.

Unfortunately, Problem $\mathbb{P}_{8}$ is still non-convex. Then, we will use the SDR method to transform this QCQP problem into a semidefinite programming (SDP) problem that can be efficiently handled by off-the-shelf solvers. To achieve this goal, we introduce a new matrix variable: $\bm{S}:=(\bm{Q}^\intercal,1)^\intercal(\bm{Q}^\intercal,1)$. This lifted matrix captures the outer product structure of the optimization variable $\bm{Q}$ and enables the reformulation of various quadratic constraints and objective terms into equivalent linear matrix trace expressions, as shown in the following Lemma. In the subsequent step, we will relax the implicit rank-one constraint on $\bm{S}$, which yields a convex SDP relaxation of $\mathbb{P}_8$ and provides a tractable surrogate for the original non-convex problem.

\begin{lemma}\label{lemma_p8top9}
QCQP Problem $\mathbb{P}_{8}$ can be finally transformed into a solvable SDR Problem $\mathbb{P}_{9}$:
\begin{subequations}\label{prob9}
\begin{talign}
\mathbb{P}_{9}: &\min\limits_{\bm{S},T^{(u)},T^{(t)},T^{(a)},T^{(s)}}\quad  \text{Tr}(\bm{P}_1 \bm{S})\tag{\ref{prob9}}\\
\text{s.t.} \quad & \text{Tr}(\bm{P}_2 \bm{S})=0, \label{x_constr1_sdr}\\       & \text{Tr}(\bm{P}_3 \bm{S})=0, \label{x_constr2_sdr}\\
         & \text{Tr}(\bm{P}_4 \bm{S})\leq0, \label{varphi_constr_sdr}\\
         & \text{Tr}(\bm{P}_5 \bm{S})=0, \label{varphi_sum_sdr}\\
         & \text{Tr}(\bm{P}_6 \bm{S})\leq0, \label{x_phi_constr_sdr}\\
         & \text{Tr}(\bm{P}_7 \bm{S})\leq0, \label{x_gamma_constr_sdr}\\
         & \text{Tr}(\bm{P}_8 \bm{S})\leq0, \label{x_rho_constr_sdr}\\
         & \text{Tr}(\bm{P}_9 \bm{S})\leq T^{(u)}, \label{Tu_constr_sdr}\\
         & \text{Tr}(\bm{P}_{10} \bm{S})\leq T^{(i)}, \forall i \in \{t,a,s\},\label{Ttas_constr_sdr}\\
         & \bm{S}\succeq0, \label{S_constr_sdr}
\end{talign}
\end{subequations}
where $\text{Tr}(\cdot)$ means the trace of a matrix.
\end{lemma}
\begin{proof}
    Refer to Appendix \ref{append_lemma_p8top9}.
\end{proof}

In the Problem $\mathbb{P}_{9}$, each constraint and objective term in the Problem $\mathbb{P}_{8}$ is rewritten as $\text{Tr}(\bm{P}_i \bm{S})$ for $i \in \{1,2,\cdots,10\}$, where $\bm{P}_i$ is a predefined symmetric matrix derived from the corresponding constraint. For more details, please refer to Appendix \ref{append_lemma_p8top9}.
Now, Problem $\mathbb{P}_{8}$ is finally transformed into a solvable SDR Problem $\mathbb{P}_{9}$. Standard convex solvers can efficiently solve the SDR Problem $\mathbb{P}_{9}$ in polynomial time, providing a continuous version of $\vect{Q}$. However, this version often only serves as the lower bound for the ideal solution and may not satisfy the $\text{rank}(\mathbf{S})=1$ constraint. To rectify this, we apply rounding techniques. The final $NM$ components of $\vect{Q}$, represented by $x_{n,m}$ for every $n \in \mathcal{N}$ and $m \in \mathcal{M}$, reflect the partial connection of users to servers. If the sum $\sum_{m \in \mathcal{M}} x_{n,m}$ exceeds 1 for any user, we normalize $x_{n,m}$ by dividing it by the absolute sum. The Hungarian algorithm \cite{kuhn1955hungarian}, augmented with zero vectors, is used to identify the optimal matching, denoted as $\mathcal{X}$. Within this matching, we set $x_{n,m}$ to 1 if nodes $n$ and $m$ are paired, and 0 otherwise, labeling this as $\bm{x}^\star$. We set the results of $\bm{\varphi}$ in $\vect{Q}$ as $\bm{\varphi}^\star$. The whole procedure of solving sub-problem 2 is presented in Algorithm 2.
\begin{algorithm}[t]
\caption{QCQP method to solve Sub-problem 2.}
\label{algorithm_subproblem2}
For all $n \in \mathcal{N}, m \in \mathcal{M}$: Randomly set a one-hot vector for each user $n$ in $\bm{x^{(k)(0)}}$, $k \in \{t,a,s\}$,
$\varphi_n^{(k)(0)}=0.25$, $k \in \{u,t,a,s\}$, $\bm{\phi}^{(0)}$, $\bm{\rho}^{(0)}$, and $\bm{\gamma}^{(0)}$ obtained by Algorithm \ref{algorithm_subproblem1};

Calculate $\bm{\alpha}^{(0)}, \bm{\psi}^{(0)}$ with initial settings;


Initialize $j = -1$; 

Set $[\bm{x}^{(i,0)}, \bm{\varphi}^{(i,0)}] \leftarrow [\bm{x}^{(i)}, \bm{\varphi}^{(i)}]$;

Initialize $[\bm{P}_1^{(i,0)}$, $\bm{P}_2^{(i,0)}$, $\bm{P}_3^{(i,0)}$, $\bm{P}_4^{(i,0)}$, $\bm{P}_5^{(i,0)}$, $\bm{P}_6^{(i,0)}$, $\bm{P}_7^{(i,0)}$, $\bm{P}_8^{(i,0)}$, $\bm{P}_9^{(i,0)}$, $\bm{P}_{10}^{(i,0)}]$$\leftarrow$ $[\bm{P}_1^{(i)}$, $\bm{P}_2^{(i)}$, $\bm{P}_3^{(i)}$, $\bm{P}_4^{(i)}$, $\bm{P}_5^{(i)}$, $\bm{P}_6^{(i)}$, $\bm{P}_7^{(i)}$, $\bm{P}_8^{(i)}$, $\bm{P}_9^{(i)}$, $\bm{P}_{10}^{(i)}]$;

\Repeat{$\frac{V_{\mathbb{P}_{9}}(\bm{x}^{(i,j+1)}, \bm{\varphi}^{(i,j+1)}}{V_{\mathbb{P}_{9}}(\bm{x}^{(i,j)}, \bm{\varphi}^{(i,j)})}- 1 \leq \epsilon_2$, where $\epsilon_2$ is a small positive number}{
Let $j\leftarrow j+1$;

Obtain $[\bm{x}^{(i,j+1)}, \bm{\varphi}^{(i,j+1)}]$ of continuous values by solving Problem $\mathbb{P}_{9}$;

Update $[\bm{P}_1^{(i,j+1)}$, $\bm{P}_2^{(i,j+1)}$, $\bm{P}_3^{(i,j+1)}$, $\bm{P}_4^{(i,j+1)}$, $\bm{P}_5^{(i,j+1)}$, $\bm{P}_6^{(i,j+1)}$, $\bm{P}_7^{(i,j+1)}$, $\bm{P}_8^{(i,j+1)}$, $\bm{P}_9^{(i,j+1)}$, $\bm{P}_{10}^{(i,j+1)}]$ with $[\bm{x}^{(i,j+1)}, \bm{\varphi}^{(i,j+1)}]$;
}

Return $[\bm{x}^{(i,j+1)}, \bm{\varphi}^{(i,j+1)}]$ as a solution to the SDR Problem $\mathbb{P}_{9}$;

If the sum $\sum_{m \in \mathcal{M}} x_{n,m}$ exceeds 1 for any user, we normalize $x_{n,m}$ by dividing it by the absolute sum. Use the Hungarian algorithm augmented with zero vectors to identify the optimal matching, denoted as $\mathcal{X}$. Within this matching, we set $x_{n,m}$ to 1 if nodes $n$ and $m$ are paired and 0 otherwise. Denote that integer association results as $\bm{x}_\star^{(i,j+1)}$.

Set $[\bm{x}^{(i+1)}, \bm{\varphi}^{(i+1)}]\leftarrow$ $[\bm{x}_\star^{(i,j+1)}, \bm{\varphi}^{(i,j+1)}]$;

Return $[\bm{x}^{(i+1)}, \bm{\varphi}^{(i+1)}$ at the $(i+1)$-th iteration as a solution $[\bm{x}^\star, \bm{\varphi}^\star]$ to Problem $\mathbb{P}_{3}$.
\end{algorithm}
\begin{figure}[htbp]
\begin{talign}
    \small\textnormal{\textbf{Theorem \ref{theorem_solvep1}}} \Leftarrow 
    \begin{cases}
        \textnormal{\textbf{Lemma \ref{lemma_p1top2}}}\\
        \textnormal{\textbf{Lemma \ref{lemma_p2top3}}}\\
        \textnormal{\textbf{Theorem \ref{theorem_subproblem1}}} \Leftarrow \textnormal{\textbf{Lemma \ref{lemma_p4top5}}}\\
        \textnormal{\textbf{Theorem \ref{theorem_subproblem2}}} \Leftarrow 
        \begin{cases}
            \textnormal{\textbf{Lemma \ref{lemma_p7top8}}}\\
            \textnormal{\textbf{Lemma \ref{lemma_p8top9}}} \nonumber
        \end{cases}
    \end{cases}\\
    \small\textnormal{\text{$\mathbb{P}_{1}$}} \small\stackrel{\textnormal{\textbf{Lemma \ref{lemma_p1top2}}}}\Longrightarrow \textnormal{\text{$\mathbb{P}_{2}$}} \small\stackrel{\textnormal{\textbf{Lemma \ref{lemma_p2top3}}}}\Longrightarrow \textnormal{\text{$\mathbb{P}_{3}$}} \small\stackrel{\textnormal{\textbf{AO}}}\rightarrow
    \begin{cases}
        \small\textnormal{\text{$\mathbb{P}_{4}$}} \small\stackrel{\textnormal{\textbf{Lemma \ref{lemma_p4top5}}}}\Longrightarrow \textnormal{\text{$\mathbb{P}_{5}$}}\\
        \small\textnormal{\text{$\mathbb{P}_{7}$}} \small\stackrel{\textnormal{\textbf{Lemma \ref{lemma_p7top8}}}}\Longrightarrow \textnormal{\text{$\mathbb{P}_{8}$}} \small\stackrel{\textnormal{\textbf{Lemma \ref{lemma_p8top9}}}}\Longrightarrow \textnormal{\text{$\mathbb{P}_{9}$}}\nonumber
    \end{cases}
\end{talign}
    \caption{A graph of the transformation relationship between Problems, Theorems, and Lemmas.}
    \label{fig:relat_prob_theorem_lemma}
\end{figure}
\subsection{Whole procedure of proposed PARA algorithm}
Let the objective function value of Problem $\mathbb{P}_{i}$ be $V_{\mathbb{P}_{i}}$. 
Here we summarize the overall flow of the optimization algorithm. At the $i$-th iteration, we first initialize $\bm{\alpha}^{(i-1)}, \bm{\psi}^{(i-1)}$ with $\bm{x}^{(i-1)}$, $\bm{\varphi}^{(i-1)}$, $\bm{\phi}^{(i-1)}$, $\bm{\rho}^{(i-1)}$, $\bm{\gamma}^{(i-1)}$. Then, we fix $\bm{\alpha}, \bm{\psi}$ as $\bm{\alpha}^{(i-1)}, \bm{\psi}^{(i-1)}$ and optimize $\bm{x}$, $\bm{\varphi}$, $\bm{\phi}$, $\bm{\rho}$, $\bm{\gamma}$. For the optimization of $\bm{x}$, $\bm{\varphi}$, $\bm{\phi}$, $\bm{\rho}$, $\bm{\gamma}$, we use the alternative optimization technique. 

In the first step, we fix $\bm{x}$, $\bm{\varphi}$ as $\bm{x}^{(i-1)}$, $\bm{\varphi}^{(i-1)}$ and optimize $\bm{\phi}$, $\bm{\rho}$, $\bm{\gamma}$. At this optimization step, we also introduce an auxiliary variable $\varrho^{(t)}_{n,m}$, $\varrho^{(a)}_{n,m}$, $\varrho^{(s)}_{n,m}$ to transform Problem $\mathbb{P}_{4}$ into a solvable concave problem $\mathbb{P}_{5}$. At the $j$-th inner iteration, we initialize $\bm{\varrho}^{\bm{(t)}(i-1,j-1)}$, $\bm{\varrho}^{\bm{(a)}(i-1,j-1)}$, $\bm{\varrho}^{\bm{(s)}(i-1,j-1)}$ with $\bm{x}^{(i-1)}$, $\bm{\varphi}^{(i-1)}$, $\bm{\phi}^{(i-1,j-1)}$, $\bm{\rho}^{(i-1,j-1)}$, $\bm{\gamma}^{(i-1,j-1)}$. We fix $\bm{\varrho}^{\bm{(t)}}$, $\bm{\varrho}^{\bm{(a)}}$, $\bm{\varrho}^{\bm{(s)}}$ as $\bm{\varrho}^{\bm{(t)}(i-1,j-1)}$, $\bm{\varrho}^{\bm{(a)}(i-1,j-1)}$, $\bm{\varrho}^{\bm{(s)}(i-1,j-1)}$ and optimize $\bm{\phi}$, $\bm{\rho}$, $\bm{\gamma}$. Then we obtain the optimization results $\bm{\phi}^{(i-1,j)}$, $\bm{\rho}^{(i-1,j)}$, $\bm{\gamma}^{(i-1,j)}$ and update $\bm{\varrho}^{\bm{(s)}(i-1,j)}$ with these results. This optimization cycle is repeated until the difference in the objective function value of Problem $\mathbb{P}_{5}$ between the $j$-th and $(j-1)$-th iterations falls below a predefined threshold. We set the results of this alternative optimization step as $\bm{\phi}^{(i)}$, $\bm{\rho}^{(i)}$, $\bm{\gamma}^{(i)}$.

In the second step, we fix the $\bm{\phi}$, $\bm{\rho}$, $\bm{\gamma}$ as $\bm{\phi}^{(i)}$, $\bm{\rho}^{(i)}$, $\bm{\gamma}^{(i)}$ and optimize $\bm{x}$, $\bm{\varphi}$. Then we first obtain $\bm{\varphi}^{(i)}$ and the continuous solution of $\bm{x}$ by solving Problem $\mathbb{P}_{10}$. Next, we use the Hungarian algorithm to obtain the discrete solution of $\bm{x}$ and denote it as $\bm{x}^{(i)}$.  Until now, we have obtained $\bm{x}^{(i)}$, $\bm{\varphi}^{(i)}$, $\bm{\phi}^{(i)}$, $\bm{\rho}^{(i)}$, $\bm{\gamma}^{(i)}$. Update $\bm{\alpha}^{(i)}, \bm{\psi}^{(i)}$ with those results.

Repeat these two optimization steps until the difference in the objective function value of Problem $\mathbb{P}_{3}$ between the $i$-th and $(i-1)$-th iterations falls in a predefined threshold. Then, we set the optimization results as $\bm{x}^\star$, $\bm{\varphi}^\star$, $\bm{\phi}^\star$, $\bm{\rho}^\star$, $\bm{\gamma}^\star$. The whole procedure of the PARA algorithm is presented in Algorithm 3.
\begin{algorithm}[t]
\caption{Whole procedure of proposed PARA algorithm in SAGIN.}
\label{algorithm}
Initialize $i \leftarrow -1$ and for all $n \in \mathcal{N}, m \in \mathcal{M}$: Randomly set a one-hot vector for each user $n$ in $\bm{x^{(k)(0)}}$, $k \in \{t,a,s\}$,
$\varphi_n^{(k)(0)}=0.25$, $k \in \{u,t,a,s\}$, $\phi_{n,m}^{(k)(0)}=\frac{1}{N}$, $k \in \{t,a,s\}$, $\rho_n^{(u)(0)}=1$, $\rho_{n,m}^{(t)(0)}=\rho_{n,m}^{(a)(0)}=\frac{1}{N}$, $\gamma_n^{(u)(0)}=1$, $\gamma_{n,m}^{(k)(0)}=\frac{1}{N}$, $k \in \{t,a,s\}$;

Calculate $\bm{\alpha}^{(0)}, \bm{\psi}^{(0)}$ with initial settings;

\Repeat{$\frac{V_{\mathbb{P}_{3}}(\bm{x}^{(i+1)}, \bm{\varphi}^{(i+1)}, \bm{\phi}^{(i+1)}, \bm{\rho}^{(i+1)}, \bm{\gamma}^{(i+1)})}{V_{\mathbb{P}_{3}}(\bm{x}^{(i)}, \bm{\varphi}^{(i)}, \bm{\phi}^{(i)}, \bm{\rho}^{(i)}, \bm{\gamma}^{(i)})}- 1 \leq \epsilon_3$, where $\epsilon_3$ is a small positive number}{
Let $i\leftarrow i+1$;

Obtain $[\bm{\phi}^{(i+1)}, \bm{\rho}^{(i+1)}, \bm{\gamma}^{(i+1)}]$ as a solution to Problem $\mathbb{P}_{5}$ by \textbf{Algorithm \ref{algorithm_subproblem1}};

Obtain $[\bm{x}^{(i+1)}, \bm{\varphi}^{(i+1)}]$ as a solution to Problem $\mathbb{P}_{8}$ by \textbf{Algorithm \ref{algorithm_subproblem2}};

Update $[\bm{\alpha}^{(i+1)}, \bm{\psi}^{(i+1)}]$ with $[\bm{x}^{(i+1)}, \bm{\varphi}^{(i+1)}, \bm{\phi}^{(i+1)}, \bm{\rho}^{(i+1)}, \bm{\gamma}^{(i+1)}]$;
}
Return $[\bm{x}^{(i+1)}, \bm{\varphi}^{(i+1)}, \bm{\phi}^{(i+1)}, \bm{\rho}^{(i+1)}, \bm{\gamma}^{(i+1)}]$ as a solution $[\bm{x}^\star, \bm{\varphi}^\star, \bm{\phi}^\star, \bm{\rho}^\star, \bm{\gamma}^\star]$ to Problem $\mathbb{P}_{3}$.
\end{algorithm}

\subsection{Mobility-aware PEFT edge training for SAGIN networks}\label{sec.mobility}
To evaluate PARA’s robustness under realistic mobility, we adopt a mobility-aware model that divides training into discrete time slots. In each slot, the network is quasi-static, and server and user positions are fixed within the slot but may change across slots due to aerial and satellite mobility. \cite{sun2024multicast, zhang2021space}. 

\subsubsection{Server and user settings}
For server mobility, we assume that each aerial or satellite server has a service duration limit $T^{(a)}_{stay}$ or $T^{(s)}_{stay}$, denoting the number of time slots it remains in the coverage area. When the aerial or satellite servers become unavailable, we denote the corresponding out-of-coverage durations as $T^{(a)}_{out}$ or $T^{(s)}_{out}$. During the service duration of one aerial or satellite server, the related channel gain is unchanged until it is offline. Fig. \ref{fig.service_duration} illustrates the service and out-of-coverage timelines for aerial and satellite servers.
Before leaving, a server checks if it can complete its tasks in time; if not, it signals its lower-level entity for reassignment. In the next slot, a replacement server takes over, and PARA is re-executed to reallocate resources for the remaining workload.

Each user has multiple sequential training tasks, offloaded one at a time. In each round, all current tasks must be completed before users proceed to the next. If a previously assigned server becomes unavailable, the task is reassigned to a newly available server of the same type. To ensure seamless handover, servers are assumed to have sufficient storage to buffer user data and intermediate states. Signaling delays and coordination overheads are ignored for simplicity, and the number of servers per layer remains constant over time.
\begin{figure}[t]
\vspace{-0.1cm}
\centering
\includegraphics[width=0.49\textwidth]{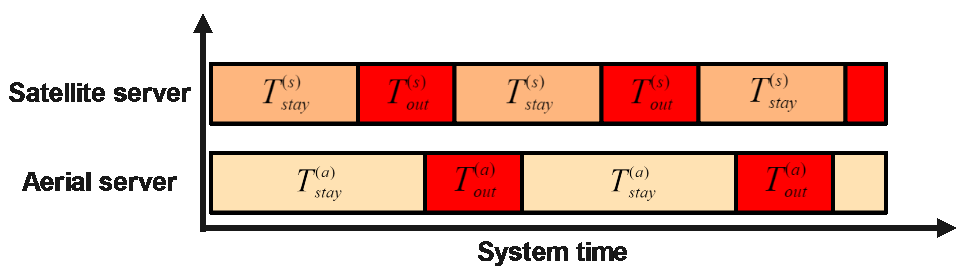}
\caption{Timeline of service and out-of-coverage durations for aerial and satellite servers.}\vspace{-5pt}
\label{fig.service_duration}
\end{figure}
\begin{figure}[t]
\vspace{-0.1cm}
\centering
\includegraphics[width=0.45\textwidth]{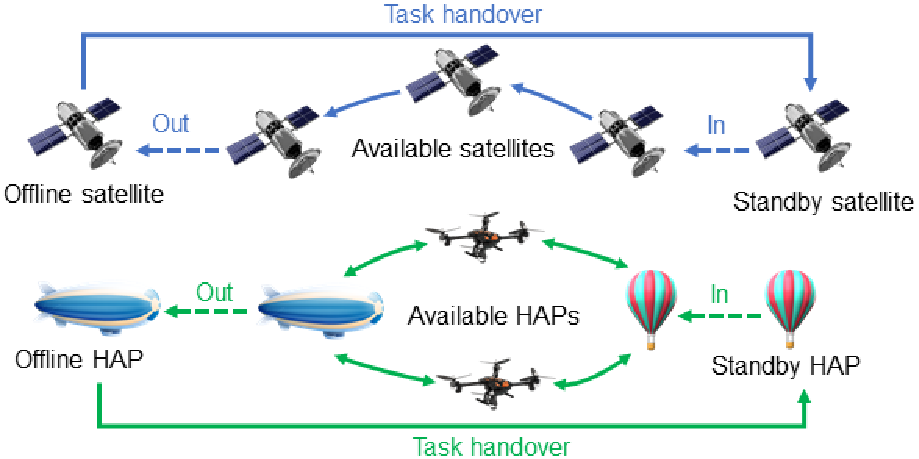}
\caption{Server availability in mobility-aware SAGIN networks.}\vspace{-5pt}
\label{fig.mobility_system}
\end{figure}
\subsubsection{PARA algorithm under mobility-aware SAGIN networks}
In the mobility-aware SAGIN networks, the PARA algorithm proceeds as follows:
\begin{itemize}
    \item \textit{Initialization:} At the beginning, the Problem $\mathbb{P}_{1}$ solved by the PARA algorithm. The obtained solutions are stored as the baseline resource configuration.
    \item \textit{Task execution loop:} This phase includes multiple task rounds. In each round, the system checks whether the assigned aerial and satellite servers remain within coverage for the duration of the offloading and processing. If a user cannot complete their task due to server unavailability (mobility-induced), we set this user to silent and release the related resources. Then we use the PARA algorithm to reallocate the released resources (e.g., bandwidth, power) among the remaining users. If any tasks remain unfinished after reallocation, the system performs a retry loop: It continues redistributing resources until all users complete their tasks or servers become available again.
\end{itemize}
In the next section, we will present the novelty and some potential applications of our proposed algorithm.

\subsection{Novelty and applications of our proposed algorithm}
In this paper, we address maximizing the combined PTE of users and servers in a SAGIN system, using the PARA algorithm. This algorithm optimizes user-server association, and work offloading ratio together, as well as jointly optimizes communication and computational resources like bandwidth, transmission power, and computing allocations for both users and servers. Unlike previous methods that treat communication and computational resources separately, our approach integrates them into a unified optimization problem, leading to better solutions than traditional alternating optimization methods. Additionally, the PARA algorithm's application extends beyond PTE maximization; it's also suitable for solving energy efficiency and various utility-cost ratio problems. For non-concave utility functions, we use successive convex approximation (SCA) \cite{razaviyayn2014successive} to enable PARA's application in mobile edge computing for user connection and resource allocation in wireless scenarios.

The proposed network architecture is designed for emerging mobile AIGC scenarios, e.g., personalized LLM fine-tuning, real-time semantic image generation, and adaptive multimodal learning in mobile edge environments. These applications demand low-latency, privacy-preserving, and resource-aware training, which SAGIN is well-suited to support by leveraging the complementary strengths of terrestrial, aerial, and satellite servers.

For example, consider a remote agricultural monitoring system where autonomous drones collect field data and use generative AI models to generate crop health summaries or alerts. These models must be fine-tuned on-site for different regions or crop types. Since terrestrial edge connectivity may be limited, the drone offloads part of the model tuning to an aerial platform (e.g., HAP) and then to a satellite node. This setup ensures model updates are completed efficiently without relying on distant cloud data centers, while maintaining low latency and minimizing power use.

The PARA algorithm is intended to be executed at a centralized terrestrial controller (e.g., terrestrial base station), which has global knowledge of user demands and SAGIN network states. The controller computes optimal associations and resource schedules and distributes them to involved nodes, making the framework practically deployable and efficient in dynamic multi-layered edge computing systems.

\subsection{Complexity analysis}\label{sec.complexity_analysis}
In this section, we analyze the complexity of the proposed PARA algorithm. In Algorithm \ref{algorithm_subproblem1}, there are $3N+3NM+NM^{(t)}+NM^{(a)}$ variables and $2N+3NM+2M+NM^{(t)}+NM^{(a)}+M^{(t)}+M^{(a)}$ constraints. Note that $M = M^{(t)}+M^{(a)}+M^{(s)}$. The worst-case complexity of it is $\mathcal{O}((N^{3.5}+M^{3.5}+N^{3.5}M^{3.5})\log(\frac{1}{\epsilon_1}))$ with a given solution accuracy $\epsilon_1 > 0$. In Algorithm \ref{algorithm_subproblem2}, there are $NM+4N+4$ variables and $5N+2NM+2M+M^{(t)}+M^{(a)}+4$ constraints. The worst-case complexity of it is $\mathcal{O}((N^{3.5}+M^{3.5}+N^{3.5}M^{3.5})\log(\frac{1}{\epsilon_2}))$ with a given solution accuracy $\epsilon_2 > 0$. The complexity of the Hungarian algorithm is $\mathcal{O}(N^3M^3)$. To summarize, if Algorithm \ref{algorithm} takes $\mathcal{I}$ iterations, the whole complexity is $\mathcal{O}\big(\mathcal{I}(N^{3.5}+M^{3.5}+N^{3.5}M^{3.5})\log(\frac{1}{\epsilon_3})\big)$ with a given solution accuracy $\epsilon_3 > 0$ \cite{wang2022intelligent}.

\section{Numerical Results}\label{sec.numerical_results}
In this section, we present the default settings and numerical results.
\subsection{Default settings}
We first consider a SAGIN topology of 20 users, three terrestrial servers, three aerial servers, and two LEO servers. This moderate-scale setting models a representative regional SAGIN segment and follows common practice in existing SAGIN resource-allocation studies\cite{gao2024cost,tang2022blockchain}.
The path loss between the user $n$ and server $m$ is modeled as $128.1+37.6 \log_{10}d_{ut}$, where $d_{ut}$ denotes the  Euclidean distance between the user and terrestrial server and $d_{ut}$ is no more than 1 km. The path loss between a terrestrial server and an aerial server is $116.7+15 \log_{10}\frac{d_{ta}}{2.6\times10^3}$\cite{li2020enabling}, where $d_{ta}$ is the distance between them. The path loss between an aerial server and an LEO satellite is the same as between a terrestrial server and an aerial server\cite{li2020enabling}. We set $d_{as}$ as 550 km, which is the same setting as Starlink LEO networks. $d_{ta}$ is 20 km. To match practical systems, we set the variable $T^{(s)}$ that is no more than seven minutes to keep the constant link between the aerial server and the LEO server. Gaussian noise power spectral density $\sigma^2$ is $-174$ dBm. The total bandwidth for each server is 10 MHz. The maximum transmit power of mobile users is 2 W. The maximum transmit power of servers is 20 W. We assume the GPU resource utilization is $0.55$ for users and servers. The maximum GPU computation speed of mobile users is $19.58$ TFLOPs with four GTX 1080 GPUs and that of servers is $1372.8$ TFLOPs with eight A100 GPUs. The computational efficiency of mobile users and servers ($\kappa_n$ and $\kappa_m$) is $10^{-38}$. We refer to the adapter parameter sizes in \cite{zhang2023llama} and \cite{gao2023llama}. The training parameter sizes of mobile users are randomly selected from $[1.2, 14]$ M. To achieve this, pseudorandom values are generated, which follow a standard uniform distribution over the open interval $(0,1)$. These pseudorandom values are then scaled to the range of $[1.2, 14]$ M to determine the specific adapter parameter sizes for each mobile user. The token data sizes of users are randomly selected from $[10, 50]$ Mbits and $d_n^{(l)}$ is almost double that. we consider the ``float32'' method to represent the floating-point number and  $\omega_b$ is 32. User and server training epochs $e_n$ and $e_m$ are both one. Delay and energy weights are set as $0.5$, and we reduce the value of the energy by a factor of 1000 so that the energy and delay are in the same order. PTE preferences of users and servers $c_n$ and $c_{n,m}$ are set as one. The Mosek tool in MATLAB is used to conduct the simulations. The hardware configuration is NVIDIA GeForce RTX 2080.
\begin{figure}[t]
\centering
\includegraphics[width=0.4\textwidth]{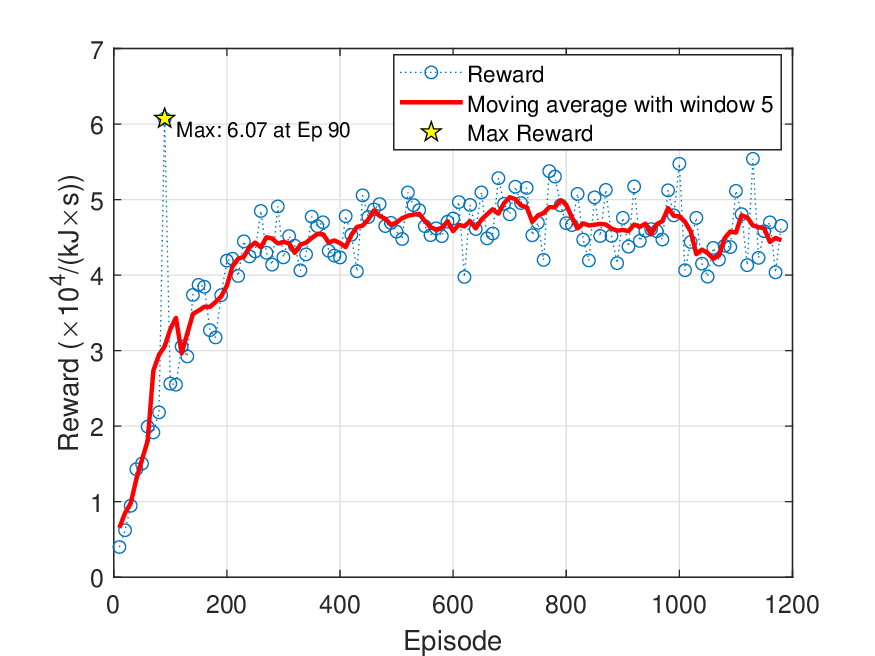}
\caption{Reward convergence of the PPO method during the training phase.}
\label{fig.ppo_reward}
\end{figure}
\begin{figure*}[t]
\vspace{-0.3cm}
\subfigure[Performance comparisons with baselines.]{\includegraphics[width=.33\textwidth]{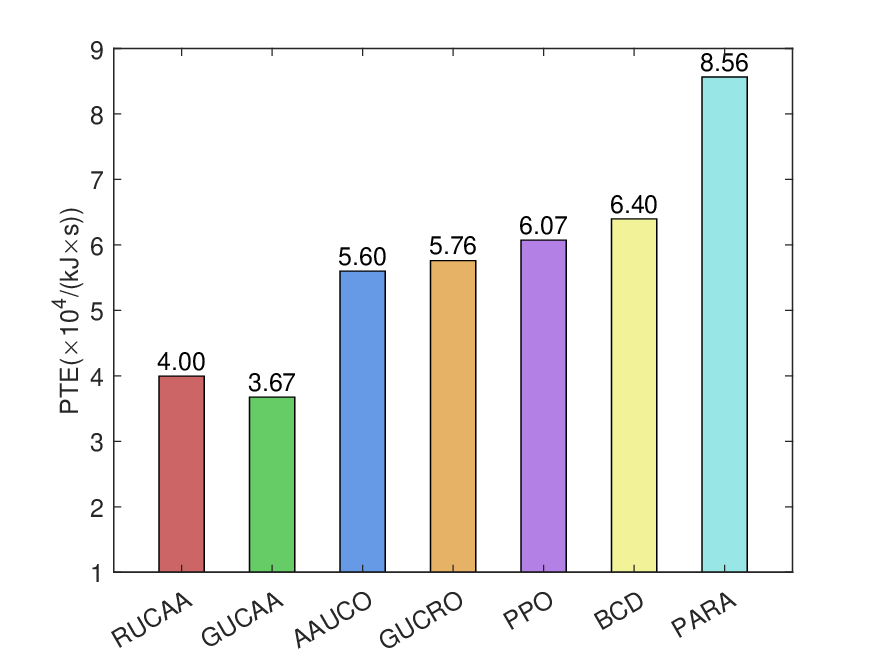}\label{fig.comparison_baselines}}
\subfigure[Performance comparisons under different bandwidth.]{\includegraphics[width=.33\textwidth]{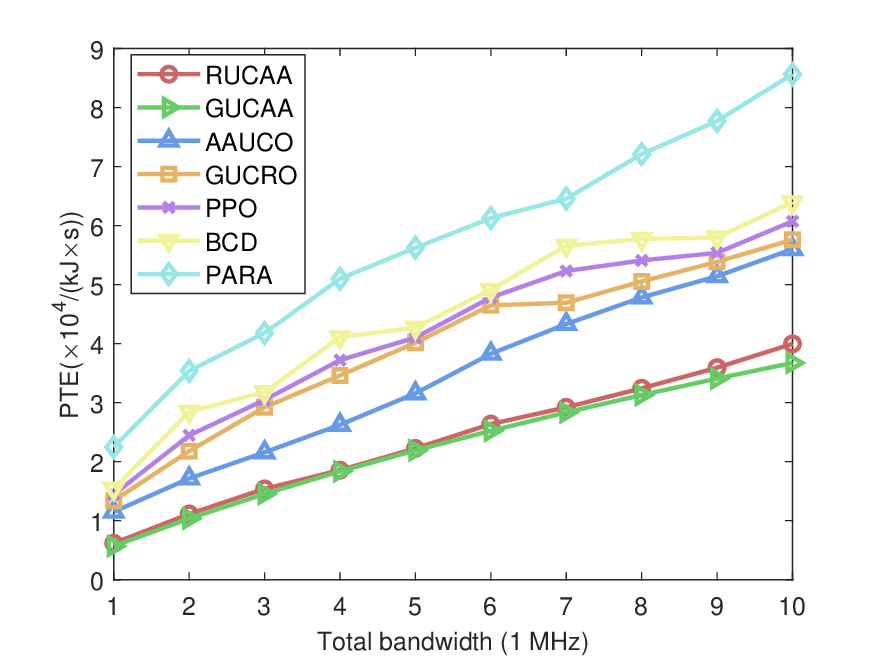}\label{fig.comparison_b}}
\subfigure[Performance comparisons under different user computing speeds.]{\includegraphics[width=.33\textwidth]{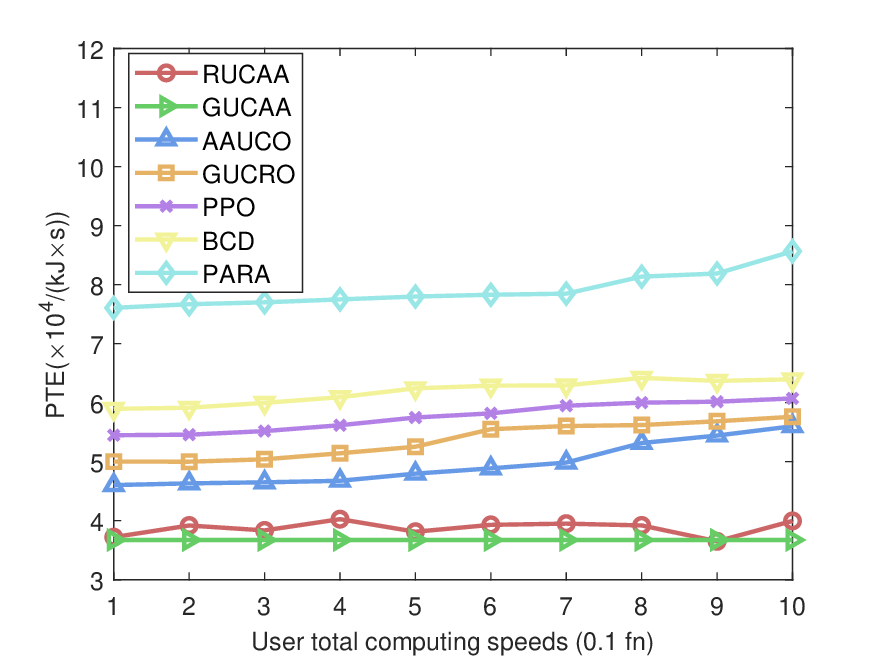}\label{fig.comparison_fu}} \vspace{-10pt}
\caption{Performance comparisons with baselines and under different bandwidth and user computing speeds.}
\end{figure*}
\begin{figure*}[t]
\vspace{-0.3cm}
\subfigure[Performance comparisons under different server computing speeds.]{\includegraphics[width=.33\textwidth]{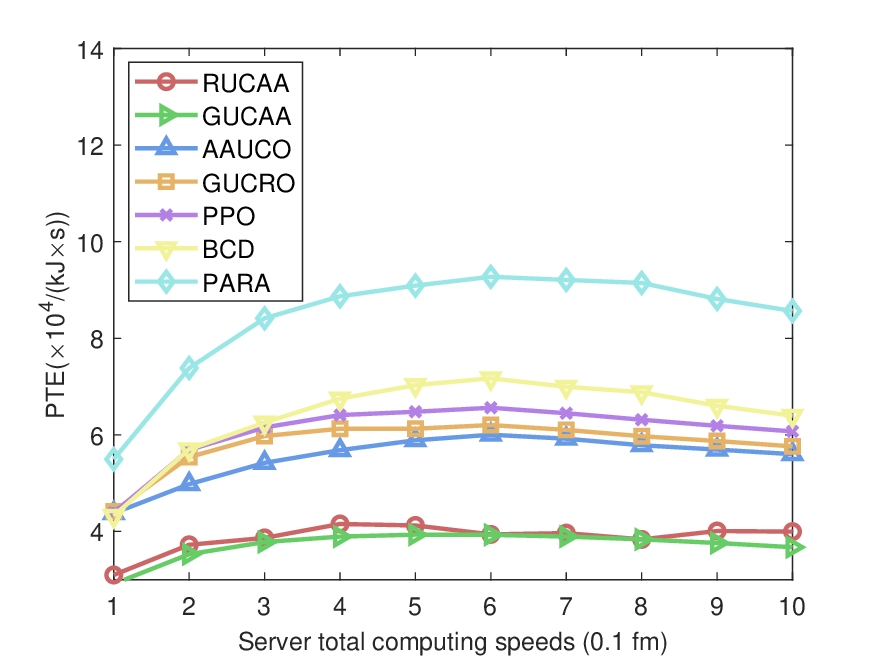}\label{fig.comparison_fs}}
\subfigure[Performance comparisons under different user transmit powers.]{\includegraphics[width=.33\textwidth]{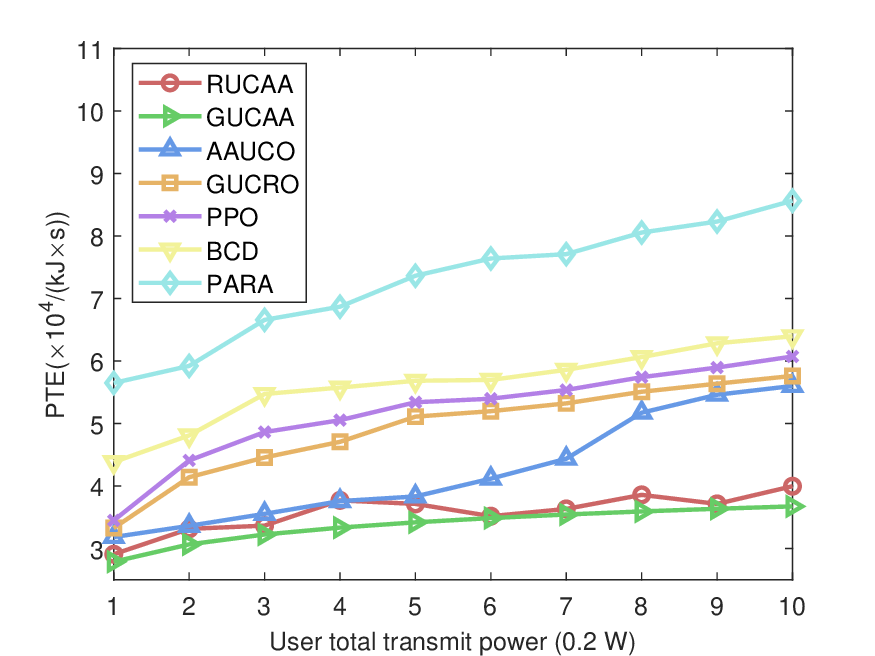}\label{fig.comparison_pu}}
\subfigure[Performance comparisons under different server transmit powers.]{\includegraphics[width=.33\textwidth]{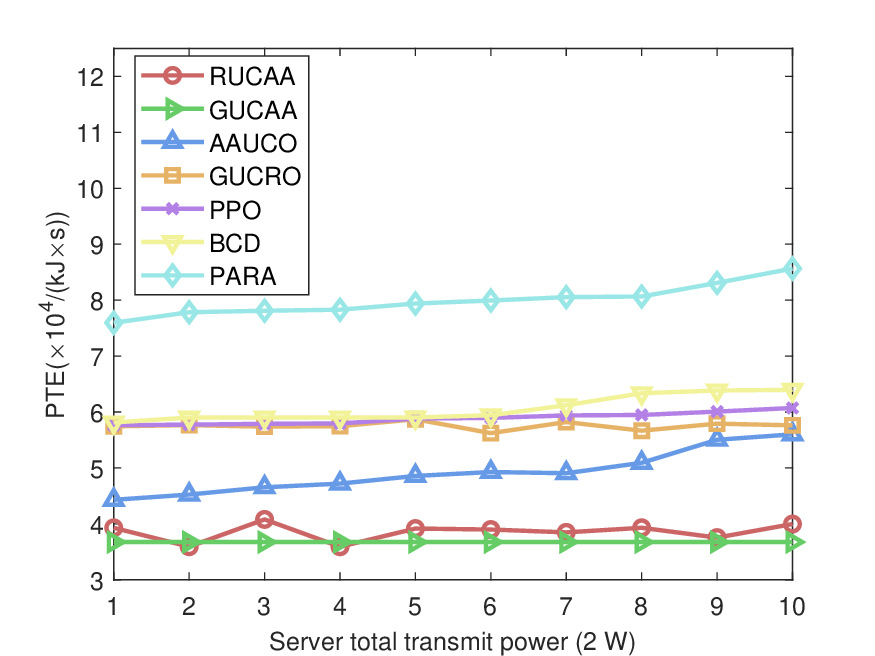}\label{fig.comparison_ps}} \vspace{-10pt}
\caption{Performance comparisons under different server computing speeds, user transmit powers, and server transmit powers.} 
\end{figure*}

\subsection{Performance comparison with other baselines}
We choose the following baselines in \cite{qian2024user}: RUCAA (random user connection with average resource allocation), GUCAA (greedy user connection with average resource allocation), AAUCO (average resource allocation with user connection optimization), and GUCRO (greedy user connection with resource allocation optimization).
Note that user connection optimization and resource allocation refer to the QCQP and FP methods in Sections \ref{sec_subproblem2} and \ref{sec_subproblem1}, respectively. We also choose the block coordinate descent (BCD) optimization method, which iteratively improves the solution by solving the problem along one variable at a time, as a baseline \cite{tseng2001convergence}. The popular reinforcement learning (RL) is also considered in our comparable simulations. In the RL baseline, we adopt a fixed user association identical to that obtained in the BCD method, while employing a PPO framework \cite{yu2022surprising} to optimize the remaining variables. The state, reward, and action are defined by the channel gain matrix, the objective value in Problem $\mathbb{P}_{1}$, and the variable set comprising $\boldsymbol{\phi}$, $\boldsymbol{\rho}$, $\boldsymbol{\varphi}$, and $\boldsymbol{\gamma}$, respectively. In principle, a more comprehensive RL-based baseline would also include user association $\boldsymbol{x}$ as part of the action space. However, this is impractical due to the exponentially large discrete action space introduced by the one-hot constraint on $\boldsymbol{x}$. Therefore, the proposed PPO method is designed as a pragmatic compromise, enabling meaningful comparison against traditional RL methods while avoiding intractable action space complexity.

Fig.~\ref{fig.ppo_reward} illustrates the reward convergence behavior of the PPO-based resource allocation method. The blue circles represent the raw reward values sampled every 10 training episodes, while the red curve denotes the smoothed reward trajectory computed using a moving average with a window size of 5. The training process exhibits a rapid initial increase in reward, demonstrating effective learning in the early stages. The maximum reward, marked by a gold star, is achieved at Episode 90 with a value of approximately 6.07 $(\times 10^4/(kJ \times s))$. After reaching this peak, the reward stabilizes with moderate fluctuations, indicating convergence of the PPO policy. This convergence behavior validates the efficacy and stability of the RL-based baseline PPO in optimizing resource allocation under the defined environment.

In Fig. \ref{fig.comparison_baselines}, we present the simulation results with other baselines. In the comparative analysis of user connection and resource allocation strategies, the proposed PARA method emerges as the most effective. Unlike the RUCAA and GUCAA methods, which either randomly connect users or employ a greedy approach without fully optimizing resource distribution, or the AAUCO and GUCRO strategies that optimize either user connection or resource allocation but not both, PARA integrates all aspects of network optimization into a combinative framework. By leveraging a holistic optimization approach, PARA significantly outperforms the conventional strategies, including the BCD method, which only optimizes variables in a block-wise manner, and the PPO method, which does not jointly optimize all variables.

\subsection{Performance comparison of different communication and computation resources}
\subsubsection{Bandwidth}The simulation data in Fig. \ref{fig.comparison_b} reveals a clear trend across different user association and resource allocation strategies as the total bandwidth of each level increases from 1 MHz to 10 MHz. The PARA method consistently outperforms the other approaches, demonstrating significant gains, especially as bandwidth increases. Notably, while AAUCO, RUCAA, and GUCAA show comparable performance with relatively modest improvements as bandwidth expands, BCD, PPO, and GUCRO exhibit more pronounced growth, suggesting that resource allocation optimization plays a key role in leveraging additional bandwidth effectively.

\subsubsection{User computing speed}In Fig. \ref{fig.comparison_fu}, the PTE performance impact of varying computational resource allocations (from $0.1f_n$ to $f_n$) is reflected. The PARA method consistently demonstrates superior performance as computational resources increase, with its performance metric peaking at approximately 8.56 $(\times 10^4/(kJ \times s))$. Interestingly, while the performance of RUCAA shows variations with changes in user computing speeds, indicating sensitivity to resource allocation, AAUCO, GUCRO, PPO, and BCD exhibit a more stable increase in performance, with BCD showing significant improvement towards higher resource allocations. Notably, GUCAA's performance remains relatively constant, suggesting that its greedy user connection strategy may not effectively leverage additional computational resources compared to the other methods.
\begin{figure*}[t]
\vspace{-0.3cm}
\subfigure[Performance comparisons under different user and server configurations.]{\includegraphics[width=.33\textwidth]{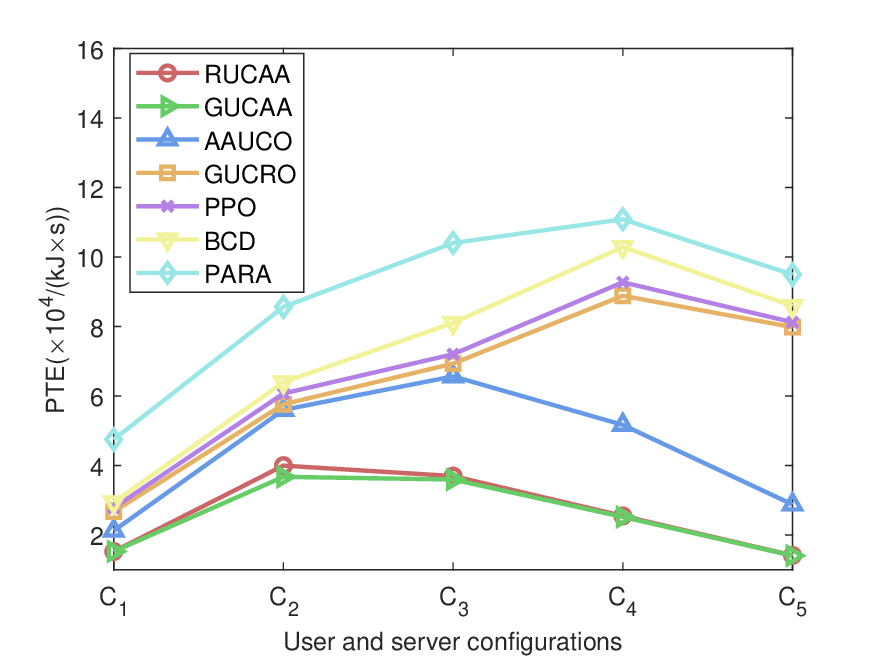}\label{fig.comparison_user_server_configs}}
\subfigure[Performance comparisons under energy and delay weight parameters.]{\includegraphics[width=.33\textwidth]{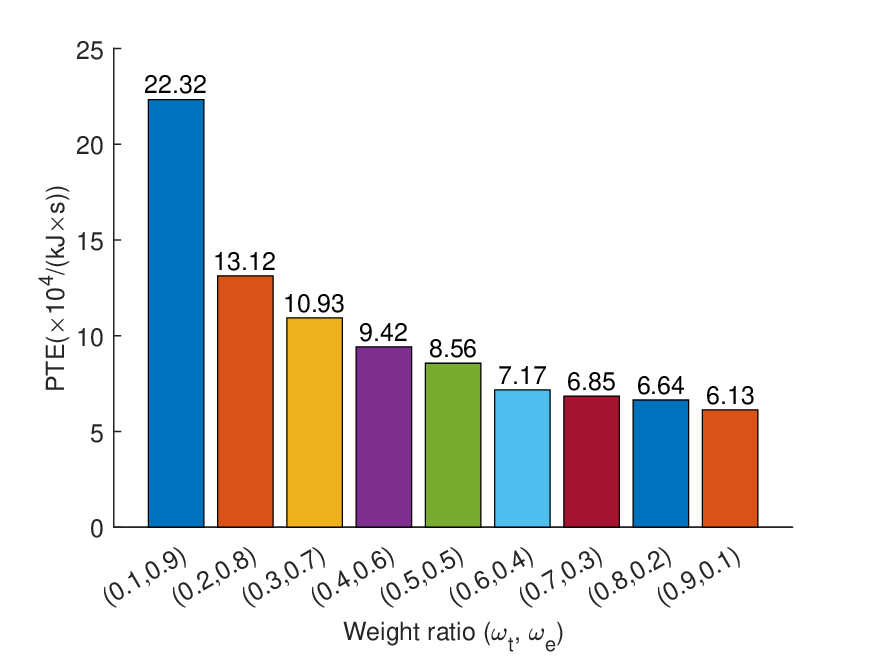}\label{fig.comparison_omega}}
\subfigure[Performance comparisons under different PTE preferences.]{\includegraphics[width=.33\textwidth]{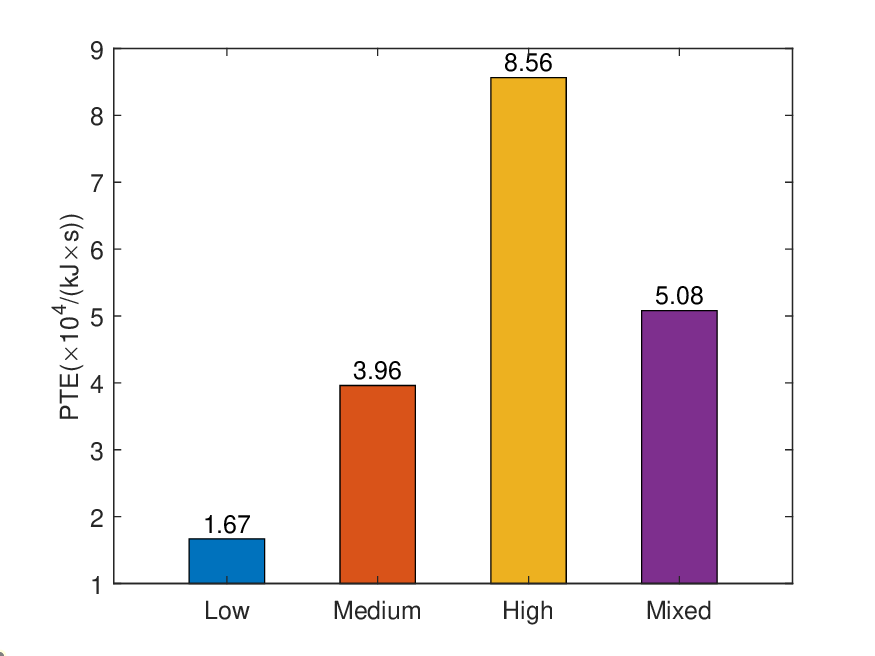}\label{fig.comparison_preferences}} \vspace{-10pt}
\caption{Performance comparisons under different user and server configurations, energy and delay weight parameters, and PTE preferences.} 
\end{figure*}
\begin{figure}[t] 
\vspace{-0.3cm}
\subfigure[Task completion ratio.]{\includegraphics[width=.24\textwidth]{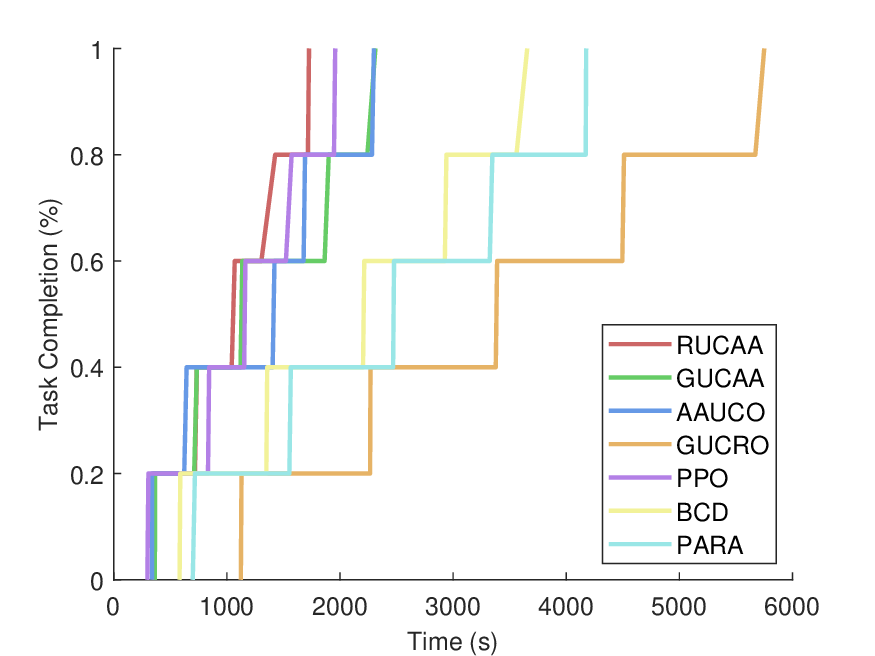}\label{fig.mobility_taskcompletion}}
\subfigure[PTE performance comparison.]{\includegraphics[width=.24\textwidth]{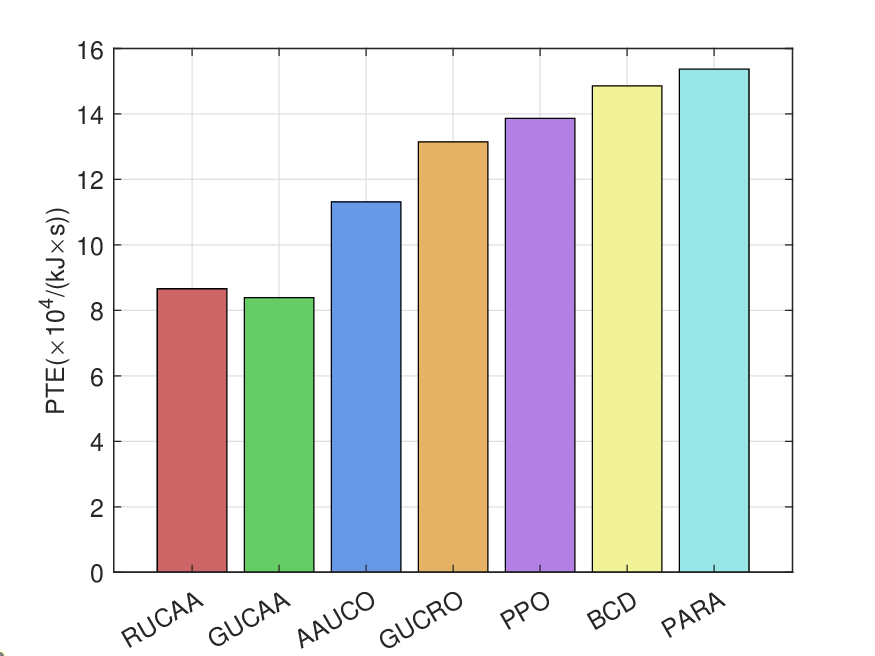}\label{fig.mobility_pte}} 
\subfigure[System delay.]{\includegraphics[width=.24\textwidth]{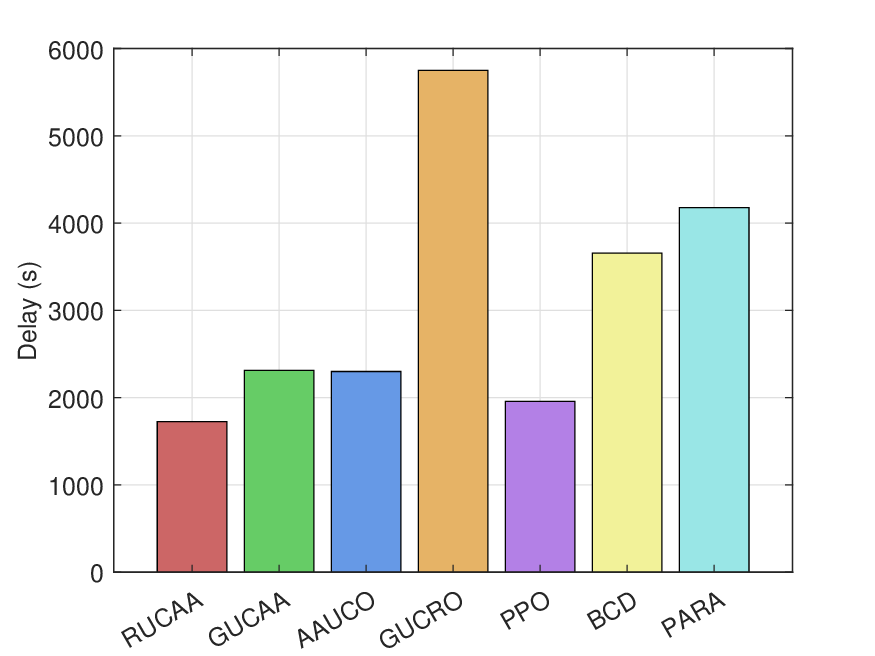}\label{fig.mobility_delay}}
\subfigure[Energy consumption.]{\includegraphics[width=.24\textwidth]{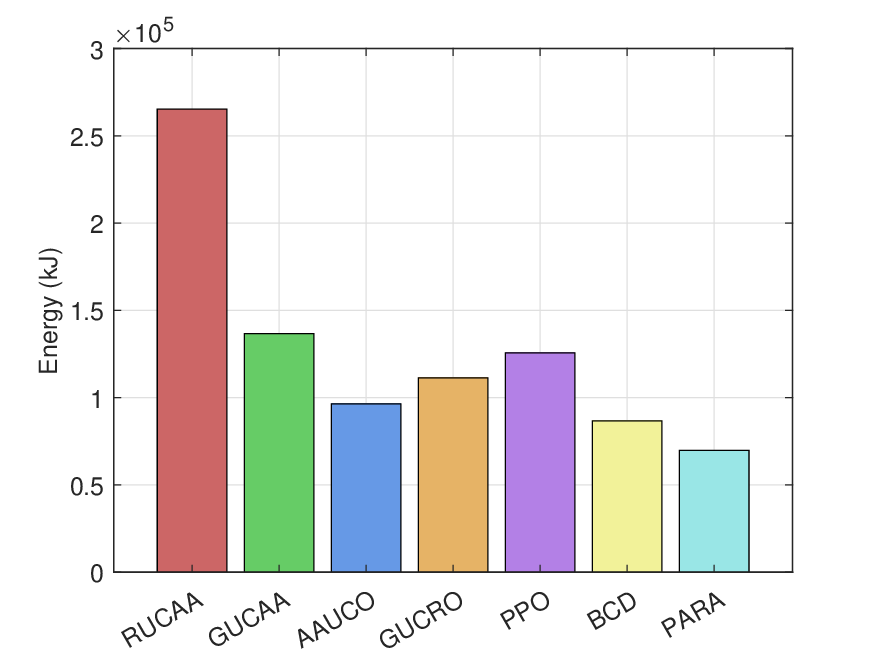}\label{fig.mobility_energy}}\vspace{-10pt}
\caption{Performance comparisons under dynamic SAGIN topology.} 
\label{fig.mobility}
\end{figure}
\subsubsection{Server computing speed}In Fig. \ref{fig.comparison_fs}, we present the impact of increasing server computational resources (from $0.1f_m$ to $f_m$). The PARA method distinctly outshines the other strategies, demonstrating a robust increase in performance as server resources are augmented, peaking at an impressive 9.27 $(\times 10^4/(kJ \times s))$ before a slight decline as resources continue to increase. This suggests an optimal range for resource allocation beyond which additional resources do not translate into proportional performance gains, possibly due to inefficiencies or saturation in resource utilization. The same thing happens with other baselines.

\subsubsection{User transmit power}The simulation results shown in Fig. \ref{fig.comparison_pu} highlight how increasing user transmission power (from 0.2 W to 2 W) boosts performance. The PARA method consistently improves as power increases, reaching its best performance at 2 W. This shows that PARA effectively uses extra power to boost network performance by joint optimization of user association and resource allocation. On the other hand, the RUCAA and GUCAA methods see only modest improvements with more power, hinting that they might not be making the most of the extra power for better performance. AAUCO and GUCRO also get better with more power, but not as quickly as PARA, with GUCRO especially benefiting at the higher power settings, showing its strength in using more power for optimizing resources. The PPO and BCD methods, strong comparison points, also improve significantly at higher power levels, but don't reach the performance levels of PARA.

\subsubsection{Server transmit power}In Fig. \ref{fig.comparison_ps}, the influence of progressively increasing server transmission power from 2 W to 20 W across different optimization strategies is studied, with the PARA method outshining others by effectively leveraging higher power to significantly enhance performance. While RUCAA's performance fluctuates, suggesting a complex relationship between transmission power and its random connection strategy, GUCAA remains notably stable, indicating its insensitivity to changes in server transmit power. In contrast, AAUCO demonstrates an upward trend, benefiting from the power increase, yet GUCRO exhibits some variability, reflecting the challenges in optimally utilizing additional power. PPO and BCD show consistent improvements, particularly at higher power levels.

\subsection{Performance comparison of heterogeneous settings}
\subsubsection{Different user and server configurations}
We consider five user and server configurations in Table \ref{tab:configurations}. In Fig. \ref{fig.comparison_user_server_configs}, the simulation data across different user and server configurations reveals a distinct pattern in performance across various optimization strategies. As the number of users and servers increases, the PARA method consistently outperforms other baseline strategies, showcasing its superior capability to adapt and optimize resource allocation, user connection, and offloading ratios effectively. Notably, while RUCAA and GUCAA exhibit modest performance, likely due to their simpler allocation and connection strategies, AAUCO and GUCRO show significant improvements, suggesting the effectiveness of user connection optimization and resource optimization, respectively. However, GUCRO, PPO, and BCD, which focus on resource optimization, optimizing variables without user association, and a baseline comparison, respectively, also demonstrate substantial gains in larger configurations, indicating their potential to handle increased complexity.
\begin{table}[ht]
\centering
\caption{User and server configurations.}
\begin{tabular}{ccccc}
\toprule
Configuration & $N$ & $M_t$ & $M_a$ & $M_s$ \\
\midrule
$C_1$ & 10 & 2 & 2 & 2 \\
$C_2$ & 20 & 3 & 3 & 2 \\
$C_3$ & 40 & 4 & 4 & 3 \\
$C_4$ & 80 & 8 & 5 & 4 \\
$C_5$ & 160 & 16 & 8 & 5 \\
\bottomrule
\end{tabular}
\label{tab:configurations}
\end{table}
\subsubsection{Delay and energy weights}In Fig. \ref{fig.comparison_omega}, the impact of varying weights for delay and energy consumption ($\omega_t$ and $\omega_e$) on the system PTE performance is analyzed. As the weight shifts from prioritizing energy efficiency towards a more balanced consideration with delay, there's a notable decrease in the PTE performance, from 22.32 $(\times 10^4/(kJ \times s))$ when the emphasis is heavily on energy efficiency (0.1, 0.9) to 6.13 $(\times 10^4/(kJ \times s))$ when the delay is prioritized (0.9, 0.1). This trend indicates a trade-off between delay and energy efficiency, where focusing solely on minimizing energy consumption leads to higher PTE performance, which gradually diminishes as the emphasis shifts towards reducing delay.
\subsubsection{PTE preference}We consider four preference parameter setting cases: 1) low preference: set $c_n$ and $c_{n,m}$ as $0.2$; 2) medium preference: set $c_n$ and $c_{n,m}$ as $0.5$; 3) high preference: set $c_n$ and $c_{n,m}$ as $1$; 4) mixed preference: set $c_n$ and $c_{n,m}$ as $a$, where $a$ is a random value uniformly taken from [0, 1].  
Figure~\ref{fig.comparison_preferences} shows PARA’s performance under different user preference settings. High preferences lead to the best performance, highlighting the importance of aligning resource allocation with user needs. Low preferences result in the weakest outcomes, while medium and mixed settings yield moderate improvements.
These variations underscore the critical role of understanding and integrating user preferences into optimization processes to enhance system effectiveness within the PARA framework.

\subsection{Performance under mobility-aware SAGIN networks}
We set $T^{(a)}_{stay}$, $T^{(s)}_{stay}$, $T^{(a)}_{out}$, and $T^{(s)}_{out}$ as 600 s, 420 s, 100 s, 100 s, respectively. Each user has five training tasks for edge offloading. Fig. \ref{fig.mobility_taskcompletion} illustrates the task completion ratio over time across all methods. The PPO and RUCAA methods demonstrate rapid and stable task execution, indicating effective initial allocation and low overhead in dynamic conditions. The GUCAA and AAUCO algorithms show consistent but slower progress, suggesting more conservative or static resource strategies. In contrast, the PARA, GUCRO, and BCD methods experience prolonged delays in later stages, as also reflected in Fig. \ref{fig.mobility_delay}. However, the PARA method obtains the lowest energy consumption in Fig. \ref{fig.mobility_energy} and in Fig. \ref{fig.mobility_pte}, it also obtains the best PTE performance, after which is the BCD method.
More discussion about weight settings and additional results for dynamic SAGIN with different weight settings are provided in Appendix \ref{sec.weight_setting} and \ref{sec.more_results_under_mobility}, respectively.

\section{Conclusion and Future Work}\label{sec.conclusion}
In conclusion, our work focuses on the optimization of SAGIN for maximizing parameter training efficiency. The introduction of a new metric, PTE, for assessing data processing efficiency, coupled with the proposed PARA technique. We study the joint optimization of user association, offloading ratios, and communication and computational resource allocations across SAGIN's layered architecture sets it apart from existing methodologies. Theoretical proofs and simulation results demonstrate the effectiveness of the proposed optimization technique, presenting a stationary point solution to the sum of the ratios optimization problem.

While our focus has been on theoretical modeling and algorithm development, we recognize the importance of addressing real-world deployment challenges. These include synchronization delay and the practical implementation of distributed fine-tuning at scale. In future work, we plan to extend our framework to incorporate these aspects and validate it in more practical systems and scenarios.

\bibliographystyle{IEEEtran}
\bibliography{ref}

\begin{IEEEbiography}
[{\includegraphics[width=1in,height=1.25in,clip,keepaspectratio]{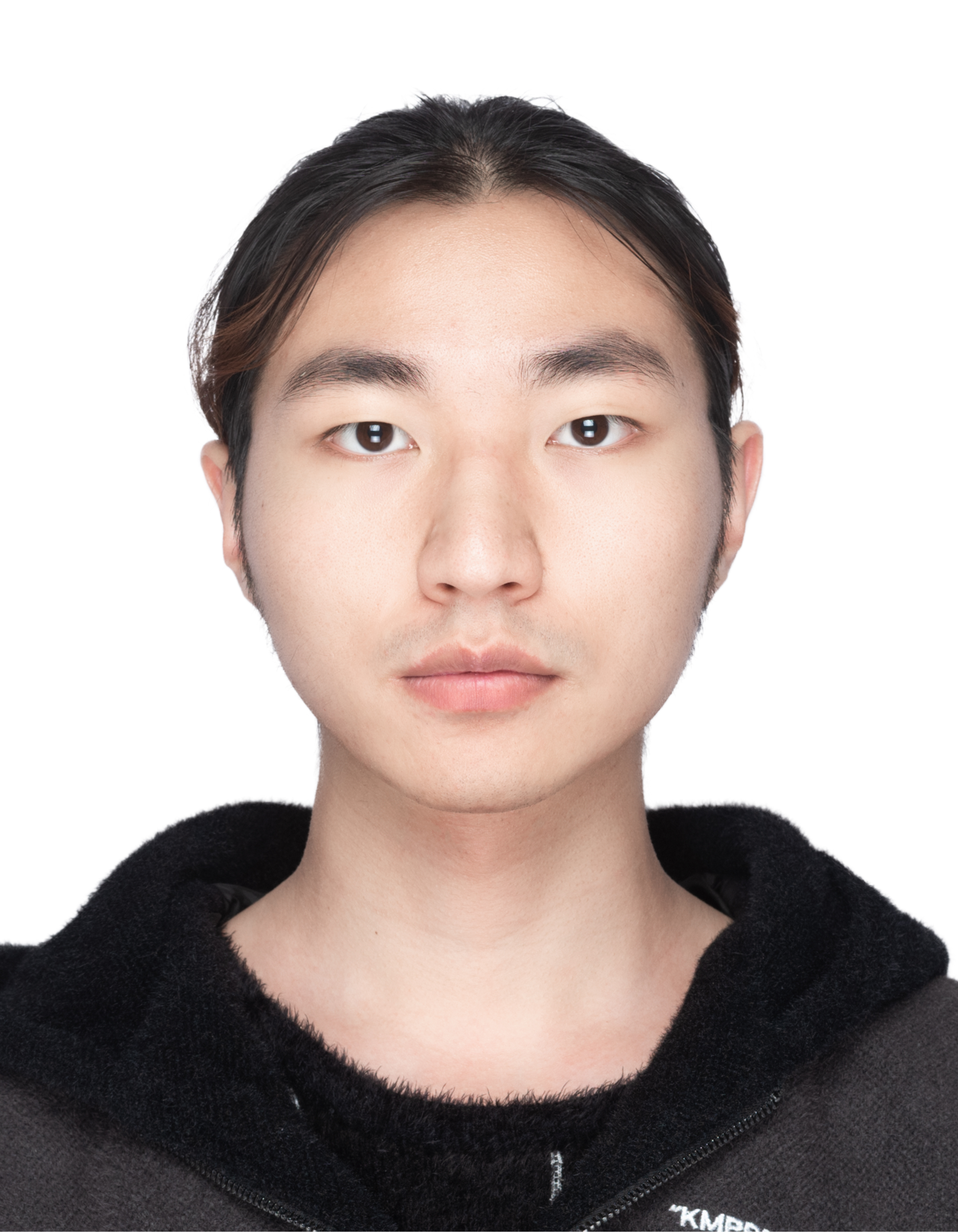}}]
{Liangxin Qian}
(Graduate Student Member, IEEE) received bachelor's and master's degrees in communication engineering from the University of Electronic Science and Technology of China, Chengdu, China, in 2019 and 2022, respectively. He is currently working toward his Ph.D. at the College of Computing and Data Science (CCDS), Nanyang Technological University, Singapore. His research interests include bilevel optimization, mobile edge computing, and secure communications.
\end{IEEEbiography}

\begin{IEEEbiography}
[{\includegraphics[width=1in,height=1.25in,clip,keepaspectratio]{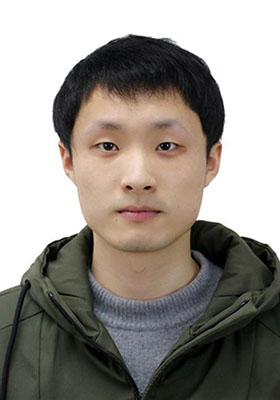}}]
{Peiyuan Si}
(Graduate Student Member, IEEE) received bachelor's and master's degrees in communication engineering from Zhejiang University of Technology of China, Zhejiang, China, in 2018 and 2021, respectively. He is currently working toward his Ph.D. at the College of Computing and Data Science, Nanyang Technological University, Singapore. His research interests include semantic communication, unmanned aerial vehicles, and reinforcement learning.
\end{IEEEbiography}

\begin{IEEEbiography}
[{\includegraphics[width=1in,height=1.25in,clip,keepaspectratio]{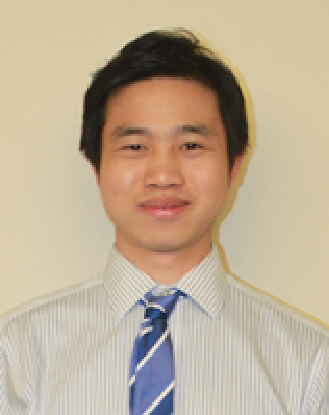}}]
{Jun Zhao} 
(S'10-M'15) received the bachelor’s degree from Shanghai Jiao Tong University, China, in July 2010, and the Ph.D. degree in electrical and computer engineering from Carnegie Mellon University (CMU) in May 2015. He is currently an Assistant Professor with the College of Computing and Data Science (CCDS), Nanyang Technological University (NTU), Singapore.
\end{IEEEbiography}

\begin{IEEEbiography}
[{\includegraphics[width=1in,height=1.25in,clip,keepaspectratio]{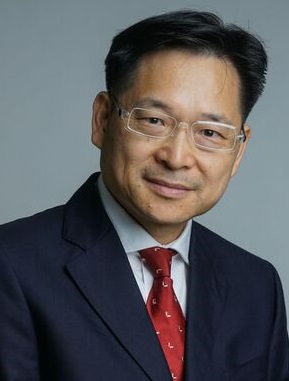}}]
{Kwok-Yan Lam}
(Senior Member, IEEE) received the B.Sc. degree from the University of London, London, U.K., in 1987, and the Ph.D. degree from the University of Cambridge, Cambridge, U.K., in 1990. He is currently an Associate Vice President (Strategy and Partnerships) and a Professor with the College of Computing and Data Science, Nanyang Technological University (NTU), Singapore.
\end{IEEEbiography}

\clearpage
\begin{appendices}
\section{Proof of \textbf{Lemma \ref{lemma_p1top2}}}\label{append_lemma_p1top2}
\begin{proof}
We first define new auxiliary variables $\psi_n^{(u)}$, $\psi_{n,m}^{(t)}$, $\psi_{n,m}^{(a)}$, and $\psi_{n,m}^{(s)}$. Let
\begin{talign}
    &\psi_n^{(u)} \leq \frac{c_n^{(u)}\varphi_n^{(u)}d_n}{cost_n^{(u)}},\\
    &\psi_{n,m}^{(i)} \leq \frac{c_{n,m}^{(i)}\varphi_n^{(i)}d_n}{cost_{n,m}^{(i)}}, i \in \{t,a,s\}.
\end{talign}
Next, we analyze the new constraints by substituting in expressions of $cost_n^{(u)}$, $cost_{n,m}^{(t)}$, $cost_{n,m}^{(a)}$, and $cost_{n,m}^{(s)}$:
\begin{talign}
&cost_n^{(u)} = \omega_t T_n^{(up)} + \omega_e E_n^{(up)}, cost_n^{(u)} \leq \frac{c_n^{(u)}\varphi_n^{(u)}d_n}{\psi_n^{(u)}}, \nonumber \\ 
&\Rightarrow \omega_t T_n^{(up)} + \omega_e E_n^{(up)} \leq \frac{c_n^{(u)}\varphi_n^{(u)}d_n}{\psi_n^{(u)}},\nonumber \\
&\Rightarrow \omega_t T_n^{(up)} + \omega_e e_n t_n\kappa_n\varphi_n^{(u)}d_n(\gamma_n^{(u)})^2 f_n^2 - \frac{c_n^{(u)}\varphi_n^{(u)}d_n}{\psi_n^{(u)}}\leq 0,\\
&cost_{n,m}^{(t)} = \omega_t (T_{n,m}^{(ut)} + T_{n,m}^{(tp)}) + \omega_e (E_{n,m}^{(ut)} + E_{n,m}^{(tp)}), \nonumber \\ &cost_{n,m}^{(t)} \leq \frac{c_{n,m}^{(t)}\varphi_n^{(t)}d_n}{\psi_{n,m}^{(t)}}, \nonumber \\
&\Rightarrow\omega_t (T_{n,m}^{(ut)} + T_{n,m}^{(tp)}) + \omega_e (E_{n,m}^{(ut)} + E_{n,m}^{(tp)}) \leq \frac{c_{n,m}^{(t)}\varphi_n^{(t)}d_n}{\psi_{n,m}^{(t)}}, \nonumber \\
&\Rightarrow\omega_t (T_{n,m}^{(ut)} + T_{n,m}^{(tp)}) + \omega_e (\rho_n^{(u)}p_n \frac{x_{n,m}^{(t)}[\omega_b(1-\varphi_n^{(u)})d_n+d_n^{(l)}]}{r_{n,m_t}},  \nonumber \\&+x_{n,m}^{(t)}\kappa_{m_t}e_{m_t}t_n\varphi_n^{(t)}d_n(\gamma_{n,m}^{(t)}f_{m_t})^2) - \frac{c_{n,m}^{(t)}\varphi_n^{(t)}d_n}{\psi_{n,m}^{(t)}}\leq 0,\\
&cost_{n,m}^{(a)} = \omega_t (T_{n,m}^{(tt)} + T_{n,m}^{(ap)}) + \omega_e (E_{n,m}^{(tt)} + E_{n,m}^{(ap)}), \nonumber \\
&cost_{n,m}^{(a)} \leq \frac{c_{n,m}^{(a)}\varphi_n^{(a)}d_n}{\psi_{n,m}^{(a)}}, \nonumber \\
&\Rightarrow\omega_t (T_{n,m}^{(tt)} + T_{n,m}^{(ap)}) + \omega_e (E_{n,m}^{(tt)} + E_{n,m}^{(ap)}) \leq \frac{c_{n,m}^{(a)}\varphi_n^{(a)}d_n}{\psi_{n,m}^{(a)}}, \nonumber \\
&\Rightarrow\omega_t (T_{n,m}^{(tt)} \!+\! T_{n,m}^{(ap)}) \!+\! \omega_e (\rho_{n,m}^{(t)}p_{m_t}\frac{x_{n,m}^{(a)}[\omega_b(1\!-\!\varphi_n^{(u)}\!-\!\varphi_n^{(t)})d_n \!+\! d_n^{(l)}]}{r_{m_t,m_a}},  \nonumber \\    &+x_{n,m}^{(a)}e_{m_a}\kappa_{m_a}t_n\varphi_n^{(a)}d_n(\gamma_{n,m}^{(a)}f_{m_a})^2) \!-\! \frac{c_{n,m}^{(a)}\varphi_n^{(a)}d_n}{\psi_{n,m}^{(a)}} \leq 0,\\
&cost_{n,m}^{(s)} = \omega_t (T_{n,m}^{(at)} + T_{n,m}^{(sp)}) + \omega_e (E_{n,m}^{(at)} + E_{n,m}^{(sp)}), \nonumber \\
&cost_{n,m}^{(s)} \leq \frac{c_{n,m}^{(s)}\varphi_n^{(s)}d_n}{\psi_{n,m}^{(s)}} \nonumber \\
&\Rightarrow\omega_t (T_{n,m}^{(at)} + T_{n,m}^{(sp)}) + \omega_e (E_{n,m}^{(at)} + E_{n,m}^{(sp)})\leq \frac{c_{n,m}^{(s)}\varphi_n^{(s)}d_n}{\psi_{n,m}^{(s)}}, \nonumber \\
&\Rightarrow\omega_t (T_{n,m}^{(at)} + T_{n,m}^{(sp)})+\omega_e (\frac{[\omega_b(1-\varphi_n^{(u)}-\varphi_n^{(t)}-\varphi_n^{(a)})d_n + d_n^{(l)}]}{r_{m_a,m_s}}x_{n,m}^{(s)} \nonumber \\
&\cdot \rho_{n,m}^{(a)}p_{m_a}+x_{n,m}^{(s)}e_{m_s}\kappa_{m_s}t_n\varphi_n^{(s)}d_n(\gamma_{n,m}^{(s)}f_{m_s})^2)  \nonumber \\
&- \frac{d_n}{\psi_{n,m}^{(s)}} c_{n,m}^{(s)}\varphi_n^{(s)} \leq 0.
\end{talign}
Here, we introduce auxiliary variables $T^{(u)}_n$, $T^{(t)}_{n,m}$, $T^{(a)}_{n,m}$, $T^{(s)}_{n,m}$ to replace the delay formulas $T_n^{(up)}$, $T_{n,m}^{(ut)} + T_{n,m}^{(tp)}$, $T_{n,m}^{(tt)} + T_{n,m}^{(ap)}$, and $T_{n,m}^{(at)} + T_{n,m}^{(sp)}$, respectively. Therefore, we can obtain the following new constraints:
\begin{talign}
&\omega_t T_n^{(u)} \!+\! \omega_e e_n\kappa_n\varphi_n^{(u)}t_n d_n(\gamma_n^{(u)}f_n)^2 \!-\! \frac{c_n^{(u)}\varphi_n^{(u)}d_n}{\psi_n^{(u)}} \!\leq\! 0, \nonumber \\
&\hspace{180pt}\forall n \in \mathcal{N},\label{eq.cost_u_in_proof}
\end{talign}
\begin{talign}
&\omega_t T_{n,m}^{(t)} + \omega_e \big(\rho_n^{(u)}p_n \frac{x_{n,m}^{(t)}[\omega_b(1-\varphi_n^{(u)})d_n + d_n^{(l)}]}{r_{n,m_t}} \nonumber \\
&+ x_{n,m}^{(t)}\kappa_{m_t}\varphi_n^{(t)} e_{m_t}t_n d_n(\gamma_{n,m}^{(t)}f_{m_t})^2\big) - \frac{c_{n,m}^{(t)}\varphi_n^{(t)}d_n}{\psi_{n,m}^{(t)}}\leq 0, \nonumber \\
&\hspace{140pt}\forall n \in \mathcal{N}, \forall m \in \mathcal{M},\label{eq.cost_t_in_proof}
\end{talign}
\begin{talign}
&\omega_t T_{n,m}^{(a)} + \omega_e (\rho_{n,m}^{(t)}p_{m_t}\frac{x_{n,m}^{(a)}[\omega_b(1-\varphi_n^{(u)}-\varphi_n^{(t)})d_n+d_n^{(l)}]}{r_{m_t,m_a}} \nonumber \\
&+ x_{n,m}^{(a)} e_{m_a}\kappa_{m_a}\varphi_n^{(a)}t_n d_n(\gamma_{n,m}^{(a)}f_{m_a})^2) - \frac{c_{n,m}^{(a)}\varphi_n^{(a)}d_n}{\psi_{n,m}^{(a)}} \leq 0,  \nonumber \\
&\hspace{140pt}\forall n \in \mathcal{N}, \forall m \in \mathcal{M}, \label{eq.cost_a_in_proof}
\end{talign}
\begin{talign}
&\omega_t T_{n,m}^{(s)} \!+\! \omega_e (\rho_{n,m}^{(a)}p_{m_a}\frac{x_{n,m}^{(s)}[\omega_b(1-\varphi_n^{(u)}-\varphi_n^{(t)}-\varphi_n^{(a)})d_n + d_n^{(l)}]}{r_{m_a,m_s}} \nonumber \\
&+\!x_{n,m}^{(s)} e_{m_s}\kappa_{m_s}\varphi_n^{(s)}t_n d_n(\gamma_{n,m}^{(s)}f_{m_t})^2) - \frac{c_{n,m}^{(s)}\varphi_n^{(s)}d_n}{\psi_{n,m}^{(s)}} \leq 0, \nonumber \\
&\hspace{140pt}\forall n \in \mathcal{N}, \forall m \in \mathcal{M},\label{eq.cost_s_in_proof}\\
&T_n^{(up)} \leq T_n^{(u)}, \forall n \in \mathcal{N},\\
&T_{n,m}^{(ut)} + T_{n,m}^{(tp)} \leq T_{n,m}^{(t)}, \forall n \in \mathcal{N}, \forall m \in \mathcal{M},\\
&T_{n,m}^{(tt)} + T_{n,m}^{(ap)} \leq T_{n,m}^{(a)}, \forall n \in \mathcal{N}, \forall m \in \mathcal{M},\\
&T_{n,m}^{(at)} + T_{n,m}^{(sp)} \leq T_{n,m}^{(s)}, \forall n \in \mathcal{N}, \forall m \in \mathcal{M}.
\end{talign}
We define $\bm{T^{(u)}} = [T_n^{(u)}]|_{n \in \mathcal{N}}$ and $\bm{\psi^{(u)}}:=[\psi^{(u)}_n]|_{n \in \mathcal{N}}$.
For $i \in \{t,a,s\}$, let $\bm{T^{(i)}} = [T^{(i)}_{n,m}]|_{n \in \mathcal{N},m \in \mathcal{M}^{(i)}}$, $\bm{\psi^{(i)}}:=[\psi^{(i)}_{n,m}]|_{n \in \mathcal{N}, m \in \mathcal{M}^{(i)}}$, $\bm{T}:=\{\bm{T^{(u)}},\bm{T^{(t)}},\bm{T^{(a)}},\bm{T^{(s)}}\}$, and $\bm{\psi}:=\{\bm{\psi^{(u)}},\bm{\psi^{(t)}},\bm{\psi^{(a)}},\bm{\psi^{(s)}}\}$.
To express the new constraints on the optimization problem clearer to read, we define functions $\varpi^{(u)}_n$, $\varpi^{(t)}_{n,m}$, $\varpi^{(a)}_{n,m}$, and $\varpi^{(s)}_{n,m}$ according to Equations~(\ref{defvarpiu})~(\ref{defvarpit})~(\ref{defvarpia})~(\ref{defvarpis}) given in the statement of Lemma \ref{lemma_p1top2}.
Therefore, the constraints (\ref{eq.cost_u_in_proof}), (\ref{eq.cost_t_in_proof}), (\ref{eq.cost_a_in_proof}), (\ref{eq.cost_s_in_proof}) would be $\varpi^{(u)}_n \leq 0$, $\varpi^{(t)}_{n,m} \leq 0$, $\varpi^{(a)}_{n,m} \leq 0$, and $\varpi^{(s)}_{n,m} \leq 0$, respectively.
Based on the above discussion, the Problem $\mathbb{P}_1$ can be transformed into the Problem $\mathbb{P}_{2}$.

\textbf{Lemma \ref{lemma_p1top2}} is proven.
\end{proof}

\section{Proof of \textbf{Lemma \ref{lemma_p2top3}}}\label{append_lemma_p2top3}
\begin{proof}
We analyze part of the KKT condition of Problem $\mathbb{P}_{2}$ to facilitate our subsequent discussion. Given the non-negative multipliers $\alpha_n^{(u)}$, $\alpha_{n,m}^{(t)}$, $\alpha_{n,m}^{(a)}$, and $\alpha_{n,m}^{(s)}$ in \textbf{Lemma 2},
the Lagrangian function is given as follows:
\begin{talign}
&L_{\mathbb{P}_{2}}(\bm{x},\bm{\varphi},\bm{\gamma},\bm{\phi},\bm{\rho},\bm{\psi},\bm{T},\bm{\alpha})=\nonumber \\
&- \sum_{n \in \mathcal{N}}\psi_n^{(u)}-\sum_{n \in \mathcal{N}}\sum_{m \in \mathcal{M}} \psi_{n,m}^{(t)} + \psi_{n,m}^{(a)} + \psi_{n,m}^{(s)} \nonumber \\
&+\sum_{n\in \mathcal{N}}\alpha_n^{(u)}\cdot[\psi_n^{(u)}cost_n^{(u)}-c_n^{(u)}\varphi_n^{(u)}d_n] \nonumber \\
&+\sum_{n\in \mathcal{N},m\in \mathcal{M}^{(t)}}\alpha_{n,m}^{(t)}\cdot(\psi_{n,m}^{(t)}cost_{n,m}^{(t)}-c_{n,m}^{(t)}\varphi_n^{(t)}d_n) \nonumber \\
&+\sum_{n\in \mathcal{N},m\in \mathcal{M}^{(a)}}\alpha_{n,m}^{(a)}\cdot(\psi_{n,m}^{(a)}cost_{n,m}^{(a)}-c_{n,m}^{(a)}\varphi_n^{(a)}d_n) \nonumber \\
&+\sum_{n\in \mathcal{N},m\in \mathcal{M}^{(s)}}\alpha_{n,m}^{(s)}\cdot(\psi_{n,m}^{(s)}cost_{n,m}^{(s)}-c_{n,m}^{(s)}\varphi_n^{(s)}d_n) \nonumber \\
&+ \hat{L}_{\mathbb{P}_{2}},
\end{talign}
where $\hat{L}_{\mathbb{P}_{2}}$ is the remaining Lagrangian terms that we don’t
care about. Next, we analyze some stationarity and complementary slackness properties of $L_{\mathbb{P}_{2}}$.
\newline
\textbf{Stationarity:}
\begin{talign}
    &\frac{\partial L_{\mathbb{P}_{2}}}{\partial \psi_n^{(u)}} = -1 + \alpha^{(u)}_n cost_n^{(u)}=0, \forall n \in \mathcal{N},\\
    &\frac{\partial L_{\mathbb{P}_{2}}}{\partial \psi_{n,m}^{(t)}} = -1 + \alpha^{(t)}_{n,m} cost_{n,m}^{(t)}=0, \forall n \in \mathcal{N}, m \in \mathcal{M},\\
    &\frac{\partial L_{\mathbb{P}_{2}}}{\partial \psi_{n,m}^{(a)}} = -1 + \alpha^{(a)}_{n,m} cost_{n,m}^{(a)}=0, \forall n \in \mathcal{N}, m \in \mathcal{M},\\
    &\frac{\partial L_{\mathbb{P}_{2}}}{\partial \psi_{n,m}^{(s)}} = -1 + \alpha^{(s)}_{n,m} cost_{n,m}^{(s)}=0, \forall n \in \mathcal{N}, m \in \mathcal{M}.
\end{talign}
\textbf{Complementary slackness:}
\begin{talign}
    &\alpha_n^{(u)}\cdot[\psi_n^{(u)}cost_n^{(u)}-c_n^{(u)}\varphi_n^{(u)}d_n]=0,\forall n \in \mathcal{N},\\
    &\alpha_{n,m}^{(t)}\cdot(\psi_{n,m}^{(t)}cost_{n,m}^{(t)}-c_{n,m}^{(t)}\varphi_n^{(t)}d_n)=0,\forall n \in \mathcal{N}, m \in \mathcal{M},\\
    &\alpha_{n,m}^{(a)}\cdot(\psi_{n,m}^{(a)}cost_{n,m}^{(a)}-c_{n,m}^{(a)}\varphi_n^{(a)}d_n)=0,\forall n \in \mathcal{N}, m \in \mathcal{M},\\
    &\alpha_{n,m}^{(s)}\cdot(\psi_{n,m}^{(s)}cost_{n,m}^{(s)}-c_{n,m}^{(s)}\varphi_n^{(s)}d_n)=0,\forall n \in \mathcal{N}, m \in \mathcal{M}.
\end{talign}
Therefore, for KKT points of Problem $\mathbb{P}_{2}$, we can obtain the conclusions, i.e., Eq. (\ref{eq_psi_u}), (\ref{eq_psi_tas}), (\ref{eq_alpha_u}), and (\ref{eq_alpha_tas}).
Based on the above discussion, Problem $\mathbb{P}_{2}$ can be transformed into a new Problem $\mathbb{P}_{3}$ \cite{10.1145/3565287.3610271}.

\textbf{Lemma \ref{lemma_p2top3}} is proven.
\end{proof}

\section{Proof of \textbf{Lemma \ref{lemma_p4top5}}}\label{append_lemma_p4top5}
\begin{proof}
In the term $\frac{\text{power}}{\text{transmission data rate}}$ included in ``$cost$'', the ``power'' part is an affine function of $\rho$, and the ``transmission data rate'' part is a joint concave function of $\phi$ and $\rho$. Therefore, this term is actually $\frac{\text{affine function}}{\text{concave function}}$, which is general non-convex and NP-hard. Since there are other polynomial functions in ``$cost$'', the technique proposed in \cite{jong2012efficient} can't be applied in this case. Thanks to the recent findings in \cite{10368052}, we can efficiently transform this ``$cost$'' term into a convex term. We will present how to solve it.

To convexify the expressions $cost_{n,m}^{(t)}$, $cost_{n,m}^{(a)}$, and $cost_{n,m}^{(s)}$, we introduce auxiliary variables $\varrho_{n,m}^{(t)}$, $\varrho_{n,m}^{(a)}$, and $\varrho_{n,m}^{(s)}$ as shown in \textbf{Lemma 3}.
These original cost terms are then reformulated into their convex equivalents $\widetilde{cost}_{n,m}^{(t)}$, $\widetilde{cost}_{n,m}^{(a)}$, and $\widetilde{cost}_{n,m}^{(s)}$, whose expressions are also provided in \mbox{\textbf{Lemma 3}}.
Those new terms are all convex when we fix $\varrho_{n,m}^{(t)}$, $\varrho_{n,m}^{(a)}$, and $\varrho_{n,m}^{(s)}$. Let 
\begin{talign}
&\chi(\rho_n^{(u)}) = \rho_n^{(u)}p_n x_{n,m}^{(t)}[\omega_b(1-\varphi_n^{(u)})d_n+d_n^{(l)}],\\
&\varsigma (\phi_{n,m}^{(t)}, \rho_n^{(u)}) = r_{n,m_t},
\end{talign}
where $r_{n,m_t} = \phi_{n,m}^{(t)}b_{n,m_t} \log_2 (1 + \frac{\rho_n^{(u)} p_n g_{n,m_t}}{\sigma^2 \phi_{n,m}^{(t)}b_{n,m_t}})$. It's easy to know that $\chi(\rho_n^{(u)})$ is convex of $\rho_n^{(u)}$ and $\varsigma (\phi_{n,m}^{(t)}, \rho_n^{(u)})$ is jointly concave of $(\phi_{n,m}^{(t)}, \rho_n^{(u)})$. Let
\begin{talign}
&\mathcal{F}(\phi_{n,m}^{(t)}, \rho_n^{(u)}, \gamma_{n,m}^{(t)}, \psi_n^{(t)}, T_{n,m}^{(t)})  \nonumber \\
&= \omega_t T_{n,m}^{(t)} \nonumber \\
&+ \omega_e (\frac{\chi(\rho_n^{(u)})}{\varsigma (\phi_{n,m}^{(t)}, \rho_n^{(u)})} +x_{n,m}^{(t)}e_{m_t}\kappa_{m_t}t_n\varphi_n^{(t)}d_n (\gamma_{n,m}^{(t)}f_{m_t})^2).
\end{talign}
Let 
\begin{talign}
&\mathcal{G}(\phi_{n,m}^{(t)}, \rho_n^{(u)}, \gamma_{n,m}^{(t)}, \psi_n^{(t)}, T_{n,m}^{(t)})  \nonumber \\
&= \omega_t T_{n,m}^{(t)} + \omega_e \{\chi(\rho_n^{(u)})^2\varrho_{n,m}^{(t)} + \frac{1}{4\varsigma (\phi_{n,m}^{(t)}, \rho_n^{(u)})^2\varrho_{n,m}^{(t)}}\nonumber \\
&+ x_{n,m}^{(t)}e_{m_t}\kappa_{m_t}t_n \varphi_n^{(t)}d_n(\gamma_{n,m}^{(t)}f_{m_t})^2\}.
\end{talign}
The partial derivative of $T_{n,m}^{(t)}$ is
\begin{talign}
\frac{\partial(\mathcal{F}(\phi_{n,m}^{(t)}, \rho_n^{(u)}, \gamma_{n,m}^{(t)}, \psi_n^{(t)}, T_{n,m}^{(t)}))}{\partial T^{(t)}} =\omega_t, \\
\frac{\partial(\mathcal{G}(\phi_{n,m}^{(t)}, \rho_n^{(u)}, \gamma_{n,m}^{(t)}, \psi_n^{(t)}, T_{n,m}^{(t)}))}{\partial T^{(t)}} =\omega_t.
\end{talign}
The partial derivative of $\gamma_{n,m}^{(t)}$ is given as
\begin{talign}
&\frac{\partial(\mathcal{F}(\phi_{n,m}^{(t)}, \rho_n^{(u)}, \gamma_{n,m}^{(t)}, \psi_n^{(t)}, T_{n,m}^{(t)}))}{\partial \gamma_{n,m}^{(t)}} \nonumber\\
&= 2\omega_ex_{n,m}^{(t)}e_{m_t}\kappa_{m_t}t_n \varphi_n^{(t)}d_n f_{m_t}^2\gamma_{n,m}^{(t)},\\
&\frac{\partial(\mathcal{G}(\phi_{n,m}^{(t)}, \rho_n^{(u)}, \gamma_{n,m}^{(t)}, \psi_n^{(t)}, T_{n,m}^{(t)}))}{\partial \gamma_{n,m}^{(t)}} \nonumber \\
&= 2\omega_e x_{n,m}^{(t)}e_{m_t}\kappa_{m_t}t_n \varphi_n^{(t)}d_n f_{m_t}^2\gamma_{n,m}^{(t)}.
\end{talign}
We get the partial derivative of $\phi_{n,m}^{(t)}$ as
\begin{talign}
&\frac{\partial(\mathcal{F}(\phi_{n,m}^{(t)}, \rho_n^{(u)}, \gamma_{n,m}^{(t)}, \psi_n^{(t)}, T_{n,m}^{(t)}))}{\partial \phi_{n,m}^{(t)}} = -\frac{\omega_e\chi(\rho_n^{(u)})}{\varsigma (\phi_{n,m}^{(t)}, \rho_n^{(u)})^2}\frac{\partial\varsigma (\phi_{n,m}^{(t)}, \rho_n^{(u)})}{\partial \phi_{n,m}^{(t)}},
\end{talign}
\begin{talign}
&\frac{\partial(\mathcal{G}(\phi_{n,m}^{(t)}, \rho_n^{(u)}, \gamma_{n,m}^{(t)}, \psi_n^{(t)}, T_{n,m}^{(t)}))}{\partial \phi_{n,m}^{(t)}}  \nonumber \\
&=-\frac{\omega_e}{2\varrho_{n,m_t}^{(t)}\varsigma (\phi_{n,m}^{(t)}, \rho_n^{(u)})^3}\frac{\partial \varsigma (\phi_{n,m}^{(t)}, \rho_n^{(u)})}{\partial \phi_{n,m}^{(t)}}.
\end{talign}
When $\varrho_{n,m_t}^{(t)} = \frac{1}{2\chi(\rho_n^{(u)}) \varsigma (\phi_{n,m}^{(t)}, \rho_n^{(u)})}$, we can obtain the following conclusion that
\begin{talign}
&\frac{\partial(\mathcal{F}(\phi_{n,m}^{(t)}, \rho_n^{(u)}, \gamma_{n,m}^{(t)}, \psi_n^{(t)}, T_{n,m}^{(t)}))}{\partial \phi_{n,m}^{(t)}} \nonumber \\
&= \frac{\partial(\mathcal{G}(\phi_{n,m}^{(t)}, \rho_n^{(u)}, \gamma_{n,m}^{(t)}, \psi_n^{(t)}, T_{n,m}^{(t)}))}{\partial \phi_{n,m}^{(t)}}.
\end{talign}
The partial derivative of $\rho_n^{(u)}$ is
\begin{talign}
    &\frac{\partial(\mathcal{F}(\phi_{n,m}^{(t)}, \rho_n^{(u)}, \gamma_{n,m}^{(t)}, \psi_n^{(t)}, T_{n,m}^{(t)}))}{\partial \rho_n^{(u)}} \nonumber \\
    &= \omega_e \frac{\frac{\partial \chi(\rho_n^{(u)})}{\partial \rho_n^{(u)}}\varsigma (\phi_{n,m}^{(t)}, \rho_n^{(u)}) - \chi(\rho_n^{(u)})\frac{\partial \varsigma (\phi_{n,m}^{(t)}, \rho_n^{(u)})}{\partial \rho_n^{(u)}}}{\varsigma (\phi_{n,m}^{(t)}, \rho_n^{(u)})^2},\\
    &\frac{\partial(\mathcal{G}(\phi_{n,m}^{(t)}, \rho_n^{(u)}, \gamma_{n,m}^{(t)}, \psi_n^{(t)}, T_{n,m}^{(t)}))}{\partial \rho_n^{(u)}} \nonumber \\
    &= \omega_e (2\varrho_{n,m_t}^{(t)}\chi(\rho_n^{(u)})\frac{\partial \chi(\rho_n^{(u)})}{\partial \rho_n^{(u)}} \nonumber \\
    &\quad - \frac{1}{2\varrho_{n,m_t}^{(t)} \varsigma (\phi_{n,m}^{(t)}, \rho_n^{(u)})^3} \frac{\partial \varsigma (\phi_{n,m}^{(t)}, \rho_n^{(u)})}{\partial \rho_n^{(u)}}).
\end{talign}
When $\varrho_{n,m}^{(t)} = \frac{1}{2\chi(\rho_n^{(u)}) \varsigma (\phi_{n,m}^{(t)}, \rho_n^{(u)})}$, we know
\begin{talign}
&\frac{\partial(\mathcal{F}(\phi_{n,m}^{(t)}, \rho_n^{(u)}, \gamma_{n,m}^{(t)}, \psi_n^{(t)}, T_{n,m}^{(t)}))}{\partial \rho_n^{(u)}}\nonumber \\
&= \frac{\partial(\mathcal{G}(\phi_{n,m}^{(t)}, \rho_n^{(u)}, \gamma_{n,m}^{(t)}, \psi_n^{(t)}, T_{n,m}^{(t)}))}{\partial \rho_n^{(u)}}.
\end{talign}
Based on the above discussion, we can obtain that
\begin{talign}
&\frac{\partial(\mathcal{F}(\phi_{n,m}^{(t)},\rho_n^{(u)},\gamma_{n,m}^{(t)},\psi_n^{(t)},T_{n,m}^{(t)}))}{\partial (\phi_{n,m}^{(t)},\rho_n^{(u)},\gamma_{n,m}^{(t)},\psi_n^{(t)},T_{n,m}^{(t)})} \nonumber \\
&= \frac{\partial(\mathcal{G}(\phi_{n,m}^{(t)},\rho_n^{(u)},\gamma_{n,m}^{(t)},\psi_n^{(t)},T_{n,m}^{(t)}))}{\partial (\phi_{n,m}^{(t)},\rho_n^{(u)},\gamma_{n,m}^{(t)},\psi_n^{(t)},T_{n,m}^{(t)})},
\end{talign}
\begin{talign}
&\mathcal{F}(\phi_{n,m}^{(t)},\rho_n^{(u)},\gamma_{n,m}^{(t)},\psi_n^{(t)},T_{n,m}^{(t)}) \nonumber \\
&= \mathcal{G}(\phi_{n,m}^{(t)},\rho_n^{(u)},\gamma_{n,m}^{(t)},\psi_n^{(t)},T_{n,m}^{(t)}).
\end{talign}
The equivalence of the remaining two pairs of terms, $cost_{n,m}^{(a)}$ and 
$\widetilde{cost}_{n,m}^{(a)}$, $cost_{n,m}^{(s)}$ and 
$\widetilde{cost}_{n,m}^{(s)}$, can be proved by the same steps, which are not detailed here. Let $\bm{\varrho}:=\{\varrho_{n,m}^{(t)}$, $\varrho_{n,m}^{(a)}$, and $\varrho_{n,m}^{(s)}\}$. Based on the above discussion, the Problem $\mathbb{P}_{4}$ is equivalent to the Problem $\mathbb{P}_{5}$.

\textbf{Lemma \ref{lemma_p4top5}} is proven.
\end{proof}

\section{Proof of \textbf{Lemma \ref{lemma_p7top8}}}\label{append_lemma_p7top8}
\begin{proof}
To transform Problem $\mathbb{P}_{7}$ to Problem $\mathbb{P}_{8}$, we first analyze the term $\alpha_{n,m}^{(t)}(c_{n,m}^{(t)}\varphi_n^{(t)}d_n - \psi_{n,m}^{(t)}cost_{n,m}^{(t)})$ as follow:
\begin{talign}
&\alpha_{n,m}^{(t)}(c_{n,m}^{(t)}\varphi_n^{(t)}d_n - \psi_{n,m}^{(t)}cost_{n,m}^{(t)}) \nonumber \\
&= \alpha_{n,m}^{(t)}\{c_{n,m}^{(t)}\varphi_n^{(t)}d_n - \psi_{n,m}^{(t)}[\omega_t T_{n,m}^{(t)} + \omega_e (\frac{[\omega_b(1-\varphi_n^{(u)})d_n + d_n^{(l)}]}{r_{n,m_t}}  \nonumber \\
&\cdot \rho_n^{(u)}p_n x_{n,m}^{(t)}+x_{n,m}^{(t)}e_{m_t}\kappa_{m_t}t_n\varphi_n^{(t)}d_n(\gamma_{n,m}^{(t)})^2f_{m_t}^2)]\}\nonumber \\
&=\alpha_{n,m}^{(t)}c_{n,m}^{(t)}\varphi_n^{(t)}d_n - \alpha_{n,m}^{(t)}\psi_{n,m}^{(t)}\omega_t T_{n,m}^{(t)} - \alpha_{n,m}^{(t)}\varphi_n^{(t)} \omega_e \rho_n^{(u)}p_n \nonumber \\
&\cdot x_{n,m}^{(t)}\frac{[\omega_b(1-\varphi_n^{(u)})d_n+d_n^{(l)}]}{r_{n,m_t}} - \alpha_{n,m}^{(t)}\varphi_n^{(t)} \omega_e x_{n,m}^{(t)}e_{m_t}\kappa_{m_t}t_n\varphi_n^{(t)}d_n\nonumber \\
&\cdot (\gamma_{n,m}^{(t)})^2 f_{m_t}^2\nonumber \\
&=- \alpha_{n,m}^{(t)}\psi_{n,m}^{(t)}\omega_t T_{n,m}^{(t)} + \alpha_{n,m}^{(t)}d_n c_{n,m}^{(t)}\varphi_n^{(t)} \!-\!\big( \alpha_{n,m}^{(t)}\psi_{n,m}^{(t)} \omega_e \rho_n^{(u)}\nonumber \\
&\cdot\frac{p_n(\omega_b d_n+d_n^{(l)})}{r_{n,m_t}} 
\big) x_{n,m}^{(t)}\!+\! \alpha_{n,m}^{(t)} \psi_{n,m}^{(t)}\omega_e \rho_n^{(u)}p_n \frac{\omega_b d_n}{r_{n,m_t}} x_{n,m}^{(t)}\varphi_n^{(u)}\nonumber \\
&-\alpha_{n,m}^{(t)}\psi_{n,m}^{(t)} \omega_e e_{m_t}\kappa_{m_t} t_n d_n(\gamma_{n,m}^{(t)})^2 f_{m_t}^2 x_{n,m}^{(t)} \varphi_n^{(t)}.
\end{talign}
To make the expression clearer and easier to understand, we define the following auxiliary variables:
\begin{talign}
&A_{n,m}^{(tt)} := -\alpha_{n,m}^{(t)}\psi_{n,m}^{(t)}\omega_e e_{m_t}\kappa_{m_t} t_n d_n(\gamma_{n,m}^{(t)})^2f_{m_t}^2,\\
&A_{n,m}^{(tu)}:=\alpha_{n,m}^{(t)}\psi_{n,m}^{(t)} \omega_e \rho_n^{(u)}p_n \frac{\omega_b d_n}{r_{n,m_t}},\\
&B_{n,m}^{(t)}:=- \alpha_{n,m}^{(t)}\psi_{n,m}^{(t)} \omega_e \rho_n^{(u)}p_n\frac{\omega_b d_n+d_n^{(l)}}{r_{n,m_t}},\\ 
&C_{n,m}^{(t)}:=\alpha_{n,m}^{(t)}c_{n,m}^{(t)}d_n,\\
&D_{n,m}^{(t)}:=- \alpha_{n,m}^{(t)}\psi_{n,m}^{(t)}\omega_t.
\end{talign}
Based on the predefined auxiliary variables, we can rewrite the term $\alpha_{n,m}^{(t)}(c_{n,m}^{(t)}\varphi_n^{(t)}d_n - \psi_{n,m}^{(t)}cost_{n,m}^{(t)})$ more clearly as
\begin{talign}
&\alpha_{n,m}^{(t)}(c_{n,m}^{(t)}\varphi_n^{(t)}d_n - \psi_{n,m}^{(t)}cost_{n,m}^{(t)})\nonumber \\
&=A_{n,m}^{(tt)}x_{n,m}^{(t)} \varphi_n^{(t)} + A_{n,m}^{(tu)} x_{n,m}^{(t)}\varphi_n^{(u)} + B_{n,m}^{(t)}x_{n,m}^{(t)} + C_{n,m}^{(t)}\varphi_n^{(t)} \nonumber \\
&+ D_{n,m}^{(t)} T_{n,m}^{(t)}.\\
&\alpha_{n,m}^{(a)}(c_{n,m}^{(a)}\varphi_n^{(a)}d_n - \psi_{n,m}^{(a)}cost_{n,m}^{(a)}) \nonumber \\
&=- \alpha_{n,m}^{(a)} \psi_{n,m}^{(a)} \omega_t T_{n,m}^{(a)} + \alpha_{n,m}^{(a)}c_{n,m}^{(a)}d_n \varphi_n^{(a)}  - \alpha_{n,m}^{(a)} \psi_{n,m}^{(a)} \omega_e \nonumber \\
&\cdot \!\rho_{n,m}^{(t)}p_{m_t}\frac{\omega_b d_n+d_n^{(l)}}{r_{m_t,m_a}} 
x_{n,m}^{(a)} +\! \alpha_{n,m}^{(a)} \psi_{n,m}^{(a)} \omega_e \rho_{n,m}^{(t)}p_{m_t} \frac{\omega_b d_n}{r_{m_t,m_a}}x_{n,m}^{(a)} \nonumber \\
&\cdot \varphi_n^{(u)} \!+\! \alpha_{n,m}^{(a)} \psi_{n,m}^{(a)} \omega_e \rho_{n,m}^{(t)} p_{m_t} \frac{\omega_b d_n}{r_{m_t,m_a}}x_{n,m}^{(a)} \varphi_n^{(t)}\! -\! \alpha_{n,m}^{(a)} \psi_{n,m}^{(a)} \omega_e \nonumber \\
&\cdot \kappa_{m_a} e_{m_a}t_n d_n(\gamma_{n,m}^{(a)})^2f_{m_a}^2 x_{n,m}^{(a)}\varphi_n^{(a)}.
\end{talign}
For the term $\alpha_{n,m}^{(a)}(c_{n,m}^{(a)}\varphi_n^{(a)}d_n - \psi_{n,m}^{(a)}cost_{n,m}^{(a)})$, we also define the following auxiliary variables:
\begin{talign}
&A_{n,m}^{(aa)}:=- \alpha_{n,m}^{(a)} \psi_{n,m}^{(a)} \omega_e e_{m_a}\kappa_{m_a}t_n d_n(\gamma_{n,m}^{(a)})^2f_{m_a}^2,\\
&A_{n,m}^{(au)}:=\alpha_{n,m}^{(a)} \psi_{n,m}^{(a)} \omega_e \rho_{n,m}^{(t)}p_{m_t} \frac{\omega_b d_n}{r_{m_t,m_a}},\\
&A_{n,m}^{(at)}:=\alpha_{n,m}^{(a)} \psi_{n,m}^{(a)} \omega_e \rho_{n,m}^{(t)}p_{m_t} \frac{\omega_b d_n}{r_{m_t,m_a}},\\
&B_{n,m}^{(a)}:=- \alpha_{n,m}^{(a)} \psi_{n,m}^{(a)} \omega_e \rho_{n,m}^{(t)}p_{m_t}\frac{\omega_b d_n+d_n^{(l)}}{r_{m_t,m_a}},\\
&C_{n,m}^{(a)}:= \alpha_{n,m}^{(a)}c_{n,m}^{(a)}d_n,\\
&D_{n,m}^{(a)}:=- \alpha_{n,m}^{(a)} \psi_{n,m}^{(a)} \omega_t.
\end{talign}
Therefore, the term $\alpha_{n,m}^{(a)}(c_{n,m}^{(a)}\varphi_n^{(a)}d_n - \psi_{n,m}^{(a)}cost_{n,m}^{(a)})$ can be rewrite as
\begin{talign}
&\alpha_{n,m}^{(a)}(c_{n,m}^{(a)}\varphi_n^{(a)}d_n - \psi_{n,m}^{(a)}cost_{n,m}^{(a)})\nonumber \\
&=A_{n,m}^{(aa)} x_{n,m}^{(a)}\varphi_n^{(a)} + A_{n,m}^{(au)} x_{n,m}^{(a)} \varphi_n^{(u)} + A_{n,m}^{(at)}x_{n,m}^{(a)} \varphi_n^{(t)} \nonumber \\
&+ B_{n,m}^{(a)} x_{n,m}^{(a)} + C_{n,m}^{(a)} \varphi_n^{(a)} + D_{n,m}^{(a)} T_{n,m}^{(a)}.
\end{talign}

Let's analyze the term $\alpha_{n,m}^{(s)}(c_{n,m}^{(s)}\varphi_n^{(s)}d_n - \psi_{n,m}^{(s)}cost_{n,m}^{(s)})$ by plugging in the expression of $cost_{n,m}^{(s)}$:
\begin{talign}
&\alpha_{n,m}^{(s)}(c_{n,m}^{(s)}\varphi_n^{(s)}d_n - \psi_{n,m}^{(s)}cost_{n,m}^{(s)}) \nonumber \\
&=- \alpha_{n,m}^{(s)} \psi_{n,m}^{(s)} \omega_t T_{n,m}^{(s)} + \alpha_{n,m}^{(s)}c_{n,m}^{(s)} d_n \varphi_n^{(s)}  - \alpha_{n,m}^{(s)} \psi_{n,m}^{(s)}\omega_e \nonumber \\
&\cdot \rho_{n,m}^{(a)}p_{m_a}\frac{\omega_b d_n+ d_n^{(l)}}{r_{m_a,m_s}} 
x_{n,m}^{(s)} + \alpha_{n,m}^{(s)} \psi_{n,m}^{(s)}\omega_e \rho_{n,m}^{(a)} p_{m_a} \frac{\omega_b d_n}{r_{m_a,m_s}}\nonumber \\
&\cdot x_{n,m}^{(s)} \varphi_n^{(u)} \!\!+\! \alpha_{n,m}^{(s)} \psi_{n,m}^{(s)}\omega_e \rho_{n,m}^{(a)} \frac{p_{m_a} \omega_b d_n}{r_{m_a,m_s}} x_{n,m}^{(s)} \varphi_n^{(t)}\!+\! \alpha_{n,m}^{(s)} \psi_{n,m}^{(s)}\nonumber \\
& \cdot\omega_e  \frac{\rho_{n,m}^{(a)} p_{m_a} \omega_b d_n}{r_{m_a,m_s}} x_{n,m}^{(s)} \varphi_n^{(a)}\!-\!\alpha_{n,m}^{(s)} \psi_{n,m}^{(s)} \omega_e e_{m_s}\kappa_{m_s}t_n d_n f_{m_t}^2\nonumber \\
&\cdot (\gamma_{n,m}^{(s)})^2 x_{n,m}^{(s)} \varphi_n^{(s)}.
\end{talign}
We also define the following auxiliary variables to make the above expression clearer:
\begin{talign}
&A_{n,m}^{(ss)}:=-\alpha_{n,m}^{(s)} \psi_{n,m}^{(s)}\omega_e e_{m_s}t_n\kappa_{m_s}d_n(\gamma_{n,m}^{(s)})^2f_{m_t}^2,\\
&A_{n,m}^{(su)}:=\alpha_{n,m}^{(s)} \psi_{n,m}^{(s)}\omega_e \rho_{n,m}^{(a)} p_{m_a} \frac{\omega_b d_n}{r_{m_a,m_s}},\\
&A_{n,m}^{(st)}:=\alpha_{n,m}^{(s)} \psi_{n,m}^{(s)}\omega_e \rho_{n,m}^{(a)} p_{m_a} \frac{\omega_b d_n}{r_{m_a,m_s}},\\
&A_{n,m}^{(sa)}:=\alpha_{n,m}^{(s)} \psi_{n,m}^{(s)}\omega_e \rho_{n,m}^{(a)}p_{m_a} \frac{\omega_b d_n}{r_{m_a,m_s}},\\
&B_{n,m}^{(s)}:=-\alpha_{n,m}^{(s)} \psi_{n,m}^{(s)}\omega_e \rho_{n,m}^{(a)} p_{m_a}\frac{\omega_b d_n+ d_n^{(l)}}{r_{m_a,m_s}},\\
&C_{n,m}^{(s)}:=\alpha_{n,m}^{(s)} c_{n,m}^{(s)}d_n,\\
&D_{n,m}^{(s)}:=- \alpha_{n,m}^{(s)} \psi_{n,m}^{(s)} \omega_t.
\end{talign}
With the defined auxiliary variables, we can rewrite the term $\alpha_{n,m}^{(s)}(c_{n,m}^{(s)}\varphi_n^{(s)}d_n - \psi_{n,m}^{(s)}cost_{n,m}^{(s)})$ as
\begin{talign}
&\alpha_{n,m}^{(s)}(c_{n,m}^{(s)}\varphi_n^{(s)}d_n - \psi_{n,m}^{(s)}cost_{n,m}^{(s)})\nonumber \\
&=A_{n,m}^{(ss)}x_{n,m}^{(s)}\varphi_n^{(s)} + A_{n,m}^{(su)} x_{n,m}^{(s)} \varphi_n^{(u)} + A_{n,m}^{(st)} x_{n,m}^{(s)} \varphi_n^{(t)} \nonumber\\
&+ A_{n,m}^{(sa)} x_{n,m}^{(s)} \varphi_n^{(a)} + B_{n,m}^{(s)} x_{n,m}^{(s)} + C_{n,m}^{(s)} \varphi_n^{(s)} + D_{n,m}^{(s)} T_{n,m}^{(s)}.
\end{talign}
For the term $\alpha_n^{(u)}(c_n^{(u)}\varphi_n^{(u)}d_n - \psi_n^{(u)}cost_n^{(u)})$,
\begin{talign}
&\alpha_n^{(u)}(c_n^{(u)}\varphi_n^{(u)}d_n - \psi_n^{(u)}cost_n^{(u)})\nonumber \\
&=- \alpha_n^{(u)}\psi_n^{(u)}\omega_t T_{n}^{(u)} + [\alpha_n^{(u)}c_n^{(u)}d_n \!-\! \alpha_n^{(u)}\psi_n^{(u)}\omega_e e_n \kappa_n t_n d_n\nonumber \\
&\cdot(\gamma_n^{(u)})^2 f_n^2]\varphi_n^{(u)},
\end{talign}
we define the following auxiliary variables:
\begin{talign}
&C_{n}^{(u)}:=\alpha_n^{(u)}c_n^{(u)}d_n - \alpha_n^{(u)}\psi_n^{(u)}e_n \omega_e \kappa_n t_n d_n(\gamma_n^{(u)})^2 f_n^2,\\
&D_{n}^{(u)}:=- \alpha_n^{(u)}\psi_n^{(u)}\omega_t.
\end{talign}
Therefore, we can know that 
\begin{talign}
&\alpha_n^{(u)}(\varphi_n^{(u)}c_n^{(ut)}d_n - \psi_n^{(u)}cost_n^{(u)})\nonumber \\
&=C_{n}^{(u)}\varphi_n^{(u)} + D_{n}^{(u)} T_{n}^{(u)}.
\end{talign}

Based on the above discussion, the objective function of Problem $\mathbb{P}_{7}$ can be rewritten as
\begin{talign}
&\sum_{n \in \mathcal{N}}\sum_{m \in \mathcal{M}} \alpha_{n,m}^{(t)}(c_{n,m}^{(t)}\varphi_n^{(t)}d_n - \psi_{n,m}^{(t)}cost_{n,m}^{(t)}) +  \alpha_{n,m}^{(a)} \nonumber \\
&\cdot (c_{n,m}^{(a)}\varphi_n^{(a)}d_n- \!\psi_{n,m}^{(a)}cost_{n,m}^{(a)}) \!+\!  \alpha_{n,m}^{(s)}(c_{n,m}^{(s)}\varphi_n^{(s)}d_n \!-\! \psi_{n,m}^{(s)}\nonumber \\
&\cdot cost_{n,m}^{(s)})\!+\! \sum_{n \in \mathcal{N}} \alpha_n^{(u)}(c_n^{(u)}\varphi_n^{(u)}d_n - \psi_n^{(u)}cost_n^{(u)})\nonumber\\
&=\sum_{n \in \mathcal{N}}\sum_{m \in \mathcal{M}} A_{n,m}^{(tt)}x_{n,m}^{(t)} \varphi_n^{(t)} + A_{n,m}^{(tu)} x_{n,m}^{(t)}\varphi_n^{(u)} + B_{n,m}^{(t)} \nonumber\\
&\cdot x_{n,m}^{(t)}+ C_{n,m}^{(t)}\varphi_n^{(t)} + D_{n,m}^{(t)} T_{n,m}^{(t)} + A_{n,m}^{(aa)} x_{n,m}^{(a)}\varphi_n^{(a)} + A_{n,m}^{(au)}  \nonumber\\
&\cdot x_{n,m}^{(a)} \varphi_n^{(u)}\!+\! A_{n,m}^{(at)}x_{n,m}^{(a)} \varphi_n^{(t)} \!+\! B_{n,m}^{(a)} x_{n,m}^{(a)} \!+\! C_{n,m}^{(a)} \varphi_n^{(a)} \!+\! D_{n,m}^{(a)}  \nonumber \\
&\cdot T_{n,m}^{(a)}+ A_{n,m}^{(ss)}x_{n,m}^{(s)}\varphi_n^{(s)} + A_{n,m}^{(su)} x_{n,m}^{(s)} \varphi_n^{(u)} + A_{n,m}^{(st)} x_{n,m}^{(s)} \varphi_n^{(t)} \nonumber \\
&+ A_{n,m}^{(sa)} x_{n,m}^{(s)} \varphi_n^{(a)} + B_{n,m}^{(s)} x_{n,m}^{(s)} + C_{n,m}^{(s)} \varphi_n^{(s)} + D_{n,m}^{(s)} T_{n,m}^{(s)} \nonumber \\
&+ \sum_{n \in \mathcal{N}} C_{n}^{(u)} \varphi_n^{(u)} + D_{n}^{(u)} T_{n}^{(u)}.
\end{talign}
It's clear that Problem $\mathbb{P}_{7}$  is a quadratically constrained quadratic programming (QCQP) problem. To combine $\bm{\varphi}$ and $\bm{x}$, we define a new matrix 
\begin{talign}
\bm{Q}:=&[(\bm{\varphi^{(u)}})^\intercal,(\bm{\varphi^{(t)}})^\intercal,(\bm{\varphi^{(a)}})^\intercal,(\bm{\varphi^{(s)}})^\intercal,(\bm{x^{(t)}})^\intercal,(\bm{x^{(a)}})^\intercal,\nonumber \\
&(\bm{x^{(s)}})^\intercal]^\intercal,
\end{talign}
where $\bm{\varphi^{(i)}}=(\varphi_1,\cdots,\varphi_N)^\intercal$, for $i \in \{u,t,a,s\}$, and $\bm{x^{(i)}}=(x^{(i)}_{1,m^{(i)}},\cdots,x^{(i)}_{N,m^{(i)}},\cdots,\cdots,x^{(i)}_{N,M^{(i)}})$, for $j \in \{t,a,s\}$. We define some auxiliary matrices and vectors to aid in our transformation. Let
\begin{talign}
&e_i:=(0,\cdots,1_{i\text{-th}},\cdots,0)_{NM+4N\times1}^\intercal,\\
&\bm{e}_{i,j}:=(e_i,\cdots,e_j)^\intercal,\\
&e_{\overline{i_k}}\!\nonumber \\
&:=\!(0,\!\cdots\!,\!1_{i\text{-th}},\cdots,1_{(i+N)\text{-th}},\!\cdots\!,1_{[i+N(M^{(k)}-1)]\text{-th}},\cdots,0)^\intercal,\nonumber \\
&\hspace{150pt}k \in \{t,a,s\},\\
&\bm{e}_{\overline{i},\overline{j}}:=(e_{\overline{i}},\cdots,e_{\overline{j}})^\intercal,\\
&e_{i\rightarrow j}:=(0,\cdots,1_{i\text{-th}},1,\cdots,1_{j\text{-th}},0,\cdots,0)^\intercal, i<j,\\
&\bm{e}_{M^{(t)}}:=\bm{e}_{4N+1,4N+NM^{(t)}},\\
&\bm{e}_{M^{(a)}}:=\bm{e}_{4N+NM^{(t)}+1,4N+NM^{(t)}+NM^{(a)}},\\
&\bm{e}_{M^{(s)}}:=\bm{e}_{4N+NM^{(t)}+NM^{(a)}+1,4N+NM^{(t)}+NM^{(a)}+NM^{(s)}},\\
&\bm{e}_{\varphi_u}:=\bm{e}_{1,N},\\
&\bm{e}_{\varphi_t}:=\bm{e}_{N+1,2N},\\
&\bm{e}_{\varphi_a}:=\bm{e}_{2N+1,3N},\\
&\bm{e}_{\varphi_s}:=\bm{e}_{3N+1,4N},\\
&\bm{I}_{N\rightarrow NM}:=(\bm{I}_{N},\cdots,\bm{I}_{N})_{N\times NM}.
\end{talign}
We define variables $T^{(u)}$, $T^{(t)}$, $T^{(a)}$ and $T^{(s)}$ as
\begin{talign}
&\sum_{n\in \mathcal{N}} D_{n}^{(u)} T_{n}^{(u)} = T^{(u)},\\
&\sum_{n\in \mathcal{N},m \in \mathcal{M}} D_{n,m}^{(t)} T_{n,m}^{(t)} = T^{(t)},\\
&\sum_{n\in \mathcal{N},m \in \mathcal{M}} D_{n,m}^{(a)} T_{n,m}^{(a)} = T^{(a)},\\
&\sum_{n\in \mathcal{N},m \in \mathcal{M}} D_{n,m}^{(s)} T_{n,m}^{(s)} = T^{(s)}.
\end{talign}
Next, we define the following matrices:
\begin{talign}
\bm{A^{(tt)}}:=[A_{n,m}^{(tt)}]|_{n \in \mathcal{N},m \in \mathcal{M}}.
\end{talign}
Similarly, we define other matrices $\bm{A^{(tu)}},\bm{B^{(t)}},\cdots$. We can obtain that
\begin{talign}
&\sum_{n \in \mathcal{N}}\sum_{m \in \mathcal{M}}A_{n,m}^{(tt)}x_{n,m}^{(t)} \varphi_n^{(t)}\nonumber \\
&=\bm{Q}^\intercal(\bm{0}_{N\times N},\bm{I}_{N},\bm{0}_{N\times 2N+NM})^\intercal \bm{I}_{N\rightarrow NM^{(t)}}\text{diag}(\bm{A^{(tt)}})\nonumber \\
&\cdot \bm{e}_{M^{(t)}}\bm{Q},
\end{talign}
\begin{talign}
&\sum_{n \in \mathcal{N}}\sum_{m \in \mathcal{M}}A_{n,m}^{(tu)}x_{n,m}^{(t)} \varphi_n^{(u)}\nonumber \\
&=\bm{Q}^\intercal(\bm{I}_{N},\bm{0}_{N\times 3N+NM})^\intercal \bm{I}_{N\rightarrow NM^{(t)}}\text{diag}(\bm{A^{(tu)}}) \bm{e}_{M^{(t)}}\bm{Q},
\end{talign}
\begin{talign}
&\sum_{n \in \mathcal{N}}\sum_{m \in \mathcal{M}}A_{n,m}^{(aa)}x_{n,m}^{(a)} \varphi_n^{(a)}\nonumber \\
&=\bm{Q}^\intercal(\bm{0}_{N\times 2N},\bm{I}_{N},\bm{0}_{N\times N+NM})^\intercal \bm{I}_{N\rightarrow NM^{(a)}}\text{diag}(\bm{A^{(aa)}})\nonumber \\
&\cdot \bm{e}_{M^{(a)}}\bm{Q},
\end{talign}
\begin{talign}
&\sum_{n \in \mathcal{N}}\sum_{m \in \mathcal{M}}A_{n,m}^{(au)}x_{n,m}^{(a)} \varphi_n^{(u)}\nonumber \\
&=\bm{Q}^\intercal(\bm{I}_{N},\bm{0}_{N\times 3N+NM})^\intercal \bm{I}_{N\rightarrow NM^{(a)}}\text{diag}(\bm{A^{(au)}})\bm{e}_{M^{(a)}}\bm{Q},
\end{talign}
\begin{talign}
&\sum_{n \in \mathcal{N}}\sum_{m \in \mathcal{M}}A_{n,m}^{(at)}x_{n,m}^{(a)} \varphi_n^{(t)}\nonumber \\
&=\bm{Q}^\intercal(\bm{0}_{N\times N},\bm{I}_{N},\bm{0}_{N\times 2N+NM})^\intercal \bm{I}_{N\rightarrow NM^{(a)}}\text{diag}(\bm{A^{(at)}})\nonumber \\
&\cdot \bm{e}_{M^{(a)}}\bm{Q},
\end{talign}
\begin{talign}
&\sum_{n \in \mathcal{N}}\sum_{m \in \mathcal{M}}A_{n,m}^{(ss)}x_{n,m}^{(s)} \varphi_n^{(s)}\nonumber \\
&=\bm{Q}^\intercal(\bm{0}_{N\times 3N},\bm{I}_{N},\bm{0}_{N\times NM})^\intercal \bm{I}_{N\rightarrow NM^{(s)}}\text{diag}(\bm{A^{(ss)}})\nonumber \\
&\cdot \bm{e}_{M^{(s)}}\bm{Q},
\end{talign}
\begin{talign}
&\sum_{n \in \mathcal{N}}\sum_{m \in \mathcal{M}}A_{n,m}^{(su)}x_{n,m}^{(s)} \varphi_n^{(u)}\nonumber \\
&=\bm{Q}^\intercal(\bm{I}_{N},\bm{0}_{N\times 3N+NM})^\intercal \bm{I}_{N\rightarrow NM^{(s)}}\text{diag}(\bm{A^{(su)}}) \bm{e}_{M^{(s)}}\bm{Q},
\end{talign}
\begin{talign}
&\sum_{n \in \mathcal{N}}\sum_{m \in \mathcal{M}}A_{n,m}^{(st)}x_{n,m}^{(s)} \varphi_n^{(t)}\nonumber \\
&=\bm{Q}^\intercal(\bm{0}_{N\times N},\bm{I}_{N},\bm{0}_{N\times 2N+NM})^\intercal \bm{I}_{N\rightarrow NM^{(s)}}\text{diag}(\bm{A^{(st)}})\nonumber \\
&\cdot \bm{e}_{M^{(s)}}\bm{Q},
\end{talign}
\begin{talign}
&\sum_{n \in \mathcal{N}}\sum_{m \in \mathcal{M}}A_{n,m}^{(sa)}x_{n,m}^{(s)} \varphi_n^{(a)}\nonumber \\
&=\bm{Q}^\intercal(\bm{0}_{N\times 2N},\bm{I}_{N},\bm{0}_{N\times N+NM})^\intercal \bm{I}_{N\rightarrow NM^{(s)}}\text{diag}(\bm{A^{(sa)}})\nonumber \\
&\cdot \bm{e}_{M^{(s)}}\bm{Q},
\end{talign}
\begin{talign}
&\sum_{n \in \mathcal{N}}\sum_{m \in \mathcal{M}} B_{n,m}^{(t)} x_{n,m}^{(t)}={\bm{B^{(t)}}}^\intercal \bm{e}_{M^{(t)}} \bm{Q},
\end{talign}
\begin{talign}
&\sum_{n \in \mathcal{N}}\sum_{m \in \mathcal{M}} B_{n,m}^{(a)}x_{n,m}^{(a)}={\bm{B^{(a)}}}^\intercal \bm{e}_{M^{(a)}} \bm{Q},
\end{talign}
\begin{talign}
&\sum_{n \in \mathcal{N}}\sum_{m \in \mathcal{M}} B_{n,m}^{(s)}x_{n,m}^{(s)}={\bm{B^{(s)}}}^\intercal\bm{e}_{M^{(s)}} \bm{Q},
\end{talign}
\begin{talign}
&\sum_{n \in \mathcal{N}}\sum_{m \in \mathcal{M}} C_{n,m}^{(t)}\varphi_{n,m}^{(t)}={\bm{C^{(t)}}}^\intercal\bm{e}_{\varphi_t} \bm{Q},
\end{talign}
\begin{talign}
&\sum_{n \in \mathcal{N}}\sum_{m \in \mathcal{M}} C_{n,m}^{(a)}\varphi_{n,m}^{(a)}={\bm{C^{(a)}}}^\intercal\bm{e}_{\varphi_a} \bm{Q},
\end{talign}
\begin{talign}
&\sum_{n \in \mathcal{N}}\sum_{m \in \mathcal{M}} C_{n,m}^{(s)}\varphi_{n,m}^{(s)}={\bm{C^{(s)}}}^\intercal\bm{e}_{\varphi_s} \bm{Q},
\end{talign}
\begin{talign}
\sum_{n \in \mathcal{N}}C_n^{(u)}\varphi_n^{(u)} = {\bm{C^{(u)}}}^\intercal\bm{e}_{\varphi_u} \bm{Q}.
\end{talign}
We define a matrix $\bm{P}_0$ as follows:
\begin{talign}
&\bm{P}_0 \nonumber \\
&=(\bm{0}_{N\times N},\bm{I}_{N},\bm{0}_{N\times 2N+NM})^\intercal \bm{I}_{N\rightarrow NM^{(t)}}\text{diag}(\bm{A^{(tt)}})\bm{e}_{M^{(t)}}\nonumber \\
& + (\bm{I}_{N},\bm{0}_{N\times 3N+NM})^\intercal \bm{I}_{N\rightarrow NM^{(t)}} \text{diag}(\bm{A^{(tu)}})\bm{e}_{M^{(t)}} \nonumber \\
&+(\bm{0}_{N\times 2N},\bm{I}_{N},\bm{0}_{N\times N+NM})^\intercal \bm{I}_{N\rightarrow NM^{(a)}}\text{diag}(\bm{A^{(aa)}})\bm{e}_{M^{(a)}}\nonumber \\
&+(\bm{I}_{N},\bm{0}_{N\times 3N+NM})^\intercal \bm{I}_{N\rightarrow NM^{(a)}}\text{diag}(\bm{A^{(au)}})\bm{e}_{M^{(a)}}\nonumber \\
&+(\bm{0}_{N\times N},\bm{I}_{N},\bm{0}_{N\times 2N+NM})^\intercal \bm{I}_{N\rightarrow NM^{(a)}}\text{diag}(\bm{A^{(at)}})\bm{e}_{M^{(a)}}\nonumber \\
&+(\bm{0}_{N\times 3N},\bm{I}_{N},\bm{0}_{N\times NM})^\intercal \bm{I}_{N\rightarrow NM^{(s)}}\text{diag}(\bm{A^{(ss)}})\bm{e}_{M^{(s)}}\nonumber \\
&+(\bm{I}_{N},\bm{0}_{N\times 3N+NM})^\intercal \bm{I}_{N\rightarrow NM^{(s)}}\text{diag}(\bm{A^{(su)}})\bm{e}_{M^{(s)}}\nonumber \\
&+(\bm{0}_{N\times N},\bm{I}_{N},\bm{0}_{N\times 2N+NM})^\intercal \bm{I}_{N\rightarrow NM^{(s)}}\text{diag}(\bm{A^{(st)}})\bm{e}_{M^{(s)}}\nonumber \\
&+(\bm{0}_{N\times 2N},\bm{I}_{N},\bm{0}_{N\times N+NM})^\intercal \bm{I}_{N\rightarrow NM^{(s)}}\text{diag}(\bm{A^{(sa)}})\bm{e}_{M^{(s)}}.
\end{talign}
Next, we define another matrix $\bm{W}_0^\intercal$ as follows:
\begin{talign}
&\bm{W}_0^\intercal ={\bm{B^{(t)}}}^\intercal\bm{e}_{M^{(t)}} + {\bm{B^{(a)}}}^\intercal \bm{e}_{M^{(a)}}+ {\bm{B^{(s)}}}^\intercal\bm{e}_{M^{(s)}}\nonumber \\
&+{\bm{C^{(t)}}}^\intercal\bm{e}_{\varphi_t} + {\bm{C^{(a)}}}^\intercal\bm{e}_{\varphi_a} + {\bm{C^{(s)}}}^\intercal\bm{e}_{\varphi_s} +{\bm{C^{(u)}}}^\intercal\bm{e}_{\varphi_u}.
\end{talign}
Based on the above analysis, we finally can express the objective function in Problem $\mathbb{P}_{7}$ as
\begin{talign}
\bm{Q}^\intercal \bm{P}_0 \bm{Q} + \bm{W}_0^\intercal \bm{Q} + T^{(u)} + T^{(t)} + T^{(a)} + T^{(s)}.
\end{talign}
Next, we analyze the delay terms. 
For $T^{(u)}$,
\begin{talign}
&\sum_{n \in \mathcal{N}}-D_n^{(u)}T^{(up)}_n \leq T^{(u)},\nonumber \\
&\Rightarrow \sum_{n \in \mathcal{N}} -D_n^{(u)}\frac{e_n t_n d_n}{\gamma_n^{(u)}f_n}\varphi_n^{(u)} \leq T^{(u)}.
\end{talign}
To make the expression clearer, we define that
\begin{talign}
&W^{(T_u)}_n:=-D_n^{(u)}\frac{e_n t_n d_n}{\gamma_n^{(u)}f_n},\\
&\bm{W^{(T_u)}}:=[W^{(T_u)}_n]|_{n \in \mathcal{N}}.
\end{talign}
Thus, we can obtain that
\begin{talign}
\sum_{n \in \mathcal{N}}-D_n^{(u)}T^{(up)}_n={\bm{W^{(T_u)}}}^\intercal \bm{e}_{\varphi_u} \bm{Q},
\end{talign}
\begin{talign}
\sum_{n \in \mathcal{N}}-D_n^{(u)}T^{(up)}_n \leq T^{(u)}\Longleftrightarrow{\bm{W^{(T_u)}}}^\intercal \bm{e}_{\varphi_u} \bm{Q} \leq T^{(u)}.
\end{talign}
For $T^{(t)}$,
\begin{talign}
&\sum_{n \in \mathcal{N}, m \in \mathcal{M}}-D_{n,m}^{(t)}(T_{n,m}^{(ut)} + T_{n,m}^{(tp)}) \leq T^{(t)},\nonumber \\
&\Rightarrow \sum_{n \in \mathcal{N}, m \in \mathcal{M}} -D_{n,m}^{(t)}\frac{x_{n,m}^{(t)}[\omega_b (1-\varphi_n^{(u)})d_n+d_n^{(l)}]}{r_{n,m_t}} -  D_{n,m}^{(t)}x_{n,m}^{(t)}\nonumber \\
&\hspace{150pt}\cdot \frac{e_{m_t}t_n\varphi_n^{(t)}d_n}{\gamma_{n,m}^{(t)}f_{m_t}} \leq T^{(t)},\nonumber \\
&\Rightarrow \sum_{n \in \mathcal{N}, m \in \mathcal{M}} -D_{n,m}^{(t)}\frac{\omega_b d_n + d_n^{(l)}}{r_{n,m_t}}
x_{n,m}^{(t)} + D_{n,m}^{(t)}  \frac{\omega_b d_n}{r_{n,m_t}}x_{n,m}^{(t)}\nonumber \\
&\cdot\varphi_n^{(u)}-D_{n,m}^{(t)} \frac{e_{m_t} t_n d_n}{\gamma_{n,m}^{(t)}f_{m_t}}x_{n,m}^{(t)}\varphi_n^{(t)} \leq T^{(t)}.
\end{talign}
To make the expression clearer, we define the following auxiliary variables and matrices:
\begin{talign}
&W_{1,n,m}^{(T_t)}:=-D_{n,m}^{(t)}\frac{\omega_b d_n + d_n^{(l)}}{r_{n,m_t}},\\ 
&W_{2,n,m}^{(T_t)}:=D_{n,m}^{(t)}\frac{\omega_b d_n}{r_{n,m_t}},\\
&W_{3,n,m}^{(T_t)}:=-D_{n,m}^{(t)} \frac{e_n t_n d_n}{\gamma_{n,m}^{(t)}f_{m_t}},\\
&\bm{W_1^{(T_t)}}:=[W_{1,n,m}^{(T_t)}]|_{n \in \mathcal{N}, m \in \mathcal{M}},\\
&\bm{W_2^{(T_t)}}:=[W_{2,n,m}^{(T_t)}]|_{n \in \mathcal{N}, m \in \mathcal{M}},\\
&\bm{W_3^{(T_t)}}:=[W_{2,n,m}^{(T_t)}]|_{n \in \mathcal{N}, m \in \mathcal{M}}.
\end{talign}
Based on the predefined auxiliary variables and matrices, we can obtain that
\begin{talign}
&\sum_{n \in \mathcal{N}, m \in \mathcal{M}}-D_{n,m}^{(t)}(T_{n,m}^{(ut)} + T_{n,m}^{(tp)})\leq T^{(t)}\nonumber \\
&\Longleftrightarrow\bm{Q}^\intercal(\bm{I}_{N},\bm{0}_{N\times 3N+NM})^\intercal \bm{I}_{N\rightarrow NM^{(t)}}\text{diag}(\bm{W_2^{(T_t)}})\bm{e}_{M^{(t)}} \bm{Q} \nonumber \\
&+ \bm{Q}^\intercal(\bm{0}_{N\times N},\bm{I}_{N},\bm{0}_{N\times 2N+NM})^\intercal \bm{I}_{N\rightarrow NM^{(t)}}\text{diag}(\bm{W_3^{(T_t)}})\nonumber \\
&\cdot \bm{e}_{M^{(t)}}\bm{Q} + {\bm{W_1^{(T_t)}}}^\intercal \bm{e}_{M^{(t)}} \bm{Q} \leq T^{(t)}.
\end{talign}
For $T^{(a)}$,
\begin{talign}
&\sum_{n \in \mathcal{N}, m \in \mathcal{M}}-D_{n,m}^{(a)}(T_{n,m}^{(tt)} + T_{n,m}^{(ap)}) \leq T^{(a)},\nonumber \\
&\Rightarrow \!\!\sum_{n \in \mathcal{N}, m \in \mathcal{M}} -D_{n,m}^{(a)}\frac{x_{n,m}^{(a)}[\omega_b(1-\varphi_n^{(u)}-\varphi_n^{(t)})d_n + d_n^{(l)}]}{r_{m_t,m_a}} \!-\!D_{n,m}^{(a)}\nonumber \\
&\hspace{130pt}\cdot \frac{x_{n,m}^{(a)}e_{m_a}t_n\varphi_n^{(a)}d_n}{\gamma_{n,m}^{(a)}f_{m_a}} \leq T^{(a)},\nonumber \\
&\Rightarrow \sum_{n \in \mathcal{N}, m \in \mathcal{M}}-D_{n,m}^{(a)} \frac{\omega_b d_n + d_n^{(l)}}{r_{m_t,m_a}}x_{n,m}^{(a)} + D_{n,m}^{(a)} \frac{\omega_b d_n}{r_{m_t,m_a}}x_{n,m}^{(a)}\nonumber \\
&\cdot\varphi_n^{(u)} + D_{n,m}^{(a)} \frac{\omega_b d_n}{r_{m_t,m_a}} x_{n,m}^{(a)}\varphi_n^{(t)} - D_{n,m}^{(a)} \frac{e_{m_a}d_n t_n}{\gamma_{n,m}^{(a)}f_{m_a}} x_{n,m}^{(a)}\varphi_n^{(a)} \nonumber \\
&\hspace{180pt}\leq T^{(a)}.
\end{talign}
To make the expression clearer, we define the following auxiliary variables and matrices:
\begin{talign}
&W_{1,n,m}^{(T_a)}:=-D_{n,m}^{(a)} \frac{\omega_b d_n + d_n^{(l)}}{r_{m_t,m_a}},\\
&W_{2,n,m}^{(T_a)}:=D_{n,m}^{(a)} \frac{\omega_b d_n}{r_{m_t,m_a}},\\
&W_{3,n,m}^{(T_a)}:=-D_{n,m}^{(a)}\frac{e_{m_a} t_n d_n}{\gamma_{n,m}^{(a)}f_{m_a}},\\
&\bm{W_1^{(T_a)}}:=[W_{1,n,m}^{(T_a)}]|_{n \in \mathcal{N}, m \in \mathcal{M}},\\
&\bm{W_2^{(T_a)}}:=[W_{2,n,m}^{(T_a)}]|_{n \in \mathcal{N}, m \in \mathcal{M}},\\
&\bm{W_3^{(T_a)}}:=[W_{2,n,m}^{(T_a)}]|_{n \in \mathcal{N}, m \in \mathcal{M}}.
\end{talign}
Therefore, we get the following conclusion:
\begin{talign}
&\sum_{n \in \mathcal{N}, m \in \mathcal{M}}-D_{n,m}^{(a)}(T_{n,m}^{(tt)} + T_{n,m}^{(ap)}) \leq T^{(a)}\nonumber \\
&\Longleftrightarrow\bm{Q}^\intercal(\bm{I}_{N},\bm{0}_{N\times 3N+NM})^\intercal \bm{I}_{N\rightarrow NM^{(a)}}\text{diag}(\bm{W_2^{(T_a)}})\bm{e}_{M^{(a)}}\bm{Q} \nonumber \\
&+ \bm{Q}^\intercal(\bm{0}_{N\times N},\bm{I}_{N},\bm{0}_{N\times 2N+NM})^\intercal \bm{I}_{N\rightarrow NM^{(a)}}\text{diag}(\bm{W_2^{(T_a)}})\nonumber \\
&\cdot \bm{e}_{M^{(a)}}\bm{Q} + \bm{Q}^\intercal (\bm{0}_{N\times 2N},\bm{I}_{N},\bm{0}_{N\times N+NM})^\intercal \bm{I}_{N\rightarrow NM^{(a)}}\nonumber \\
&\cdot \text{diag}(\bm{W_3^{(T_a)}}) \bm{e}_{M^{(a)}}\bm{Q} + {\bm{W_1^{(T_a)}}}^\intercal \bm{e}_{M^{(a)}}\bm{Q} \leq T^{(a)}.
\end{talign}
For $T^{(s)}$,
\begin{talign}
&\sum_{n \in \mathcal{N}, m \in \mathcal{M}}-D_{n,m}^{(s)}(T_{n,m}^{(at)} + T_{n,m}^{(sp)}) \leq T^{(s)},\nonumber \\
&\Rightarrow \sum_{n \in \mathcal{N}, m \in \mathcal{M}}-D_{n,m}^{(s)}\frac{x_{n,m}^{(s)}[\omega_b(1-\varphi_n^{(u)}-\varphi_n^{(t)}-\varphi_n^{(a)})d_n + d_n^{(l)}]}{r_{m_a,m_s}} \nonumber \\
&\hspace{110pt}-D_{n,m}^{(s)}\! \frac{x_{n,m}^{(s)}e_{m_s}t_n\varphi_n^{(s)}d_n}{\gamma_{n,m}^{(s)}f_{m_s}}\! \leq \!T^{(s)},\nonumber \\
&\Rightarrow \sum_{n \in \mathcal{N}, m \in \mathcal{M}} -D_{n,m}^{(s)}\frac{\omega_b d_n + d_n^{(l)}}{r_{m_a,m_s}}x_{n,m}^{(s)} + D_{n,m}^{(s)}\frac{\omega_b d_n}{r_{m_a,m_s}}x_{n,m}^{(s)}\nonumber \\
&\cdot \varphi_n^{(u)} + D_{n,m}^{(s)} \frac{\omega_b d_n}{r_{m_a,m_s}}x_{n,m}^{(s)}\varphi_n^{(t)} + D_{n,m}^{(s)} \frac{\omega_b d_n}{r_{m_a,m_s}}x_{n,m}^{(s)}\varphi_n^{(a)}\nonumber \\
&- D_{n,m}^{(s)} \frac{e_{m_s} t_n d_n}{\gamma_{n,m}^{(s)}f_{m_s}}x_{n,m}^{(s)}\varphi_n^{(s)} \leq T^{(s)}
\end{talign}
To make the expression clearer, we define the following auxiliary variables and matrices:
\begin{talign}
&W_{1,n,m}^{(T_s)}:=-D_{n,m}^{(s)}\frac{\omega_b d_n + d_n^{(l)}}{r_{m_a,m_s}},\\
&W_{2,n,m}^{(T_s)}:=D_{n,m}^{(s)} \frac{\omega_b d_n}{r_{m_a,m_s}},\\
&W_{3,n,m}^{(T_s)}:=- D_{n,m}^{(s)} \frac{e_{m_s} t_n d_n}{\gamma_{n,m}^{(s)}f_{m_s}},\\
&\bm{W_1^{(T_s)}}:=[W_{1,n,m}^{(T_s)}]|_{n \in \mathcal{N}, m \in \mathcal{M}},\\
&\bm{W_2^{(T_s)}}:=[W_{2,n,m}^{(T_s)}]|_{n \in \mathcal{N}, m \in \mathcal{M}},\\
&\bm{W_3^{(T_s)}}:=[W_{2,n,m}^{(T_s)}]|_{n \in \mathcal{N}, m \in \mathcal{M}}.
\end{talign}
Thus, we can know that 
\begin{talign}
&\sum_{n \in \mathcal{N}, m \in \mathcal{M}}-D_{n,m}^{(s)}(T_{n,m}^{(at)} + T_{n,m}^{(sp)}) \leq T^{(s)}\nonumber\\
&\Longleftrightarrow\bm{Q}^\intercal(\bm{I}_{N},\bm{0}_{N\times 3N+NM})^\intercal \bm{I}_{N\rightarrow NM^{(s)}}\text{diag}(\bm{W_2^{(T_s)}}) \bm{e}_{M^{(s)}}\bm{Q}\nonumber \\
&+ 
\bm{Q}^\intercal (\bm{0}_{N\times N},\bm{I}_{N},\bm{0}_{N\times 2N+NM})^\intercal \bm{I}_{N\rightarrow NM^{(s)}}\text{diag}(\bm{W_2^{(T_s)}})\nonumber \\
&\cdot \bm{e}_{M^{(s)}}\bm{Q}+ 
\bm{Q}^\intercal(\bm{0}_{N\times 2N},\bm{I}_{N},\bm{0}_{N\times N+NM})^\intercal \bm{I}_{N\rightarrow NM^{(s)}}\nonumber \\
&\cdot \text{diag}(\bm{W_2^{(T_s)}}) \bm{e}_{M^{(s)}}\bm{Q}+\bm{Q}^\intercal(\bm{0}_{N\times 3N},\bm{I}_{N},\bm{0}_{N\times NM})^\intercal \nonumber \\
&\cdot \bm{I}_{N\rightarrow NM^{(s)}}\text{diag}(\bm{W_3^{(T_s)}}) \bm{e}_{M^{(s)}}\bm{Q}+\!\!{\bm{W_1^{(T_s)}}}^\intercal \!\!\!\bm{e}_{M^{(s)}} \bm{Q}\leq T^{(s)}.
\end{talign}
To make the expression clearer, we define the following auxiliary matrices:
\begin{talign}
&{\bm{P^{(T_u)}}}^\intercal = {\bm{W^{(T_u)}}}^\intercal \bm{e}_{\varphi_u},
\end{talign}
\begin{talign}
&\bm{P_1^{(T_t)}} = (\bm{I}_{N},\bm{0}_{N\times 3N+NM})^\intercal \bm{I}_{N\rightarrow NM^{(t)}}\text{diag}(\bm{W_2^{(T_t)}}) \bm{e}_{M^{(t)}} \nonumber \\
&+ (\bm{0}_{N\times N},\bm{I}_{N},\bm{0}_{N\times 2N+NM})^\intercal \bm{I}_{N\rightarrow NM^{(t)}}\text{diag}(\bm{W_3^{(T_t)}}) \bm{e}_{M^{(t)}},
\end{talign}
\begin{talign}
&{\bm{P_2^{(T_t)}}}^\intercal = {\bm{W_1^{(T_t)}}}^\intercal \bm{e}_{M^{(t)}},
\end{talign}
\begin{talign}
&\bm{P_1^{(T_a)}} = (\bm{I}_{N},\bm{0}_{N\times 3N+NM})^\intercal \bm{I}_{N\rightarrow NM^{(a)}}\text{diag}(\bm{W_2^{(T_a)}}) \bm{e}_{M^{(a)}} \nonumber \\
&+ (\bm{0}_{N\times N},\bm{I}_{N},\bm{0}_{N\times 2N+NM})^\intercal \bm{I}_{N\rightarrow NM^{(a)}}\text{diag}(\bm{W_2^{(T_a)}})\bm{e}_{M^{(a)}}\nonumber \\
&+ (\bm{0}_{N\times 2N},\bm{I}_{N},\bm{0}_{N\times N+NM})^\intercal \bm{I}_{N\rightarrow NM^{(a)}}\text{diag}(\bm{W_3^{(T_a)}})\bm{e}_{M^{(a)}},
\end{talign}
\begin{talign}
&{\bm{P_2^{(T_a)}}}^\intercal = {\bm{W_1^{(T_a)}}}^\intercal \bm{e}_{M^{(a)}},
\end{talign}
\begin{talign}
&\bm{P_1^{(T_s)}} = 
(\bm{I}_{N},\bm{0}_{N\times 3N+NM})^\intercal \bm{I}_{N\rightarrow NM^{(s)}}\text{diag}(\bm{W_2^{(T_s)}}) \bm{e}_{M^{(s)}}\nonumber \\
&+ 
(\bm{0}_{N\times N},\bm{I}_{N},\bm{0}_{N\times 2N+NM})^\intercal \bm{I}_{N\rightarrow NM^{(s)}}\text{diag}(\bm{W_2^{(T_s)}})\bm{e}_{M^{(s)}}\nonumber \\
&+ 
(\bm{0}_{N\times 2N},\bm{I}_{N},\bm{0}_{N\times N+NM})^\intercal \bm{I}_{N\rightarrow NM^{(s)}}\text{diag}(\bm{W_2^{(T_s)}})\bm{e}_{M^{(s)}}\nonumber \\
&+(\bm{0}_{N\times 3N},\bm{I}_{N},\bm{0}_{N\times NM})^\intercal \bm{I}_{N\rightarrow NM^{(s)}}\text{diag}(\bm{W_3^{(T_s)}})\bm{e}_{M^{(s)}},
\end{talign}
\begin{talign}
&{\bm{P_2^{(T_s)}}}^\intercal = {\bm{W_1^{(T_s)}}}^\intercal\bm{e}_{M^{(s)}}.
\end{talign}
Therefore, the delay term constraints (\ref{Tu_constr})-(\ref{Ts_constr}) can be transformed into new constraints shown as follows:
\begin{talign}
&{\bm{P^{(T_u)}}}^\intercal \bm{Q} \leq T^{(u)},\\
&\bm{Q}^\intercal \bm{P_1^{(T_t)}} \bm{Q} + {\bm{P_2^{(T_t)}}}^\intercal \bm{Q} \leq T^{(t)},\\
&\bm{Q}^\intercal \bm{P_1^{(T_a)}} \bm{Q} + {\bm{P_2^{(T_a)}}}^\intercal \bm{Q} \leq T^{(a)},\\
&\bm{Q}^\intercal \bm{P_1^{(T_s)}} \bm{Q} + {\bm{P_2^{(T_s)}}}^\intercal \bm{Q} \leq T^{(s)}.
\end{talign}
For constraint (\ref{x_range_constr_new}), it can be rewritten as 
\begin{talign}
&\text{diag}(\bm{e}_{M^{(t)}}^\intercal\bm{Q})(\text{diag}(\bm{e}_{M^{(t)}}^\intercal\bm{Q}) - \bm{I})= 0,\\
&\text{diag}(\bm{e}_{M^{(a)}}^\intercal\bm{Q})(\text{diag}(\bm{e}_{M^{(a)}}^\intercal\bm{Q}) - \bm{I}) = 0,\\
&\text{diag}(\bm{e}_{M^{(s)}}^\intercal\bm{Q})(\text{diag}(\bm{e}_{M^{(s)}}^\intercal\bm{Q})- \bm{I}) = 0.
\end{talign}
For constraint (\ref{x_sum_constr}), it can be transformed into
\begin{talign}
&\text{diag}(\bm{e}_{\overline{1_t},\overline{M^{(t)}_t}}^\intercal \bm{e}_{M^{(t)}}^\intercal\bm{Q}) = \bm{I},\\
&\text{diag}(\bm{e}_{\overline{1_a},\overline{M^{(a)}_a}}^\intercal \bm{e}_{M^{(a)}}^\intercal\bm{Q}) = \bm{I},\\
&\text{diag}(\bm{e}_{\overline{1_s},\overline{M^{(s)}_s}}^\intercal \bm{e}_{M^{(s)}}^\intercal\bm{Q}) = \bm{I}.
\end{talign}
Let's ignore the restriction of greater than or equal to 0 for a moment and variables considered are all greater than or equal to 0. We will add this restriction in the final form of the transformed problem. For constraints (\ref{varphi_range_constr})-(\ref{varphi_sum_constr2}), their new forms are
\begin{talign}
&\text{diag}(\bm{e}_{\varphi_u}^\intercal\bm{Q}) \leq \bm{I},\\
&\text{diag}(\bm{e}_{\varphi_t}^\intercal\bm{Q}) \leq \bm{I},\\
&\text{diag}(\bm{e}_{\varphi_a}^\intercal\bm{Q}) \leq \bm{I},\\
&\text{diag}(\bm{e}_{\varphi_s}^\intercal\bm{Q}) \leq \bm{I},\\
&\text{diag}\big((\bm{e}_{\varphi_u}^\intercal+\bm{e}_{\varphi_t}^\intercal+\bm{e}_{\varphi_a}^\intercal+\bm{e}_{\varphi_s}^\intercal)\bm{Q}\big) = \bm{I}.
\end{talign}
For constraints (\ref{phi_sum_constr}), (\ref{gamma_sum_constr}), and (\ref{rho_sum_constr}), they can be rewritten as
\begin{talign}
&\bm{\phi^{(t)}}\bm{e}_{M^{(t)}}^\intercal\bm{Q}-1 \leq 0,\\
&\bm{\phi^{(a)}}\bm{e}_{M^{(a)}}^\intercal\bm{Q}-1 \leq 0,\\
&\bm{\phi^{(s)}}\bm{e}_{M^{(s)}}^\intercal\bm{Q}-1\leq 0,\\
&\bm{\gamma^{(t)}}\bm{e}_{M^{(t)}}^\intercal\bm{Q}-1 \leq 0,\\
&\bm{\gamma^{(a)}}\bm{e}_{M^{(a)}}^\intercal\bm{Q}-1 \leq 0,\\
&\bm{\gamma^{(s)}}\bm{e}_{M^{(s)}}^\intercal\bm{Q}-1\leq 0,\\
&\bm{\rho^{(t)}}\bm{e}_{M^{(t)}}^\intercal\bm{Q}-1 \leq 0,\\
&\bm{\rho^{(a)}}\bm{e}_{M^{(a)}}^\intercal\bm{Q}-1 \leq 0.
\end{talign}
Based on the above discussion, we transform ``maximization'' of Problem $\mathbb{P}_{7}$ to ``minimization'' to obtain the standard QCQP form Problem $\mathbb{P}_{8}$:
\begin{subequations}
\begin{talign}
\mathbb{P}_{8}:&\min\limits_{\bm{x},\bm{\varphi},\bm{T}} -\bm{Q}^\intercal \bm{P}_0 \bm{Q} - \bm{W}_0^\intercal \bm{Q} - T^{(u)} - T^{(t)} - T^{(a)} - T^{(s)}\nonumber \\
\text{s.t.} \quad
&\text{diag}(\bm{e}_{M^{(i)}}^\intercal\bm{Q})(\text{diag}(\bm{e}_{M^{(i)}}^\intercal\bm{Q})- \bm{I})= 0, \forall i \in \{t,a,s\},\nonumber\\
&\text{diag}(\bm{e}_{\overline{1_i},\overline{M^{(i)}_i}}^\intercal \bm{e}_{M^{(i)}}^\intercal\bm{Q}) = \bm{I},\forall i \in \{t,a,s\},\nonumber\\
&\text{diag}(\bm{e}_{\varphi_i}^\intercal\bm{Q}) \leq \bm{I},\forall i \in \{u,t,a,s\},\nonumber\\
&\text{diag}\big((\bm{e}_{\varphi_u}^\intercal+\bm{e}_{\varphi_t}^\intercal+\bm{e}_{\varphi_a}^\intercal+\bm{e}_{\varphi_s}^\intercal)\bm{Q}\big) = \bm{I},\nonumber\\
&\bm{\phi^{(i)}}\bm{e}_{M^{(i)}}^\intercal\bm{Q}-1 \leq 0,\forall i \in \{t,a,s\},\nonumber\\
&\bm{\gamma^{(i)}}\bm{e}_{M^{(i)}}^\intercal\bm{Q}-1 \leq 0,\forall i \in \{t,a,s\},\nonumber\\
&\bm{\rho^{(i)}}\bm{e}_{M^{(i)}}^\intercal\bm{Q}-1 \leq 0,\forall i \in \{t,a\},\nonumber\\
&{\bm{P^{(T_u)}}}^\intercal \bm{Q} \leq T^{(u)},\nonumber\\
&\bm{Q}^\intercal \bm{P_1^{(T_i)}} \bm{Q} + {\bm{P_2^{(T_i)}}}^\intercal \bm{Q} \leq T^{(i)},\forall i \in \{t,a,s\}.\nonumber
\end{talign}
\end{subequations}
\textbf{Lemma \ref{lemma_p7top8}} is proven.
\end{proof}

\section{Proof of \textbf{Lemma \ref{lemma_p8top9}}}\label{append_lemma_p8top9}
\begin{proof}
We introduce a new variable $\bm{S}:=(\bm{Q}^\intercal,1)^\intercal(\bm{Q}^\intercal,1)$. Let
\begin{talign}
\bm{P}_1=
\left(
    \begin{array}{cc}
       -\bm{P}_0  & -\frac{1}{2}\bm{W}_0 \\
        -\frac{1}{2}\bm{W}_0^\intercal &  - T^{(u)} - T^{(t)} - T^{(a)} - T^{(s)}
    \end{array}
\right),
\end{talign}

\begin{talign}
\bm{P}_2=
\left(
    \begin{array}{cc}
       \bm{e}_i^\intercal\bm{e}_i  & -\frac{1}{2}\bm{e}_i \\
        -\frac{1}{2}\bm{e}_i^\intercal &  0
    \end{array}
\right),\forall i \in \{4N+1,\cdots,4N+NM\},
\end{talign}

\begin{talign}
\bm{P}_3=
\left(
    \begin{array}{cc}
       \bm{0}_{4N+NM,4N+NM}  & \frac{1}{2}(\bm{e}_{\overline{1_k},\overline{M^{(k)}}} \bm{e}_{M^{(k)}}) \\
        \frac{1}{2}(\bm{e}_{\overline{1_k},\overline{M^{(k)}}} \bm{e}_{M^{(k)}})^\intercal &  -1
    \end{array}
\right),\nonumber \\
\forall i \in \{1,\cdots,N\}, \forall k \in \{t,a,s\},
\end{talign}

\begin{talign}
\bm{P}_4=
\left(
    \begin{array}{cc}
       \bm{0}_{4N+NM,4N+NM}  & \frac{1}{2}\bm{e}_i \\
        \frac{1}{2}\bm{e}^\intercal &  -1
    \end{array}
\right),
\forall i \in \{1,\cdots,4N\},
\end{talign}

\begin{talign}
\bm{P}_5=
\left(
    \begin{array}{cc}
       \bm{0}_{4N+NM,4N+NM}  & \frac{(\bm{e}_{\varphi_u}+\bm{e}_{\varphi_t}+\bm{e}_{\varphi_a}+\bm{e}_{\varphi_s})}{2} \\
        \frac{(\bm{e}_{\varphi_u}+\bm{e}_{\varphi_t}+\bm{e}_{\varphi_a}+\bm{e}_{\varphi_s})^\intercal}{2} &  -1
    \end{array}
\right),
\end{talign}

\begin{talign}
\bm{P}_6=
\left(
    \begin{array}{cc}
       \bm{0}_{4N+NM,4N+NM}  & \frac{1}{2}\bm{\phi^{(i)}}\bm{e}_{M^{(i)}} \\
        \frac{1}{2}\bm{\phi^{(i)}}\bm{e}_{M^{(i)}}^\intercal &  -1
    \end{array}
\right),
\forall i \in \{t,a,s\},
\end{talign}

\begin{talign}
\bm{P}_7=
\left(
    \begin{array}{cc}
       \bm{0}_{4N+NM,4N+NM}  & \frac{1}{2}\bm{\gamma^{(i)}}\bm{e}_{M^{(i)}} \\
        \frac{1}{2}\bm{\gamma^{(i)}}\bm{e}_{M^{(i)}}^\intercal &  -1
    \end{array}
\right),
\forall i \in \{t,a,s\},
\end{talign}

\begin{talign}
\bm{P}_8=
\left(
    \begin{array}{cc}
       \bm{0}_{4N+NM,4N+NM}  & \frac{1}{2}\bm{\rho^{(i)}}\bm{e}_{M^{(i)}} \\
        \frac{1}{2}\bm{\rho^{(i)}}\bm{e}_{M^{(i)}}^\intercal &  -1
    \end{array}
\right),
\forall i \in \{t,a\},
\end{talign}

\begin{talign}
\bm{P}_9=
\left(
    \begin{array}{cc}
       \bm{0}_{4N+NM,4N+NM}  & \frac{1}{2}\bm{P^{(T_u)}} \\
        {\frac{1}{2}\bm{P^{(T_u)}}}^\intercal &  0
    \end{array}
\right),
\end{talign}

\begin{talign}
\bm{P}_{10}=
\left(
    \begin{array}{cc}
       \bm{P_1^{(T_i)}}  & \frac{1}{2}\bm{P_2^{(T_i)}} \\
        {\frac{1}{2}\bm{P_2^{(T_i)}}}^\intercal &  0
    \end{array}
\right),
\forall i \in \{t,a,s\}.
\end{talign}
Therefore, we can obtain the following conclusions:
\begin{talign}
&-\bm{Q}^\intercal \bm{P}_0 \bm{Q} - \bm{W}_0^\intercal \bm{Q} - T^{(u)} - T^{(t)} - T^{(a)} - T^{(s)}\nonumber \\
&\Longleftrightarrow\text{Tr}(\bm{P}_1 \bm{S}).\\
&\text{diag}(\bm{e}_{M^{(i)}}^\intercal\bm{Q})(\text{diag}(\bm{e}_{M^{(i)}}^\intercal\bm{Q})- \bm{I})= 0, \forall i \in \{t,a,s\}\nonumber\\
&\Longleftrightarrow \text{Tr}(\bm{P}_2 \bm{S})=0.\\
&\text{diag}(\bm{e}_{\overline{1_i},\overline{M^{(i)}_i}}^\intercal \bm{e}_{M^{(i)}}^\intercal\bm{Q}) = \bm{I},\forall i \in \{t,a,s\}\Longleftrightarrow \text{Tr}(\bm{P}_3 \bm{S})=0.\\
&\text{diag}(\bm{e}_{\varphi_i}^\intercal\bm{Q}) \leq \bm{I},\forall i \in \{u,t,a,s\}\Longleftrightarrow\text{Tr}(\bm{P}_4 \bm{S})\leq 0.\\
&\text{diag}\big((\bm{e}_{\varphi_u}^\intercal+\bm{e}_{\varphi_t}^\intercal+\bm{e}_{\varphi_a}^\intercal+\bm{e}_{\varphi_s}^\intercal)\bm{Q}\big) = \bm{I}\Longleftrightarrow \text{Tr}(\bm{P}_5 \bm{S})=0.\\
&\bm{\phi^{(i)}}\bm{e}_{M^{(i)}}^\intercal\bm{Q}-1 \leq 0,\forall i \in \{t,a,s\}\Longleftrightarrow \text{Tr}(\bm{P}_6 \bm{S})\leq0.\\
&\bm{\gamma^{(i)}}\bm{e}_{M^{(i)}}^\intercal\bm{Q}-1 \leq 0,\forall i \in \{t,a,s\}\Longleftrightarrow \text{Tr}(\bm{P}_7 \bm{S})\leq0.\\
&\bm{\rho^{(i)}}\bm{e}_{M^{(i)}}^\intercal\bm{Q}-1 \leq 0,\forall i \in \{t,a\}\Longleftrightarrow \text{Tr}(\bm{P}_8 \bm{S})\leq0.\\
&{\bm{P^{(T_u)}}}^\intercal \bm{Q} \leq T^{(u)}\Longleftrightarrow \text{Tr}(\bm{P}_9 \bm{S})\leq T^{(u)}.
\end{talign}
\begin{talign}
&\bm{Q}^\intercal \bm{P_1^{(T_i)}} \bm{Q} + {\bm{P_2^{(T_i)}}}^\intercal \bm{Q} \leq T^{(i)},\forall i \in \{t,a,s\}\nonumber \\
&\Longleftrightarrow \text{Tr}(\bm{P}_{10} \bm{S})\leq T^{(i)}, \forall i \in \{t,a,s\}.
\end{talign}
Based on the above analysis, we can obtain a solvable SDR Problem $\mathbb{P}_{9}$.

\textbf{Lemma \ref{lemma_p8top9}} is proven.
\end{proof}

\section{Hyperparameter Settings}
Here we present the hyperparameters used in the PPO method in Table \ref{tab:ppo_hyperparams}.
\begin{table}[ht]
\centering
\caption{PPO hyperparameters used in our simulations.}
\setlength{\tabcolsep}{4pt}    
\renewcommand{\arraystretch}{1.0}
\begin{normalsize}             
\begin{tabular}{l c}
\hline
\textbf{Hyperparameter} & \textbf{Value} \\
\hline
Discount factor & 0.99 \\
GAE factor & 0.95 \\
PPO clipping ratio & 0.2 \\
PPO update epochs per iteration & 10 \\
Hidden layer width & 64 \\
Actor learning rate & $2\times10^{-3}$ \\
Critic learning rate & $2\times10^{-4}$ \\
Additional learning rate & $1\times10^{-4}$ \\
Critic L2 regularization coefficient & 0 \\
Trajectory batch size & 64 \\
Entropy coefficient & 0 \\
Entropy coefficient decay rate & 0.99 \\
Actor optimization batch size & 64 \\
Critic optimization batch size & 64 \\
\hline
\end{tabular}
\end{normalsize}
\label{tab:ppo_hyperparams}
\end{table}

\section{Discussion on the Weight Setting of Delay and Energy}
\label{sec.weight_setting}
In this section, we provide a qualitative discussion on the role and interpretation of the weight parameters $\omega_t$ and $\omega_e$ that appear in the cost terms of Problem~$\mathbb{P}_{3}$. Recall that the cost of each component is modeled as a weighted sum of delay and energy, for example
\begin{equation}
\label{eq:cost_example}
\text{cost} = \omega_t T + \omega_e E,
\end{equation}
where $\omega_t$ and $\omega_e$ are nonnegative parameters reflecting the relative importance of delay and energy consumption in the system design.

\subsection{Fix decision variables and then optimize weight parameters}
For a fixed resource allocation solution $(\bm{x}$, $\bm{\varphi}$, $\bm{\gamma}$, $\bm{\phi}$, $\bm{\rho}$, $\bm{\psi}$, $\bm{\alpha}$, $\bm{T})$, the objective of $\mathbb{P}_3$ can be written in a simplified affine form with respect to $(\omega_t,\omega_e)$ as
\begin{equation}
\label{eq:F_affine}
F(\omega_t,\omega_e) = A - B \omega_t - C \omega_e,
\end{equation}
where the constants $A$, $B$, and $C$ are nonnegative terms determined by the delay and energy components of the solution. This expression shows explicitly how the overall performance varies as the designer changes the emphasis on delay and energy through $(\omega_t,\omega_e)$.

It is common to treat the weights as
\begin{equation}
\label{eq:weight_normalization}
\omega_t + \omega_e = 1, \quad \omega_t,\omega_e > 0.
\end{equation}
In this case, $\omega_e = 1 - \omega_t$, and \eqref{eq:F_affine} becomes a function of a single scalar $\omega_t$:
\begin{equation}
\label{eq:F_omega_t}
F(\omega_t) = A - B\omega_t - C(1-\omega_t)
= A - C - (B - C)\omega_t.
\end{equation}
Since \eqref{eq:F_omega_t} is linear in $\omega_t$, the following properties hold for a fixed allocation:
\begin{enumerate}
    \item If $B > C$, then $F(\omega_t)$ decreases with $\omega_t$;
    \item If $B < C$, then $F(\omega_t)$ increases with $\omega_t$;
    \item If $B = C$, then $F(\omega_t)$ is constant in $\omega_t$.
\end{enumerate}
Therefore, the simplified objective \eqref{eq:F_omega_t} is an affine function of the delay weight $\omega_t$ and does not admit a unique interior optimum with respect to $\omega_t$. Instead, different values of $(\omega_t,\omega_e)$ correspond to different points on the delay–energy tradeoff curve for the given allocation.

\subsection{Jointly optimize decision variables and weight parameters}
In this case, the optimization problem would be different and more complex, which is beyond the scope of this work. In this paper, we instead follow the common practice of treating $(\omega_t,\omega_e)$ as design parameters that encode the operator’s delay–energy preference, and we focus on analyzing how different weight settings influence the resulting performance.

\subsection{Further discussion}
From a system design perspective, these observations imply that $(\omega_t,\omega_e)$ should be interpreted as policy parameters in this paper rather than as quantities that can be optimized in a purely mathematical sense. Larger values of $\omega_t$ place more emphasis on delay, which encourages solutions with lower latency at the expense of higher energy consumption, while smaller values of $\omega_t$ place more emphasis on energy saving and allow higher delay to reduce energy usage. In practice, the choice of $(\omega_t,\omega_e)$ should be guided by the operator's quality-of-service requirements and energy budget.

In summary, the weight parameters $(\omega_t,\omega_e)$ provide a flexible mechanism for network operators to select their preferred operating point along the delay–energy tradeoff curve. The optimization framework in Problem~$\mathbb{P}_3$ accommodates any such choice of $(\omega_t,\omega_e)$ and computes the corresponding resource allocation that maximizes the overall parameter training efficiency under the specified preference between delay and energy.

\section{Additional Simulation Results about Performances under Mobility-Aware SAGIN Networks}\label{sec.more_results_under_mobility}
\begin{figure}[t] 
\vspace{-0.3cm}
\subfigure[Total cost at each level.]{\includegraphics[width=.24\textwidth]{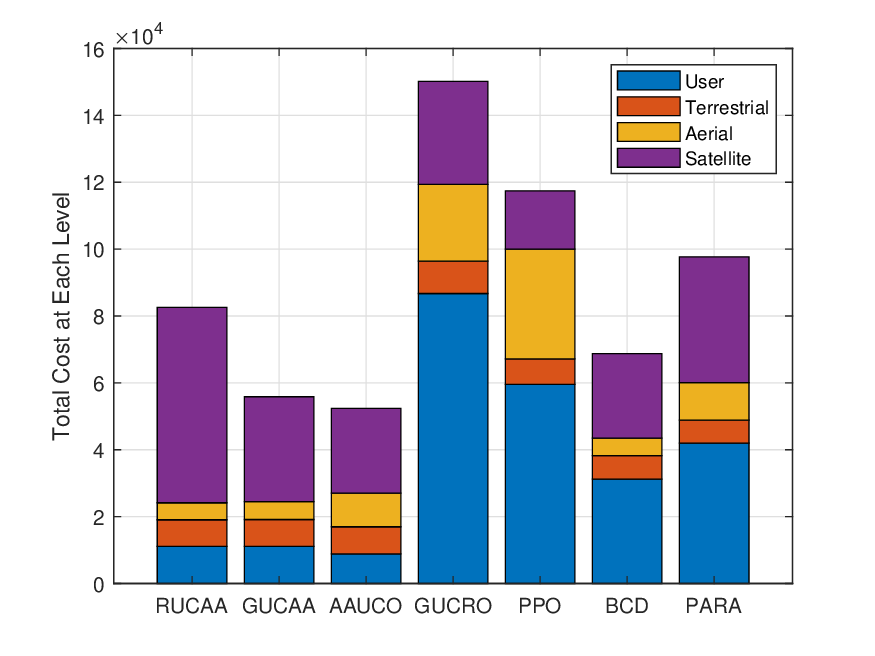}}
\subfigure[Numbers of training parameters at each level.]{\includegraphics[width=.24\textwidth]{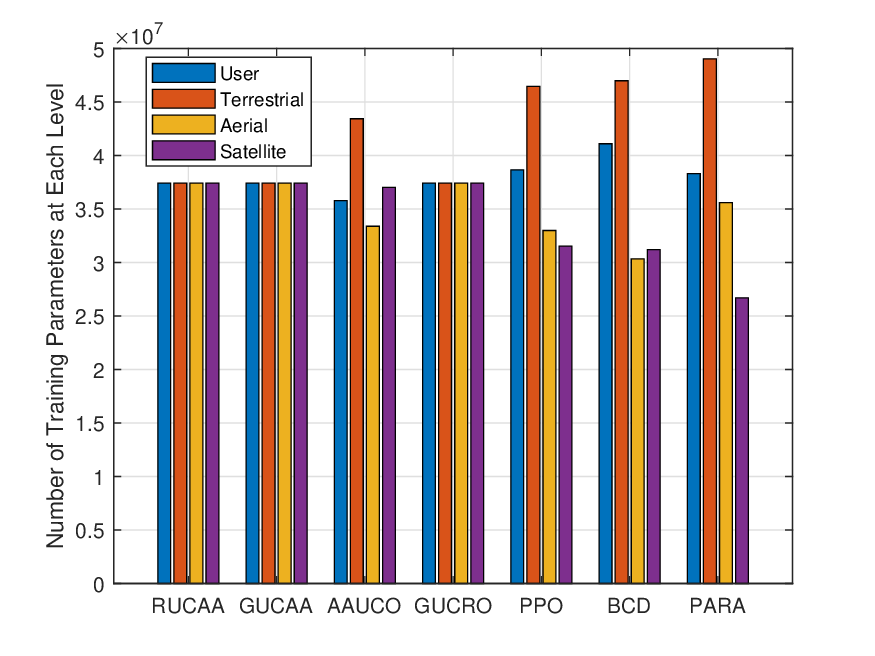}} \vspace{-10pt}
\caption{More performance comparisons under dynamic SAGIN topology with $\omega_t = 0.5, \omega_e = 0.5$.}
\label{fig.mobility_moreresults_t5e5}
\end{figure}

\begin{figure}[t] 
\vspace{-0.3cm}
\subfigure[$\omega_t = 0.1, \omega_e = 0.9$.]{\includegraphics[width=.24\textwidth]{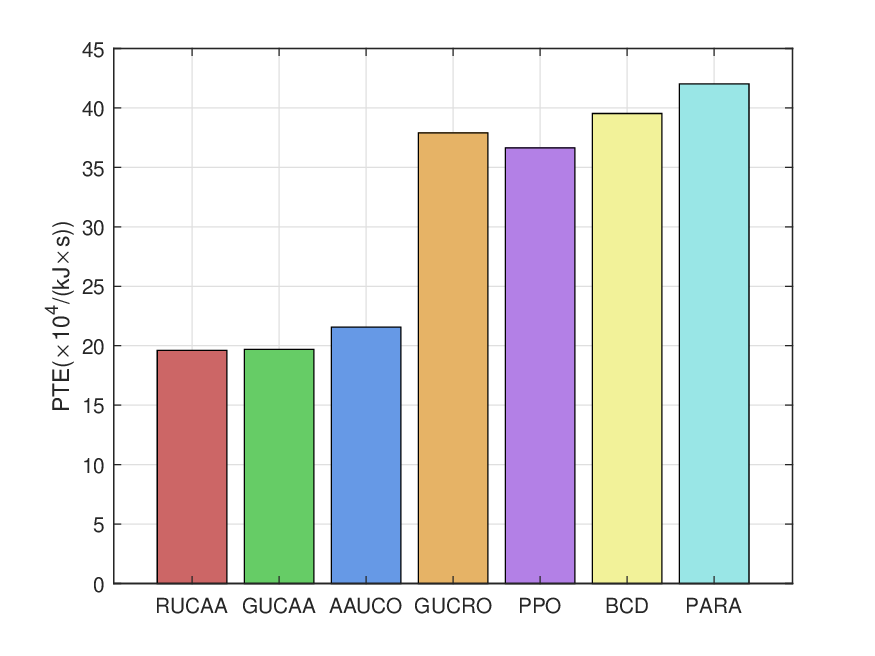}}
\subfigure[$\omega_t = 0.3, \omega_e = 0.7$.]{\includegraphics[width=.24\textwidth]{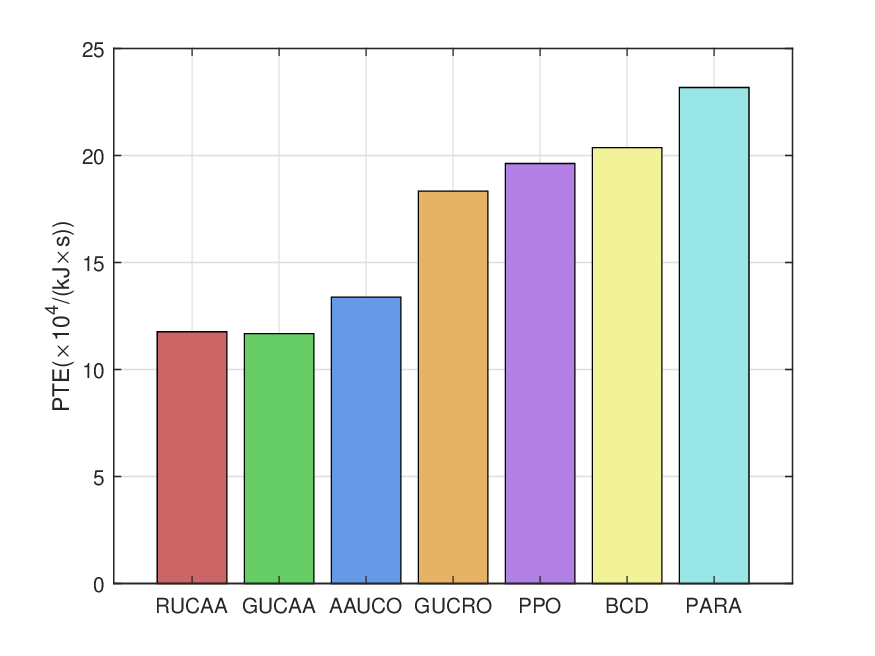}} 
\subfigure[$\omega_t = 0.7, \omega_e = 0.3$.]{\includegraphics[width=.24\textwidth]{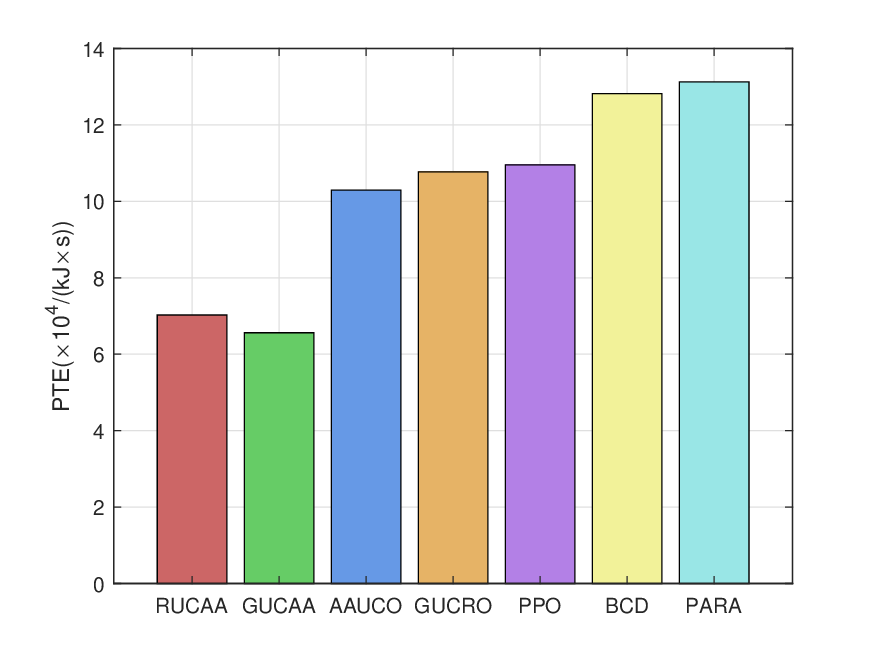}}
\subfigure[$\omega_t = 0.9, \omega_e = 0.1$.]{\includegraphics[width=.24\textwidth]{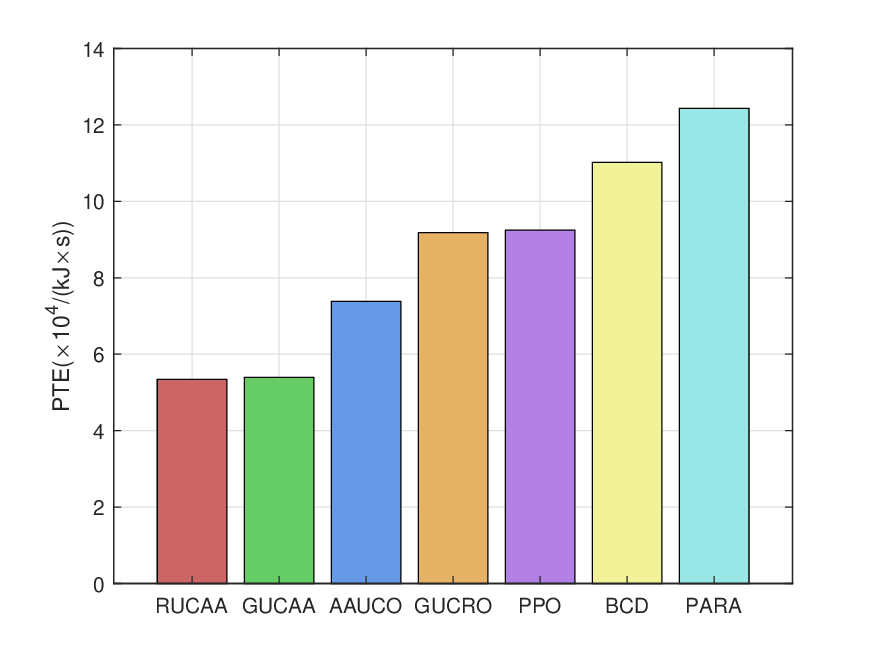}}\vspace{-10pt}
\caption{PTE performance comparisons under dynamic SAGIN topology with different $(\omega_t , \omega_e)$ settings.}
\label{fig.mobility_pte_addtional}
\end{figure}

\begin{figure}[t] 
\vspace{-0.3cm}
\subfigure[$\omega_t = 0.1, \omega_e = 0.9$.]{\includegraphics[width=.24\textwidth]{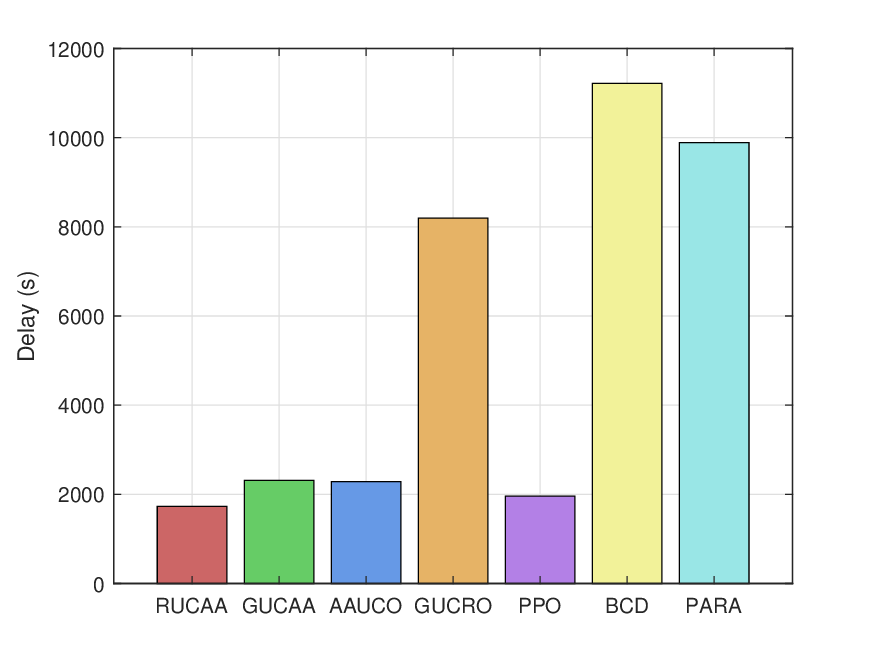}}
\subfigure[$\omega_t = 0.3, \omega_e = 0.7$.]{\includegraphics[width=.24\textwidth]{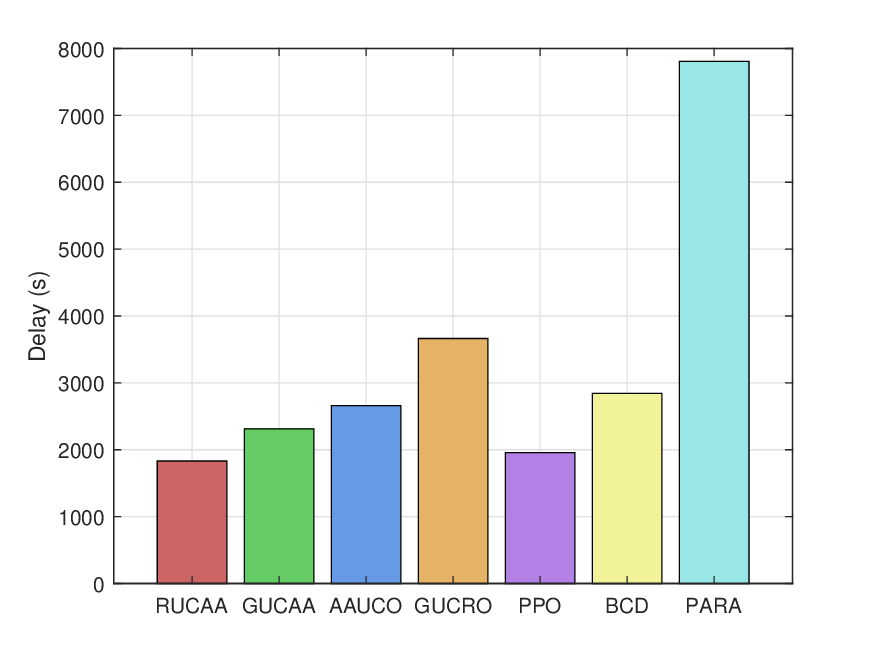}} 
\subfigure[$\omega_t = 0.7, \omega_e = 0.3$.]{\includegraphics[width=.24\textwidth]{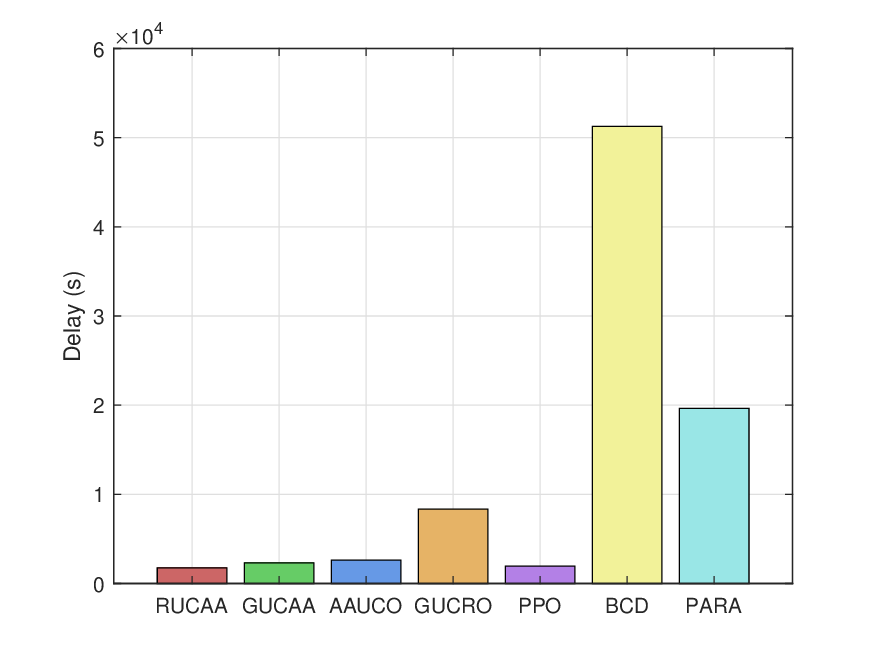}}
\subfigure[$\omega_t = 0.9, \omega_e = 0.1$.]{\includegraphics[width=.24\textwidth]{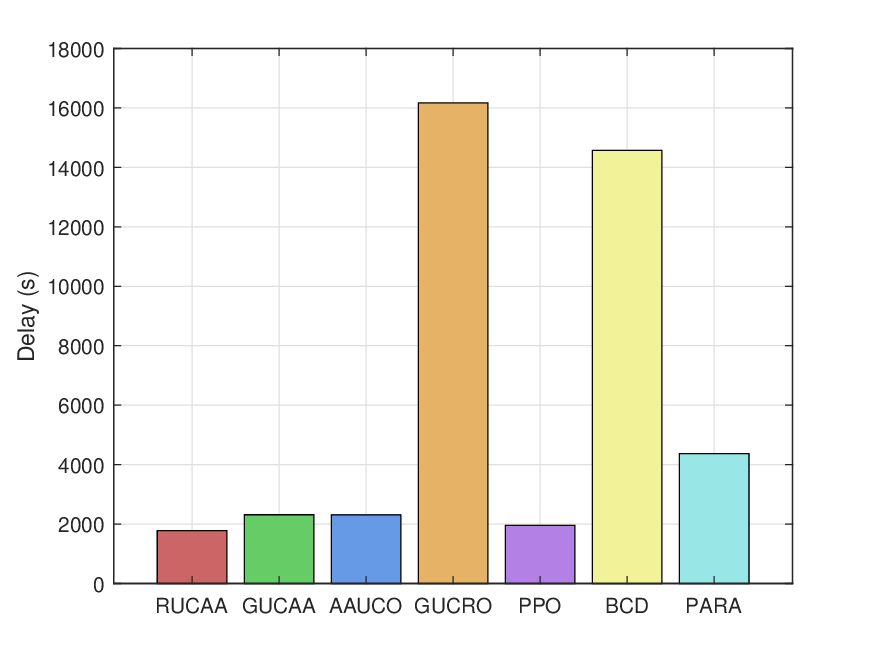}}\vspace{-10pt}
\caption{Delay performance comparisons under dynamic SAGIN topology with different $(\omega_t , \omega_e)$ settings.} 
\label{fig.mobility_delay_addtional}
\end{figure}

\begin{figure}[t]
\vspace{-0.3cm}
\subfigure[$\omega_t = 0.1, \omega_e = 0.9$.]{\includegraphics[width=.24\textwidth]{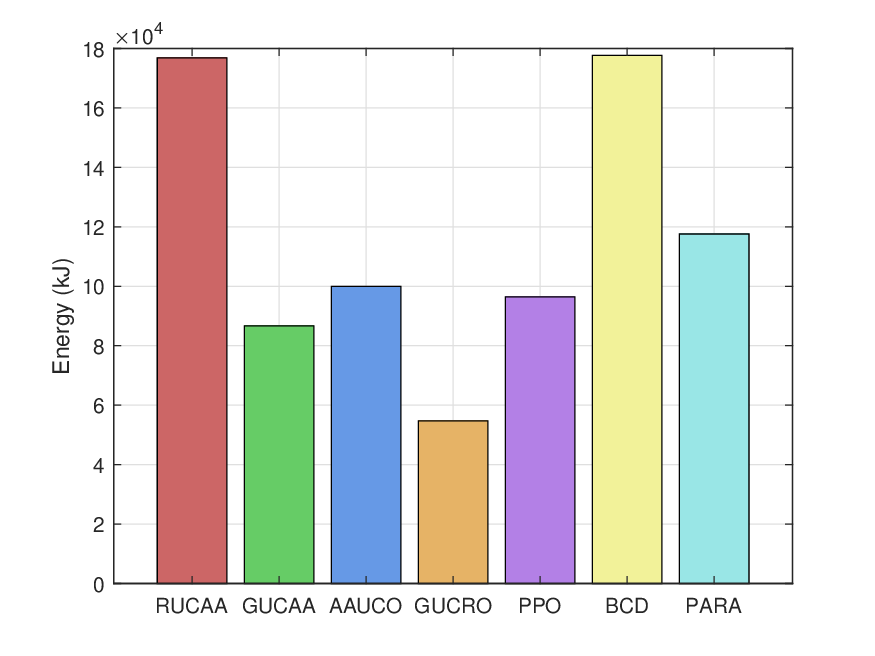}} 
\subfigure[$\omega_t = 0.3, \omega_e = 0.7$.]{\includegraphics[width=.24\textwidth]{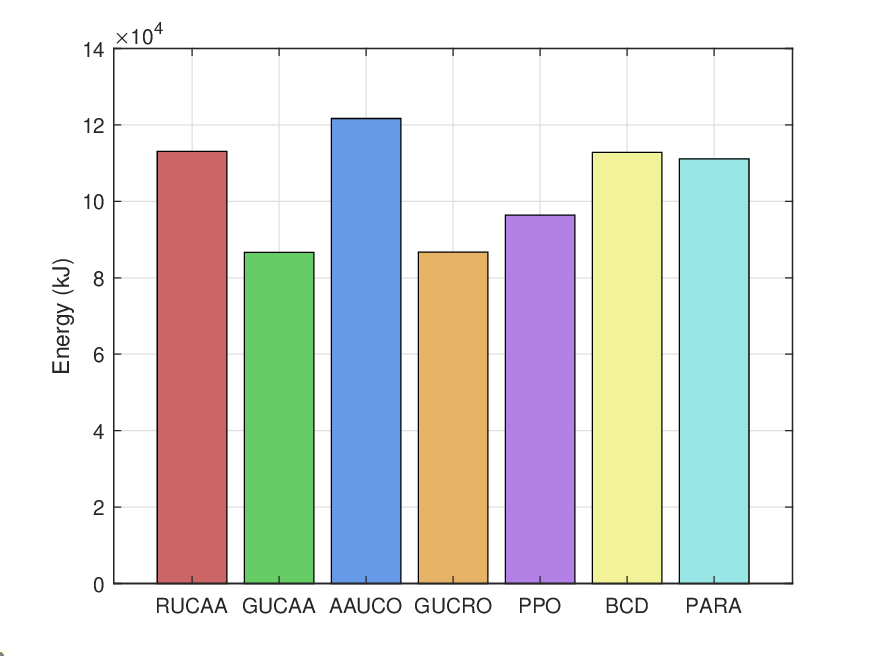}} 
\subfigure[$\omega_t = 0.7, \omega_e = 0.3$.]{\includegraphics[width=.24\textwidth]{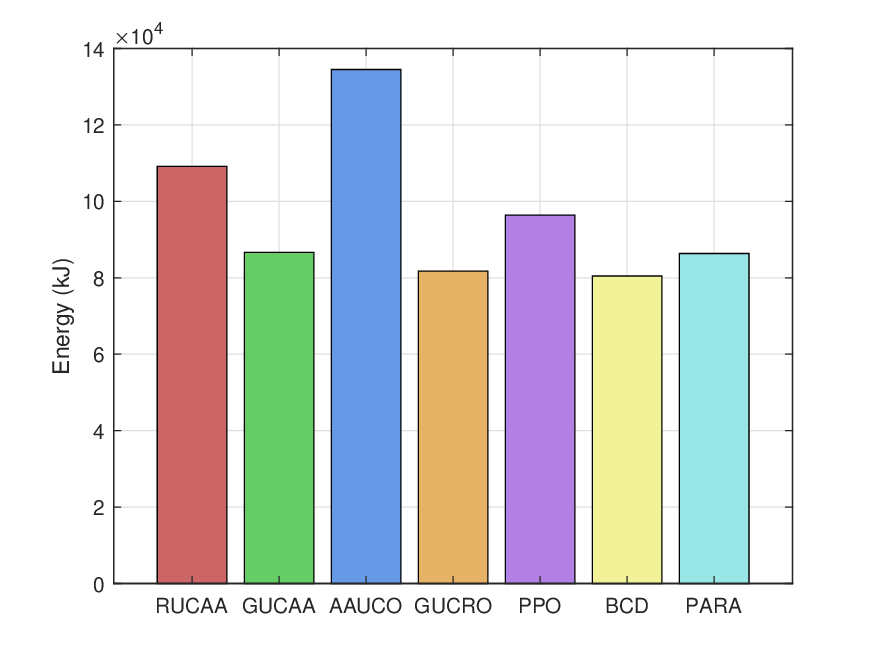}}
\subfigure[$\omega_t = 0.9, \omega_e = 0.1$.]{\includegraphics[width=.24\textwidth]{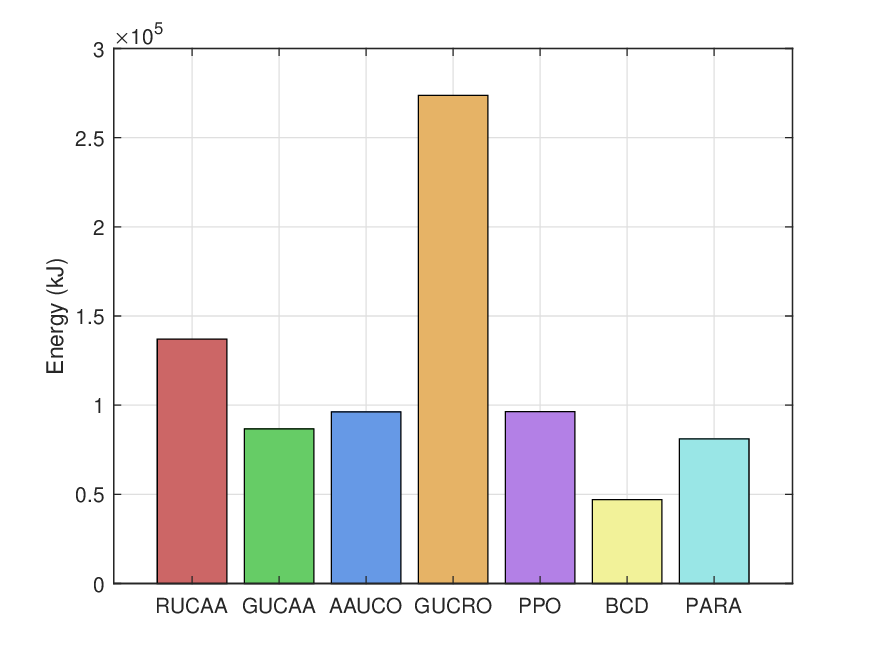}} \vspace{-10pt}
\caption{Energy consumption performance comparisons under dynamic SAGIN topology with different $(\omega_t , \omega_e)$ settings.} 
\label{fig.mobility_energy_addtional}
\end{figure}

\begin{figure}[t]
\vspace{-0.3cm}
\subfigure[$\omega_t = 0.1, \omega_e = 0.9$.]{\includegraphics[width=.24\textwidth]{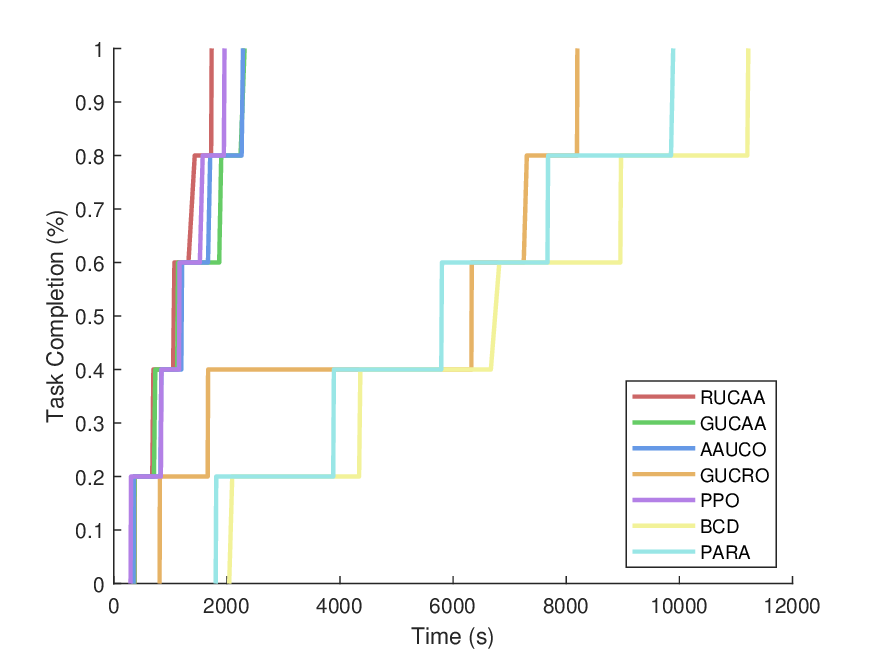}} 
\subfigure[$\omega_t = 0.3, \omega_e = 0.7$.]{\includegraphics[width=.24\textwidth]{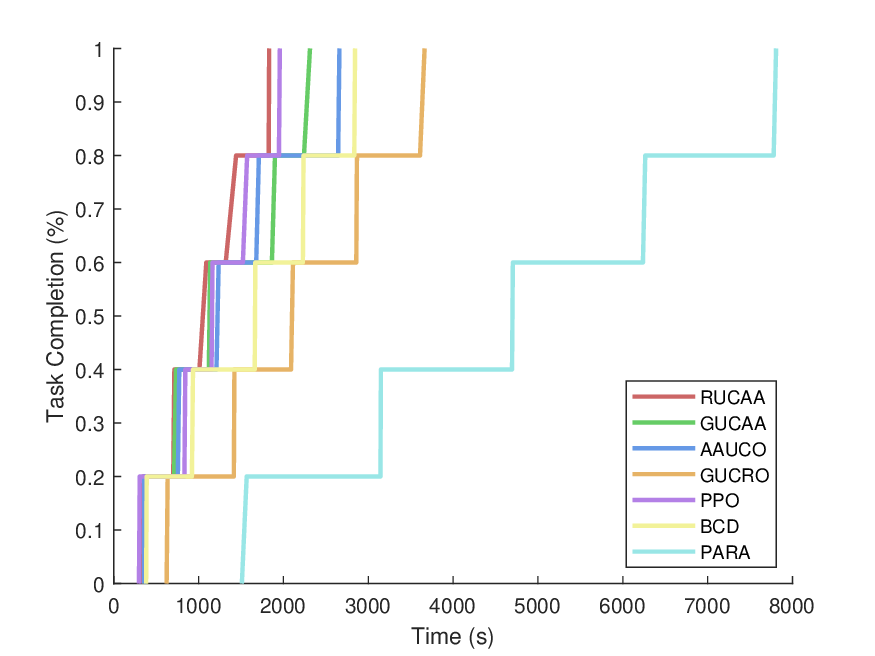}} 
\subfigure[$\omega_t = 0.7, \omega_e = 0.3$.]{\includegraphics[width=.24\textwidth]{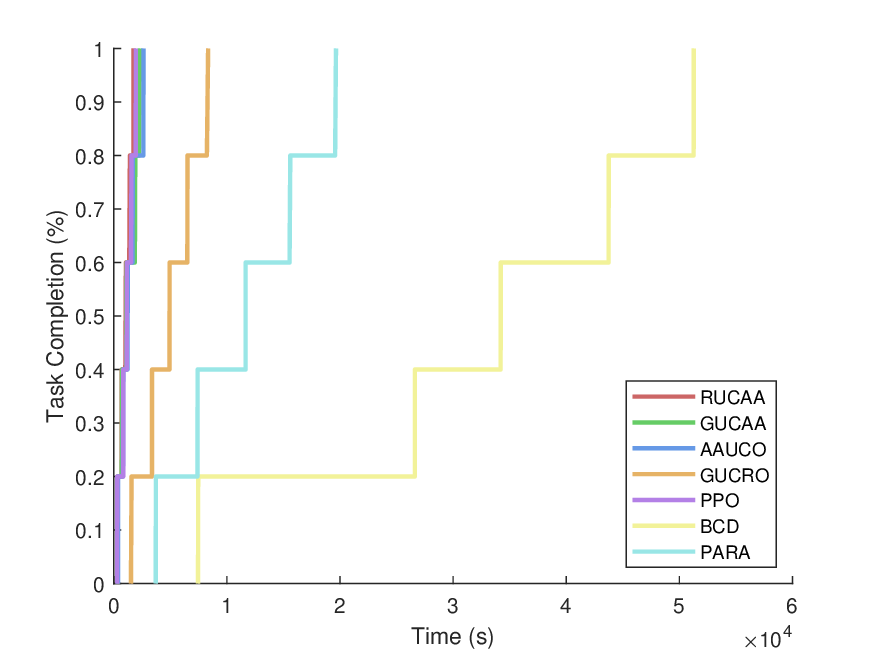}}
\subfigure[$\omega_t = 0.9, \omega_e = 0.1$.]{\includegraphics[width=.24\textwidth]{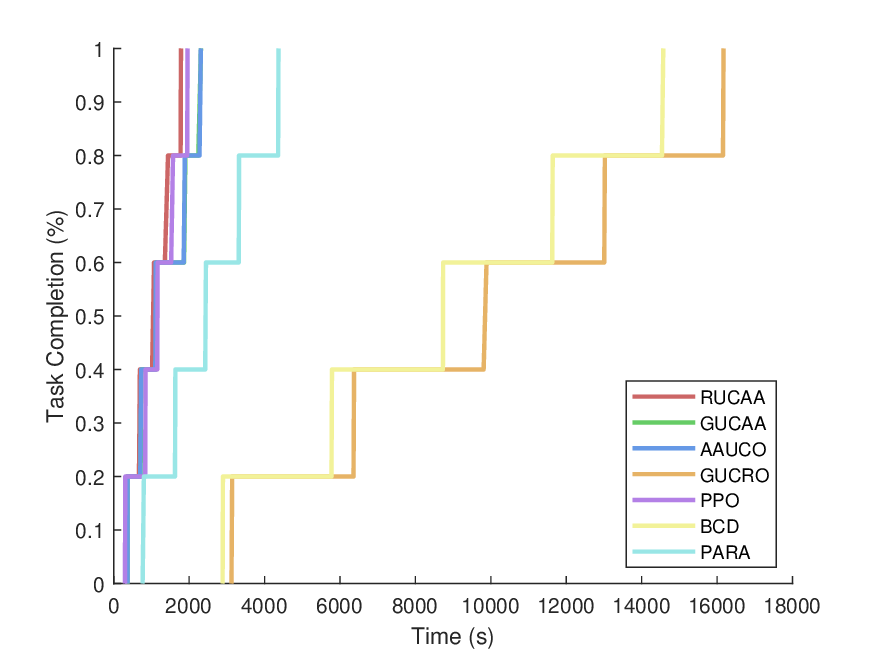}} \vspace{-10pt}
\caption{Task complete ratio performance comparisons under dynamic SAGIN topology with different $(\omega_t , \omega_e)$ settings.} 
\label{fig.mobility_taskcompleteratio_addtional}
\end{figure}

\begin{figure}[t]
\vspace{-0.3cm}
\subfigure[$\omega_t = 0.1, \omega_e = 0.9$.]{\includegraphics[width=.24\textwidth]{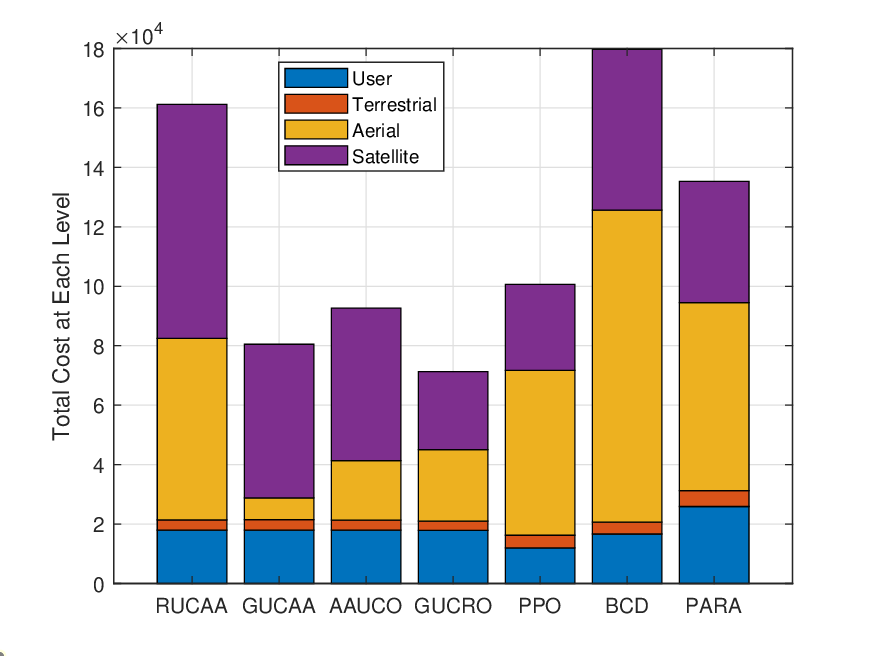}} 
\subfigure[$\omega_t = 0.3, \omega_e = 0.7$.]{\includegraphics[width=.24\textwidth]{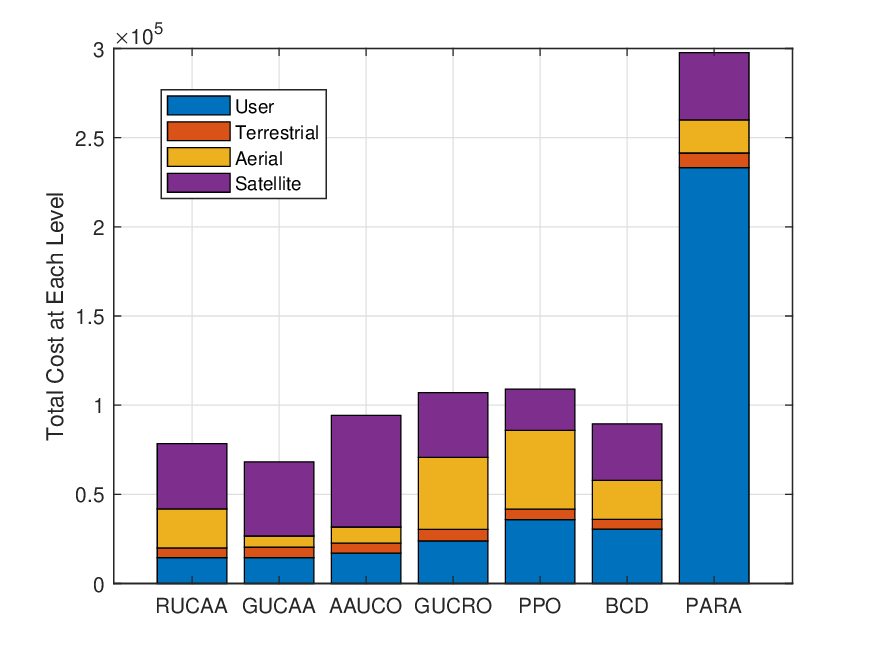}} 
\subfigure[$\omega_t = 0.7, \omega_e = 0.3$.]{\includegraphics[width=.24\textwidth]{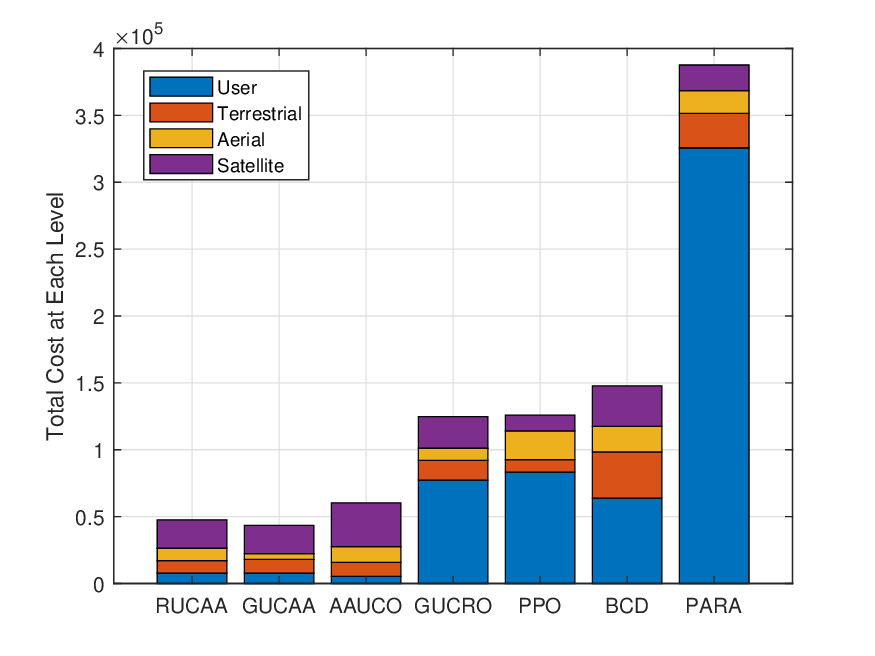}}
\subfigure[$\omega_t = 0.9, \omega_e = 0.1$.]{\includegraphics[width=.24\textwidth]{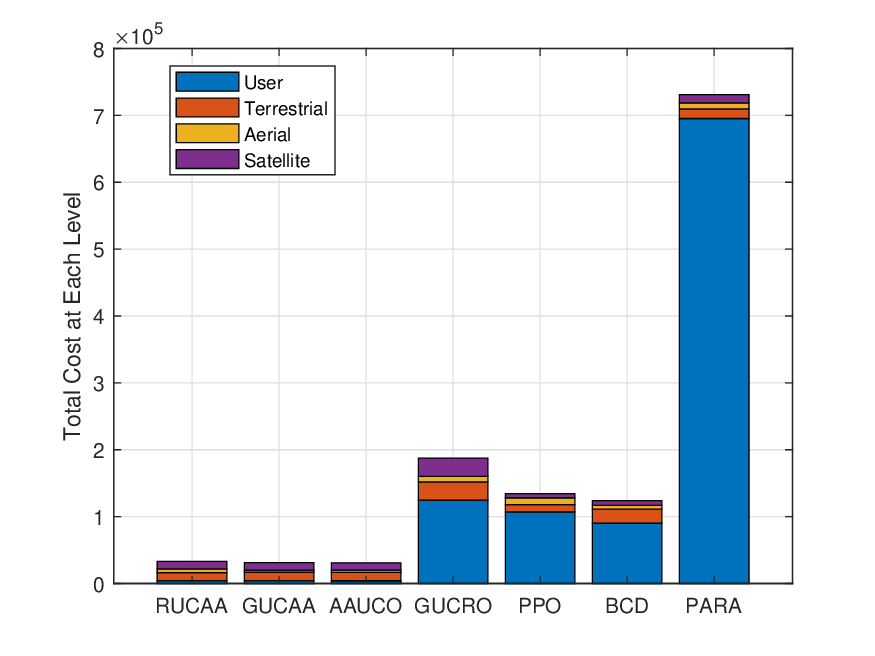}} \vspace{-10pt}
\caption{Total cost comparisons at each level under dynamic SAGIN topology with different $(\omega_t , \omega_e)$ settings.}
\label{fig.mobility_cost_eachlevel_addtional}
\end{figure}

\begin{figure}[t]
\vspace{-0.3cm}
\subfigure[$\omega_t = 0.1, \omega_e = 0.9$.]{\includegraphics[width=.24\textwidth]{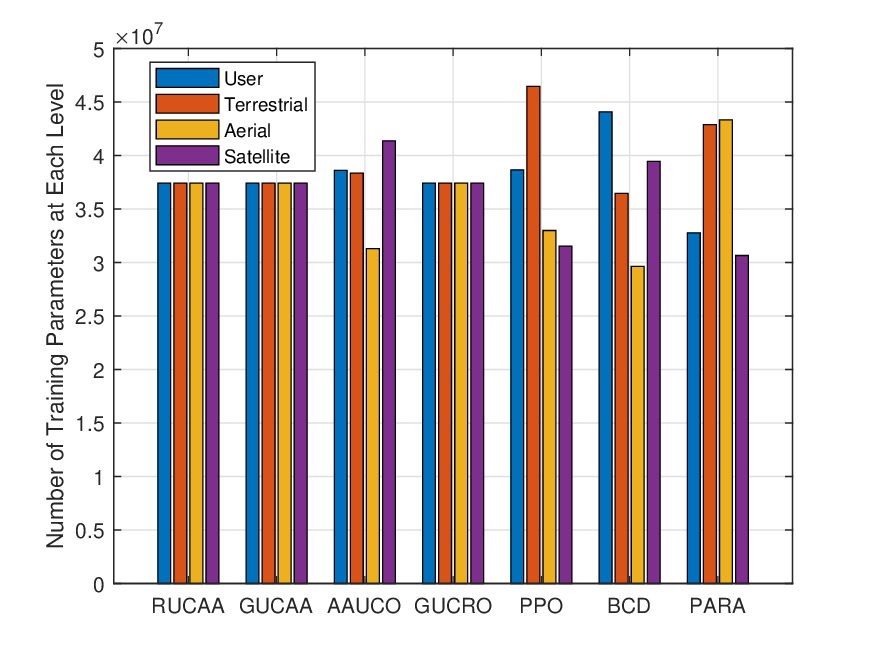}} 
\subfigure[$\omega_t = 0.3, \omega_e = 0.7$.]{\includegraphics[width=.24\textwidth]{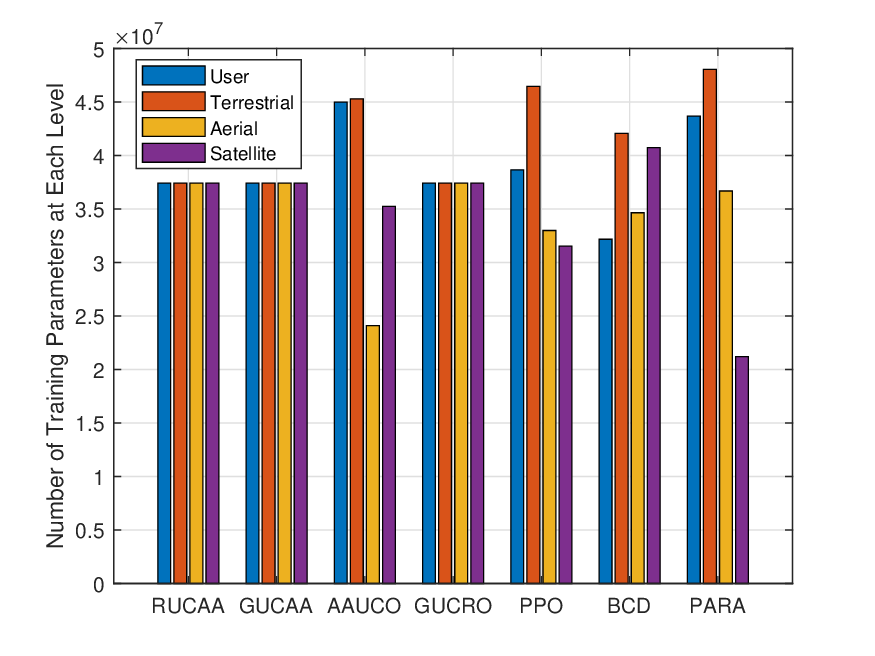}} 
\subfigure[$\omega_t = 0.7, \omega_e = 0.3$.]{\includegraphics[width=.24\textwidth]{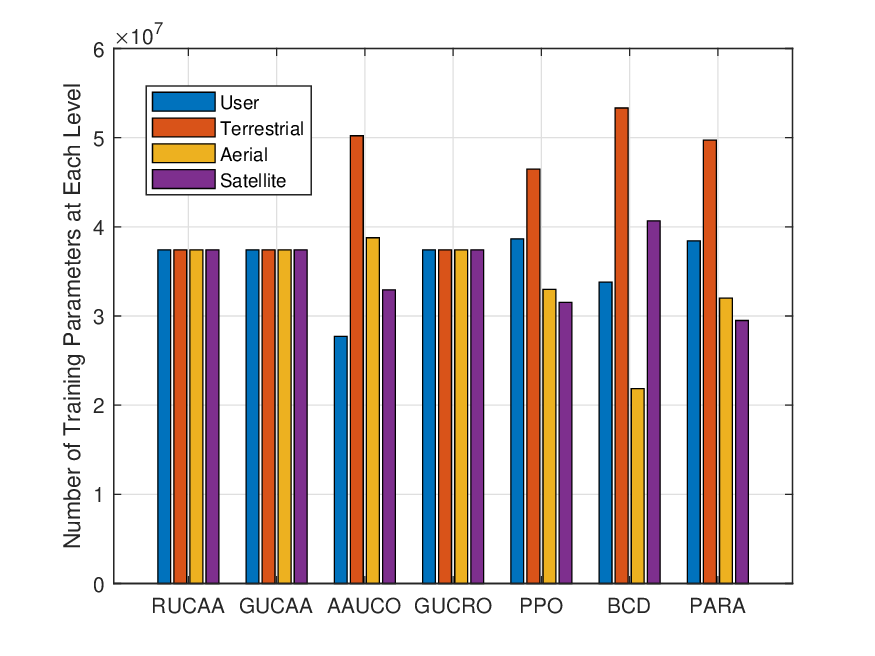}}
\subfigure[$\omega_t = 0.9, \omega_e = 0.1$.]{\includegraphics[width=.24\textwidth]{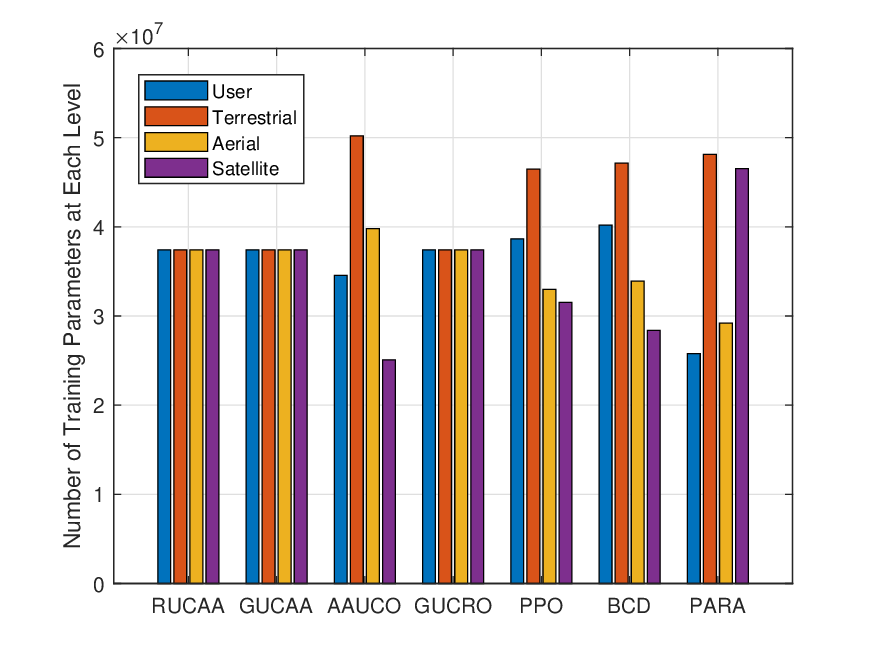}} \vspace{-10pt}
\caption{Numbers of training parameters at each level under dynamic SAGIN topology with different $(\omega_t , \omega_e)$ settings.} 
\label{fig.mobility_train_paranum_addtional}
\end{figure}

In this section, additional simulation results comparing the performance of methods under mobility-aware SAGIN networks are presented in Figs. \ref{fig.mobility_moreresults_t5e5}-\ref{fig.mobility_train_paranum_addtional}.

In Fig. \ref{fig.mobility_moreresults_t5e5}, we provide additional per-level results to compare the methods under a dynamic SAGIN topology with $\omega_t = 0.5$ and $\omega_e = 0.5$. The PARA method achieves a moderate cost at each level compared with the other methods, while allocating more trainable parameters to terrestrial servers and fewer to satellite servers. Next, we further examine the impact of the weight parameters $(\omega_t, \omega_e)$ on PTE, delay, energy consumption, task complete ratios, total cost at each level, and numbers of training parameters at each level for four representative settings, $(0.1, 0.9), (0.3, 0.7)$, $(0.7, 0.3)$, $(0.9, 0.1)$.

In Fig. \ref{fig.mobility_pte_addtional}, for all methods, the PTE value decreases as $\omega_t$ increases, which is consistent with the affine analysis in Appendix G and indicates that, after optimization, the effective delay term multiplied by $\omega_t$ remains larger than the effective energy term multiplied by $\omega_e$. In other words, giving more weight to delay inevitably reduces PTE in our setting. Across all weight pairs, PARA consistently achieves the highest PTE.

Based on the results in Fig. \ref{fig.mobility_delay_addtional} to Fig. \ref{fig.mobility_train_paranum_addtional}, it is not straightforward to attribute the PTE gain of PARA to a single dominant factor, since PTE is a sum-of-ratios metric jointly influenced by user–server association, offloading decisions of training parameters, system delay, energy consumption, and waiting time when aerial or satellite servers are unavailable. Consequently, PARA does not necessarily achieve the best value on any individual sub-metric (e.g., delay, energy, or the number of trained parameters at each layer). Instead, by directly optimizing the PTE objective, PARA finds a balanced operating point that increases the amount of trained parameters while keeping the overall training cost (delay, energy, and waiting time) under control, which leads to the highest overall PTE among all methods.

In such a coupled and non-convex problem, the optimization may move between different stationary points as $\omega_t$ varies, and not every individual metric needs to change monotonically. The proposed PARA algorithm is designed exactly for this regime: by relaxing the original problem into a sequence of convex subproblems, PARA is derived from a fractional-programming and semidefinite programming framework and systematically optimizes a convex surrogate at each iteration, whereas the baseline methods are not specifically tailored to this PTE-oriented fractional formulation. These results confirm that PARA robustly exploits the delay–energy–training parameters trade-off and maintains the highest PTE across a wide range of reasonable weight settings, even though the detailed evolution of delay and energy with $\omega_t$ is not strictly monotonic. Some post-deployment tuning methods may achieve a closer correlation between PTE and delay or energy consumption, which is beyond the scope of this work.

\end{appendices}

\end{document}